\documentclass[11pt]{amsart}
\usepackage{goodstyle}

\title[Extremal event graphs] {Extremal event graphs: a (stable) tool for analyzing noisy time series data} 

\author[R. Belton and B. Cummins and B. T. Fasy and T. Gedeon]{}


 \email{robin.belton@montana.edu}
 \email{breschine.cummins@montana.edu}
 \email{brittany.fasy@montana.edu}
 \email{gedeon@math.montana.edu}

\thanks{This material is based upon work supported by the US National Science Foundation under grant No.~DGE 1649608 (Belton) and DMS 1839299 (Cummins and Gedeon), CCF 2046730 (Fasy), as well as the National Institute of Health under grant No. 5R01GM126555-01 (Gedeon).}

\begin{document}

\centerline{\scshape Robin Belton, Bree Cummins, and Tom\'{a}\v{s} Gedeon}
\medskip
{\footnotesize
 \centerline{Department of Mathematical Sciences}
   \centerline{Montana State University}
   \centerline{Bozeman, MT, 59717, USA}
}

\medskip

\centerline{\scshape Brittany Terese Fasy}
\medskip
{\footnotesize
 \centerline{Department of Mathematical Sciences and School of Computing}
   \centerline{Montana State University}
   \centerline{Bozeman, MT, 59717, USA}
}

\bigskip

\begin{abstract}
Local maxima and minima, or \textit{extremal events}, in experimental time series can be used as a coarse summary to
characterize data. However, the discrete sampling in recording experimental measurements suggests
uncertainty on the true timing of extrema during the experiment. This in turn gives uncertainty in the 
timing order of extrema within the time series. Motivated by
applications in genomic time series and biological network analysis, we
construct a weighted directed acyclic graph (DAG) called an \textit{extremal event
DAG} using techniques from persistent homology that is robust to measurement
noise. Furthermore, we define a distance between extremal event DAGs based on
the edit distance between strings. We prove several properties including local stability for the
extremal event DAG distance with respect to pairwise $L_{\infty}$ distances between
functions in the time series data. Lastly, we provide algorithms, publicly free software, and
implementations on extremal event DAG construction and comparison. 
\end{abstract}

\maketitle

\section{Introduction}\label{sec:intro}

Experimental time series data are ubiquitous in today's science and provide a
window through which we can observe the underlying dynamics of complex systems,
ranging from cells to ecosystems and climate. We study collections of time series which 
are also referred to as multivariate time series in the literature \cite{TsayMultivariate13, WeiMultivariate19}.
We construct a weighted directed
graph descriptor of a collection of time series data using \textit{persistent
homology}, a technique that belongs to a collection of approaches known as
Topological Data Analysis (TDA) that uses algebraic topology
\cite{HatcherAlgebraic02, MunkresElements84} to extract shape from data. TDA is used to study data from a wide range of applications including material science \cite{YongjinHigh18}, cancer biology \cite{VipondMultiparameter21}, and political science \cite{FengPersistent21}. Some of the classic and foundational papers in TDA include \cite{ZomorodianComputing05, KaczynskiComputational04, GhristBarcodes08, EdelsbrunnerPersistent08, OtterA17}. 

Our descriptor characterizes a  collection of time series by the order of their extrema in a way that also captures the robustness of this order with respect to measurement uncertainty.   Our motivation comes from the desire to mathematically capture and compare collections of `omics time series data, such as transcriptomics, proteomics, and others. In particular, the coarse information of orders of extrema have been used to assess \textit{regulatory network models} of gene/protein interactions~\cite{CumminsModel18}. Other applications involve quantifying similarity between gene expression time series \cite{BerryUsing20, SmithAn20} across repeated experiments.

Our mathematical methods are motivated by the combinatorial approaches in~\cite{CumminsModel18,BerryUsing20,NeremA19} that use only the approximate timing of time series extrema as the relevant features of experimental data. To take into account the uncertainty of capturing temporal orderings of extrema,
\cite{CumminsModel18} replaced the single time point locations of extrema with time intervals that were determined by manual inspection.
If the intervals are disjoint, then the ordering of extrema is interpreted to be robust to measurement uncertainty. In the follow-up paper~\cite{BerryUsing20}, an approach was developed in which the intervals are algorithmically
constructed using \textit{merge trees} \cite{EdelsbrunnerComputational10, MorozovInterleaving13}, branch decompositions
\cite{MorozovDistributed13}, and sublevel sets. These intervals are called \textit{$\varepsilon$-extremal intervals} and have the property that they are the smallest intervals for which all continuous perturbations of a continuous function (with additional technical restrictions) that lie within an $\varepsilon$-band are guaranteed to attain an extremum under measurement uncertainty of size $\varepsilon$. 
Using the $\varepsilon$-extremal intervals, labeled directed acyclic graphs (DAGs) are
constructed to represent the time series data for any fixed value of $\varepsilon$.
We refer to these DAGs as $\varepsilon$-\textit{DAGs}. Vertices or nodes in the
$\varepsilon$-DAG represent extrema in the time series data. Directed edges
$a\to b$ indicate that we can unambiguously discern the order (in time)  of events
corresponding to vertices $a$ and $b$ under measurement uncertainty of size $\varepsilon$.

Continuing this line of research, \cite{NeremA19} defined a distance metric that compares 
two collections of time series by comparing the corresponding
$\varepsilon$-DAGs. This metric involves computing the directed maximal
common edge subgraph (DMCES) and was applied in \cite{BerryUsing20} to
quantify similarity in replicate experiments of microarray yeast
cell cycle data. Additionally, the metric was used in
\cite{SmithAn20} to provide quantitative evidence that an intrinsic oscillator
drives the blood stage cycle of the malaria parasite \textit{Plasmodium
falciparum}. The metric for $\varepsilon$-DAGs using the DMCES is effective in
capturing similarity between the time series, but is computationally expensive.
This limits the total number of extrema across all time series that can be
effectively analyzed. Another limitation is that the distance can only be measured at a single measurement uncertainty level $\varepsilon$, which is often unknown and thus the distance has to be computed multiple times for a collection of $\varepsilon$ values. A better measurement would incorporate information about changes in similarity as a function of changing $\varepsilon$ in a single value.

We significantly expand and generalize the work of \cite{BerryUsing20} and \cite{NeremA19} by
constructing a \textit{weighted} DAG that reflects robustness of the extremal ordering for all 
levels of measurement error $\varepsilon$. We call this an \textit{extremal event DAG}. 
Vertices in this graph again represent extrema in the time series data, and a directed edge 
$a\to b$ indicates that extremum $a$ occurs before extremum $b$. The node weights measure 
prominence of extrema while the edge weights indicate the
smallest $\varepsilon$ level for which the relative order between the two
associated extrema can no longer be guaranteed. The node weights are computed using
\textit{sublevel set persistence}~\cite{EdelsbrunnerComputational10}.  After
representing the collection of time series as an extremal event DAG, we define a
distance between extremal event DAGs as a modified version of the \textit{edit
distance} (Chapter 15 of \cite{CormenIntroduction09}). The edit distance
quantifies similarity of  two strings based on the minimum number of operations
(e.g., insertion, substitution, and deletion) it takes to align the two strings.
This distance is commonly used in many applications including DNA sequence
alignment, see \cite{NeedlemanA70} for one of the first papers on the topic. The standard algorithm for the
edit distance between two strings of length $n$ can be computed via dynamic
programming in $\Theta(n^2)$ (Chapter 15 of \cite{CormenIntroduction09}). 

We prove several key properties of the extremal event DAG weights. Most importantly for computability 
and applications, \thmref{edge-weights} gives a simple criterion to
compute edge weights.
Furthermore, we analyze stability properties of distances between extremal event DAGs.
 \secref{stability} gives stability results
for distances used to compare extremal event DAGs with
respect to pairwise $L_{\infty}$ distance between the underlying continuous
functions. These stability results show that small changes within
time series data lead to small changes in the corresponding extremal event DAG
distances. In \thmref{extremal-DAG-stability}, we show the extremal event DAG
distance is stable in a local case: two paired continuous functions from the two collections of time series must lie 
within an $\varepsilon$ band that allows for an unambiguous alignment of the minima and maxima
between the two time series. Additionally, one of the time series can have small
amplitude additional maxima and minima.

Furthermore, extremal event DAGs and the extremal event DAG distance is computable. The general pipeline of utilizing extremal event DAGs takes two datasets of collections of time series as input, computes the associated extremal event DAGs, and the extremal event DAG distance between the extremal event DAGs, see \figref{pipeline}. We illustrate two biological applications of extremal event DAGs and distances in \secref{applications}. Additionally, the polynomial time algorithm to compute the extremal event DAG given collections of discrete time series, and the dynamic program needed to compute the extremal event DAG distance is provided in the supplementary materials. Lastly, free and public software on computing extremal event DAGs and the distance between these descriptors is available at \cite{BeltonComputing21}.

    \begin{figure}[htp]
        \centering
        {\includegraphics[width=1\textwidth]{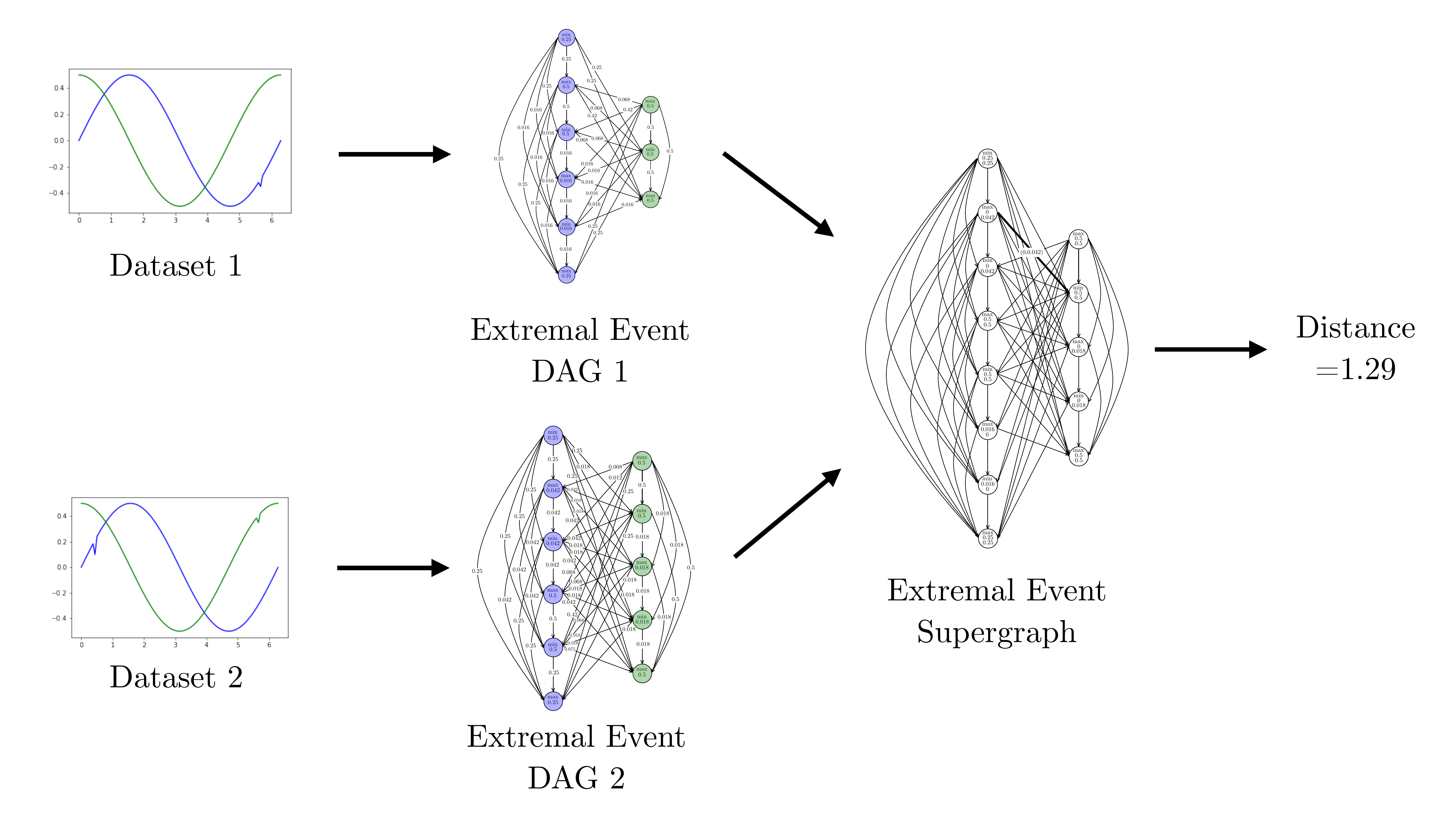}}
        \caption{Schematic for Comparing Collections of Time Series Using Extremal Event DAGs.
Each dataset consists of collections of time series. Extremal event DAGs are computed for both datasets. Each vertex corresponds to a local extremum of one of the time series in the dataset. Vertices highlighted in blue correspond to local extrema in the blue time series while vertices highlighted in green correspond to local extrema in the green time series. Vertices and directed edges capture the order of local extrema, and weights signify the robustness of extremal orderings. To compute a distance between extremal event DAGs, an extremal event supergraph is computed using a modified version of the edit distance. The extremal event supergraph has doubly weighted vertices and edges. The weight in the first component is associated to extremal event DAG 1 and the weight in the second component is associated to extremal event DAG 2. Datasets 1 and 2 differ by only a few perturbations and we see the resulting extremal event DAG distance is equal to 1.29.}
        \label{fig:pipeline}
    \end{figure}

\subsection{TDA in Analyzing Time Series Data}
TDA has been used to study time series data and the dynamics of underlying biological networks. One common method studies single variable time series by using Takens' embedding theorem \cite{TakensDetecting80}. The data is transformed by computing a sliding window embedding of the data into a point cloud in $\R^n$. This point cloud is then analyzed using one dimensional persistent homology to detect and quantify periodicity. Examples that take this or a modified approach include detecting periodicity of genomic time series data \cite{PereaSliding15}, studying entropy and the dynamics of the underlying system \cite{SmallComplex13, BandtPermutation02},  characterizing gene regulatory networks \cite{BerwaldCritical14}, and distinguishing between audio signals of the same note from different instruments \cite{SandersonComputational17}. The goal in these approaches is to quantify periodicity of single variable time series whereas we are interested in capturing order of local extrema in multivariate time series. 

The second common method to study single variable time series using TDA is to apply sublevel set persistence to detect prominent features. Applications include signal processing \cite{KhasawnehTopological18}, Fourier spectrum analysis and parameter detection \cite{KhasawnehOn20}, arrythmia detection \cite{MeryllTopological20}, and cancer studies \cite{LawsonPersistent19}. Additionally, in \cite{ChenA19}, sublevel set persistence is used to define a topological
regularizer that can then be used as a classifier for machine learning. Furthermore, sublevel set persistence on time series can be used to determine noise that is often seen as small peaks in the time series \cite{AudunSeparating20}. We also apply sublevel set persistence to detect prominent features, however, we also explore its connections to capturing order and robustness of local extrema across multiple time series. 

Lastly, TDA has been used to study other types of time dependent data. One includes dynamic metric spaces that can be used to describe phenomena such as bird flocks, insect swarms, schools of fish, and aphid trajectories. TDA techniques to study these types of data include vineyards \cite{SteinerVines06}, CROCKER plots \cite{TopazTopological15, UlmerA19, XianCapturing21}, spatiotemporal filtrations \cite{KimSpatiotemporal21} and zig-zag persistence \cite{CarlssonZigzag10, KimStable21}. Furthermore, time series data of fMRI images have been used to construct functional networks, and then applying filtrations on the weights of these networks to extract topological features \cite{StolzPersistent17, PetriHomological14}. A more extensive summary on TDA techniques to study time series can be found in \cite{GholizadehA18}. 

\subsection{Biological Motivation}\label{sec:bio-motivation}
Extremal event DAGs are developed abstractly for mulitvariate time series in general, however, we focus on the application of analyzing `omics data that measures expression levels of thousands of genes. Transcription of genes produces messenger RNA (mRNA) which are translated to proteins. Gene expression, measured by either the amount of mRNA produced (transcriptomics) or by the amount of corresponding protein (proteomics), can be used to measure the level of activity of a given gene product. There is strong evidence that the relative phases of oscillating regulators are important to controlling important cellular processes  such as the cell cycle \cite{SimmonsTranscription08}, circadian rhythm, or malaria parasite periodic infection of human blood cells. The assertion of~\cite{BerryUsing20,CumminsModel18} is that the ordering of extrema is a reasonable approximation of control by phase relationship.

For example, it is hypothesized that a small transcriptional regulatory network controlling the cell cycle can activate hundreds of other transcription factors in a phase-specific manner to play a vital role in maintaining the proper progression of DNA replication and cell division~\cite{BristowCheckpoints14, ChoReconciling17, OrlandoGlobal08, SimmonsCyclin12, SimonSerial01}. There are still many open questions about the precise role of the transcriptional network in the ordering of cell cycle events \cite{RahiThe16, SheddenAnalysis02}. A reproducible ordering of gene expression such as what was observed in \cite{BerryUsing20} provides supporting evidence for the central role of the cell cycle gene regulatory network in orchestrating timely expression of other cell cycle events.

A question of interest in biology is evaluating the similarity of two experiments across labs or experimental conditions. For example, an experimentalist may wish to measure
the similarity of expression level of genes driving the cell cycle between replicate experiments between time series collected under different growth conditions, or across organisms and tissues.
For example, circadian clock networks in different tissues that control the temporal ordering of phase specific gene expression \cite{MureDiurnal18, ZhangA14}. Similarity and differences in timing of the same network in tissues like heart and liver can tell us about their mutual coupling as well as coupling to the master circadian clock in the brain \cite{MaThe15, RubenA18}.
In summary, mathematically modeling and comparing orders of extremal events in `omics data is useful for identifying time series differences in multiple biological applications. In particular, extremal event DAGs and distances can be used to study the important biological questions about time dependent cellular processes, some of which we have mentioned here.

\section{Background}\label{sec:prelims}

We now summarize necessary terminology for the results that have been developed in this paper. Many of these terms build off of ideas mentioned in \secref{intro}. Throughout this section and the rest of this paper, we use the following notation. Let $X \subset \R$ denote an arbitrary subset of $\R$. Let $C:=[a_1,a_2] \subset \R$ be a closed interval of $\R$. 

\subsection{Extrema}

For a subset $X\subset \R$, $x \in X$ and $\varepsilon>0$, let
$B_{\varepsilon}(x)$ be the open neighborhood of radius $\varepsilon$ centered
at $x$. That is
\[ B_{\varepsilon}(x) := \{ y \in X \mid |y-x| < \varepsilon\}. \]

\begin{defn}[Local Extrema]\label{def: local-extrema}
    Let $f: X \rightarrow \R$ be a function. We say $f$ has a \emph{local
    minimum} at $x \in X$ if there exists $\varepsilon>0$ for which $f(x)< f(y)$
    for all $y \in B_{\varepsilon}(x)\setminus\{x\}$. Similarly, $f$ has a
    \emph{local maximum} at $x \in X$ if there exists $\varepsilon>0$ for which
    $f(x)>f(y)$ for all $y \in B_{\varepsilon}(x)\setminus\{x\}$. We refer to
    any local minimum or local maximum as a \emph{local extremum} of $f$.
    If $x \in X$ with $f(x) < f(y)$ for
    all $y \in X\setminus \{x\}$, we say $f$ has a \emph{global minimum} at $x$.
    Similarly, for $x \in X$ where $f(x) > f(y)$ for all $y \in X\setminus \{x\}$,
    we say $f$ has a \emph{global maximum} at $x$.
\end{defn}

We often order the extrema of a function. To ease notation, we write $[n]$ to be the set of the first $n$ integers. That is
\[ [n] := \{1,2, \dots, n\}. \]

\subsection{Distances}

We use the $L_\infty$ metric to quantify distances between collections of
points and functions.

\begin{defn}[$L_{\infty}$ metric]\label{def:l-infty}
    Let $\overline{\R} := \R \cup \{\infty\}$. For points $p = (p_1, p_2, \dots, p_n)$, $q = (q_1, q_2, \dots, q_n) \in
    \overline{\R}^n$, we define the \emph{$L_{\infty}$ distance between
    points} $p$ and $q$ as $\norm{p-q}_\infty = \max_i|p_i-q_i|$
    For functions
    $f, g: K\rightarrow \R$ where $K\subset \R$ is compact, we define the \emph{$L_{\infty}$ distance
    between functions} $f$ and $g$ as $\norm{f-g}_{\infty} = \sup_{x \in
    K}|f(x)-g(x)|$.
\end{defn}

For a subset $X \subset \overline{\R}^n$, $x \in X$ and $\varepsilon>0$, let
$\square_{\varepsilon}(x)$ be the $L_{\infty}$ open neighborhood of radius
$\varepsilon$ centered at $x$. That is
\[ \square_{\varepsilon}(x) := \{ y \in X \mid \norm{y-x}_{\infty} < \varepsilon\}. \]

\subsection{$\varepsilon$-Perturbations} We consider perturbations of a function $f$.

\begin{defn}[$\varepsilon$-Neighborhood of $f$]
    Let $f \colon K \to \R$ be a continuous function, where $K\subset \R$ is compact.
    For $\varepsilon \geq 0$, define
    \[ N_{\varepsilon}(f) := \{g: K \to \R \mid g \text{ is continuous and } \norm{f-g}_{\infty} <
    \varepsilon\} \]
    to be the \emph{$\varepsilon$-neighborhood of $f$}. A function $g \in
    N_{\varepsilon}(f)$ is called an \emph{$\varepsilon$-perturbation of $f$}.
\end{defn}

\subsection{$\varepsilon$-Extremal Intervals}

Let $\text{INT}(C)$ be the set of relatively open intervals contained in \mbox{$C:=[a_1,a_2] \subset \R$}. To enable comparability between local extrema for functions in $N_\varepsilon(f)$ for different levels of $\varepsilon$, we use the following definition.

\begin{defn}[$\varepsilon$-Extremal Interval at $t$]
    Let $f:C\rightarrow \mathbb{R}$ be a continuous function and $T $ be the set of domain coordinates of all local extrema of $f$. Let $\varepsilon>0$. Define $\varphi_{\varepsilon}^f:T\rightarrow \text{INT}(C)$ such that
    \begin{description}
        \item[Case 1] If $t \in T$ and $(t, f(t))$ is a local minimum, define $\varphi^f_{\varepsilon}(t)$ to be the
            connected component of \mbox{$(f-\varepsilon)^{-1}(-\infty,f(t)+\varepsilon)$} that contains  $t$.
        \item[Case 2] If $t \in T$ and $(t, f(t))$ is a local maximum, define $\varphi^f_{\varepsilon}(t)$
            to be the connected component of \mbox{$(f+\varepsilon)^{-1}(f(t)-\varepsilon, \infty)$} that contains $t$.
    \end{description}
    We call $\varphi^f_\varepsilon(t)$ the \emph{$\varepsilon$-extremal interval at $t$} (see~\figref{nodelife-max}).
    \label{def:epsilon-intervals}
\end{defn}

We note that sometimes we omit the superscript $f$ and simply write $\varphi_{\varepsilon}$ when the function used to construct the $\varepsilon$-extremal intervals is clear.

\begin{rem}[Notation for Endpoints of $\varepsilon$-Extremal Intervals]
    Let $f:C\rightarrow \R$ be a continuous function with a local extremum at $t
    \in C$. Suppose $\varphi^f_{\varepsilon}(t)$ is the $\varepsilon$-extremal
    interval at $t$. We define the left endpoint of~$\varphi^f_{\varepsilon}(t)$
    to be $\le(\varphi^f_{\varepsilon}(t)):= \inf(\varphi^f_{\varepsilon}(t))$. We define the right endpoint
    of $\varphi^f_{\varepsilon}(t)$ to be $\re(\varphi^f_{\varepsilon}(t)):= \sup(\varphi^f_{\varepsilon}(t))$.
    Finally, we denote the length of  $\varphi_{\varepsilon}^f(t_i)$ by
    $\length(\varphi^f_{\varepsilon}(t))$.
\end{rem}

\begin{figure}[htp]
    \centering
    {\includegraphics[width=.6\textwidth]{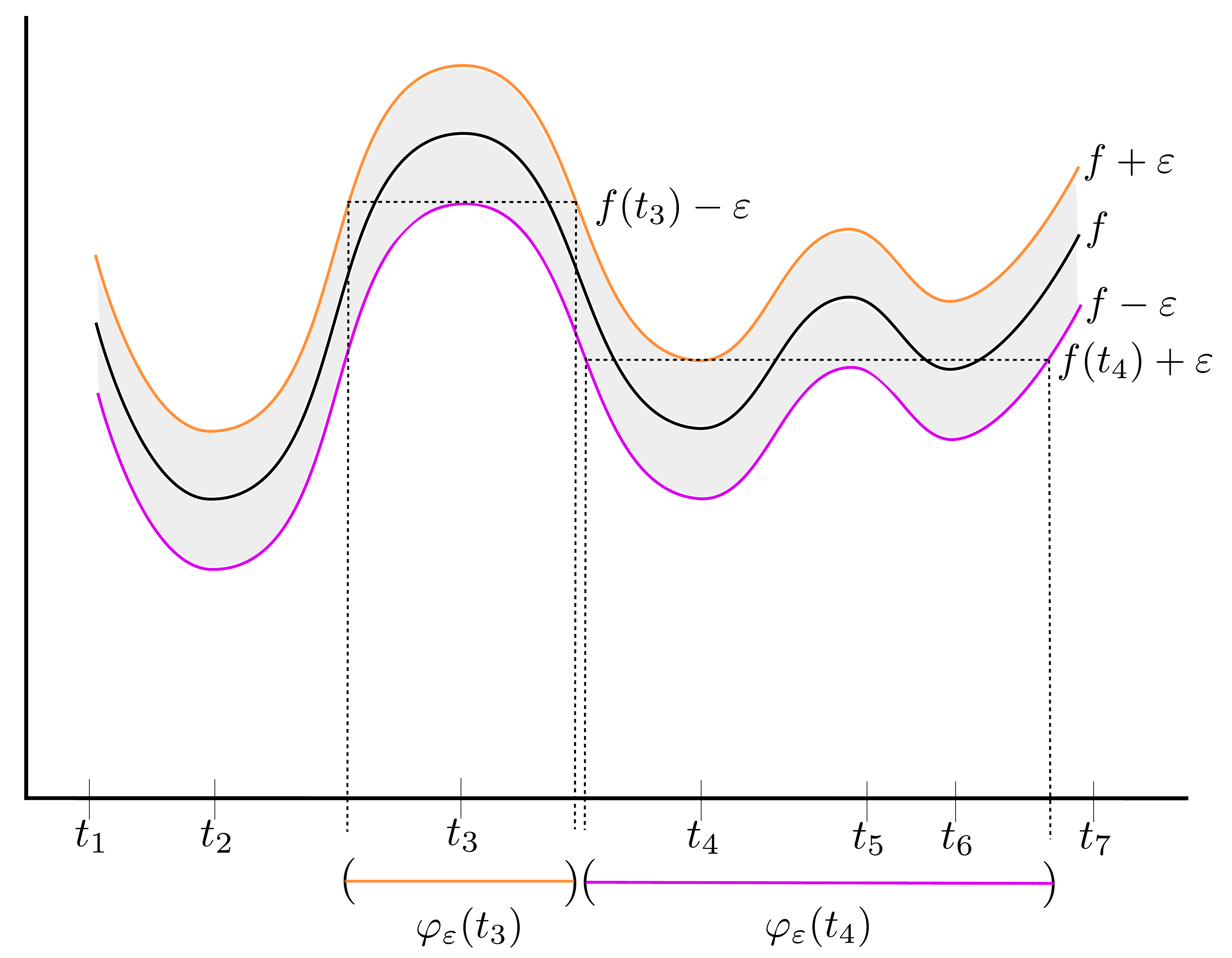}} \\
    \caption{The $\varepsilon$-extremal intervals at $t_3$ and $t_4$.
The $\varepsilon$-extremal interval at $t_3$ is the connected component of
\mbox{$(f+\varepsilon)^{-1}(f(t_3)-\varepsilon, \infty)$} that contains $t_3$ (Case 2). The
$\varepsilon$~-extremal interval at $t_4$ is the connected component of
\mbox{$(f-\varepsilon)^{-1}(-\infty,f(t_4)+\varepsilon)$} that contains $t_4$ (Case 1).}
    \label{fig:nodelife-max}
\end{figure}

\subsection{Sublevel Set Persistence}
We work with a specific case of persistent homology that
encodes the changes of the connectedness of sublevel sets of a real-valued function~$f: X
\rightarrow \R$ defined on a topological space $X$ as the height parameter ranges from~$-\infty$ to~$\infty$. This
information is encoded in a topological descriptor called a persistence diagram and encodes the prominence of
the local extrema of $f$. Persistent homology is a general mathematical framework
and we only provide the definitions necessary for our results here;
see \cite{EdelsbrunnerComputational10, PereaA19} for more detailed
introductions to persistence. Below we define sublevel set persistence using notation similar to 
page 181 of \cite{EdelsbrunnerComputational10}. One assumption we need is tameness of $f$. 

\begin{defn}[Homological Critical Values]
    Let $X$ be a topological space, $f:X \rightarrow \R$ a function. We
    call~ $a\in \R$ a \emph{homological critical value} if there exists a non-negative integer $n$ and
    $\delta > 0$ such that for all $0 < \varepsilon < \delta$, the linear map
    $H_n(f^{-1}(-\infty, a-\varepsilon]) \rightarrow H_n(f^{-1}(-\infty,
    a+\varepsilon])$ induced by the inclusion of sublevel sets is not an isomorphism.
\end{defn}

In other words, the homological critical values are the values at which the
homology of the sublevel sets change. For Morse functions $f$ on a smooth manifold, 
these points are exactly the heights of the local extrema of $f$ \cite{MilnorMorse63}. 
We also note that in this paper we work with $\mathbb{Z}/2\mathbb{Z}$ coefficients.

\begin{defn}[Tameness]
    Let $X$ be a topological space. A function $f: X \rightarrow
    \R$ is \emph{tame} if it has a finite number of homological critical values
    and the homology groups $H_n(f^{-1}(-\infty, a])$ are finite for every $a
    \in \R$.
\end{defn}

Next, we consider a nested sequence of sublevel sets. 
A \emph{filtration} of a topological space $X$ is a nested family of subspaces
$(X_r)_{r\in T}$ starting at the empty set, where $T \subset \R$,
such that for all $r, s \in T$ where $r\leq s$, we have $X_r \subset X_s$, and $\bigcup_{r \in T} X_r = X$. For
$f: X\rightarrow \R$, the sequence of all \emph{sublevel sets} $(f^{-1}(-\infty, r])_{r\in\R}$, ordered by inclusion and indexed
by $\R$, is called the \emph{sublevel set filtration}. 
In particular, for $r\leq s$, we have $f^{-1}(-\infty, r] \subset f^{-1}(-\infty, s]$. The inclusion map 
$\iota: f^{-1}(-\infty, r] \rightarrow f^{-1}(-\infty, s]$
 induces a homomorphism between homology groups 
 \[g_n^{r,s}: H_n(f^{-1}(-\infty, r]) \rightarrow H_n(f^{-1}(-\infty, s]).\]
 This homomorphism takes the homology of the sublevel set of $f^{-1}(-\infty, r]$ to the homology of the sublevel set of $f^{-1}(-\infty, s]$. The image of $g_n^{r,s}$ contains all this important information. We define the \emph{$n^{th}$ persistent homology groups} to be the images of all the homomorphisms $H_n^{r,s} := \im(g_n^{r,s})$. The \emph{$n^{th}$ persistent Betti numbers} are the ranks; $\beta_n^{r,s} := \text{rank}(H_n^{r,s})$. 

Next we describe how to encode the persistent homology groups into a multiset of points in the extended plane. Consider a tame function $f: X\rightarrow \R$. Let $(r_i)_{i=1}^N$ be the ordered sequence of homological critical values of $f$. Since we are working with a tame function, there are only a finite number of  heights we need to consider where the sublevel sets change. 
To ease notation, we write $H_n^{i,j} := H_n^{r_i, r_j}$, $\beta_n^{i,j}:=\beta_n^{r_i, r_j}$, and $g_n^{i,j} := g_n^{r_i, r_j}$. 
The persistent homology group $H_n^{i,j}$ consists of the homology classes of $f^{-1}(-\infty, r_i]$ 
that still persist in $f^{-1}(-\infty, r_j]$.  The persistent Betti number $\beta_n^{i,j}$ counts the number of homology classes that persist between $f^{-1}(-\infty, r_i]$ and $f^{-1}(-\infty, r_j]$. The first index for which a homology class appears is the \textit{birth} of that class. When a class in $f^{-1}(-\infty, r_{i-\varepsilon}]$ merges with another in $f^{-1}(-\infty, r_i]$, then the class \emph{dies} at a height of $r_i$. When classes merge together we follow the \emph{elder rule} (See section 7.1 of \cite{EdelsbrunnerComputational10}) that requires the class with a greater birth height to merge with the class of the lower birth height.  

\begin{defn}[Persistence Diagram $D_n(f)$]\label{def:pd}
    Let $f: X\rightarrow \R$ be a tame function with homological critical values $R := \{r_i\}_{i=1}^N$. 
    Consider the set $S := R \cup \{\infty\}$.
    The $n^{th}$ dimensional \emph{persistence diagram} $D_n(f)$ is the multiset
    set of points in $\overline{\R}^2$ such that the point~$p = (r_i, r_j)\in R \times S$ 
    where $r_i\leq r_j$ is included with multiplicity, 
    $$
     \mu_n(p) = \lim_{\varepsilon \rightarrow 0^+}  (\beta_n^{i,j}-\beta_n^{i,j+\varepsilon})-(\beta_n^{i-\varepsilon, j} - \beta_n^{i-\varepsilon, j+\varepsilon}).
   $$
    We set $p \in D_n(f)$ if, and only if, $\mu_n(p) >0$. 
    \end{defn}
In regards to $\mu_n(p)$, the first difference counts the number of homology classes that are born at or before a height of $r_i$ and die at a height of $r_j$. The second difference counts the number of homology classes that are born at or before a height of $r_i-\varepsilon$ and die at a height of $r_j$.

The persistence diagram summarizes the homology
groups as the height parameter ranges from~$-\infty$ to~$\infty$.  Each
persistence point $p=(b,d)\in D_n(f)$ is called a \emph{birth-death pair} since it represents a unique generator of the
homology groups of the sublevel sets of $f$ that is born at parameter~$b$ and
dies going into parameter $d$. In this work, we are concerned with a special type of tame function.

\begin{defn}[Nicely Tame Functions]
    Let $X \subset \R$ be a topological space. A function $f: X \rightarrow \R$
    is \emph{nicely tame} if $f$ is tame, continuous, and for each critical
    value $y$, the preimage $f^{-1}(y)$ is a finite set.
\end{defn}

Specifically, we work with a nicely tame function $f: [a_1,a_2] \rightarrow \R$ where $a_1,a_2 \in \R$. If the function values at the local
extrema are unique, then there is a one-to-one correspondence between
persistence points and the local minima of $f$. Then a persistence point $(b,d)$ corresponds to the local minimum $(t,f(t)=b)$ since the homological critical values are the local minima and maxima of $f$.

In the event that the values of several minima are the same, this correspondence is not unique. However, a unique
correspondence can be induced by fixing an order on the local minima (e.g.,
 the domain coordinates) and using that ordering to break ties. For a multiset $A$, we write $\size{A}$
for the \emph{total multiplicity of $A$} i.e., $\size{A} = \sum_{p \in
A} \mu(p)$. \figref{PD} gives an example of a function and its persistence
diagram from a sublevel set filtration.

In this work, we are only concerned with $D_0(f)$ and we denote $D(f) := D_0(f)$. 
Furthermore, all our persistence diagrams have a unique point in $D(f)$ with a death coordinate of $\infty$. 
We call this the \emph{essential connected component}. 
If $t$ does not represent the essential component, then there exists a local maximum $(t', f(t'))$ such that $(f(t),f(t')) \in D(f)$.
In this case, $f(t')$ is the height at which the connected
component of~$f^{-1}(-\infty, f(t')]$ containing $t$ merges with another
connected component of the sublevel set $f^{-1}(-\infty, f(t')]$ represented by
a local minimum $s$ where $f(s) < f(t)$.\footnote{The choice of pairings in
$D(f)$ follows the Elder Rule. In the case where both connected
components are born at the same height, we arbitrarily choose to continue the
connected component that occurs first in the domain.} We call $f(t')$ the
\emph{merge height of $t$}.

\begin{figure}[htp]
    \centering
    {\includegraphics[width=.7\textwidth]{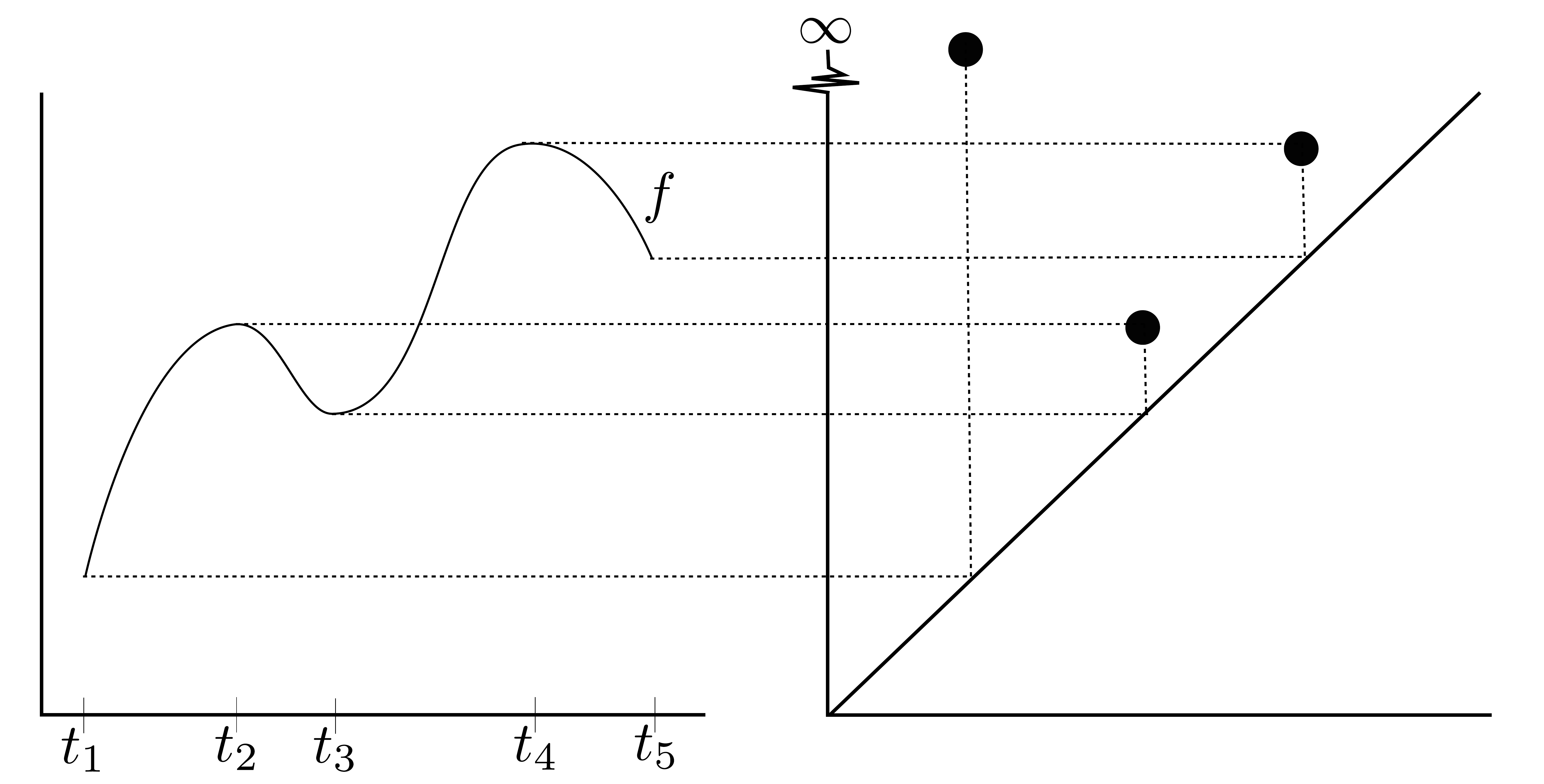}} \\
    \caption{Left. A function, $f:[t_1,t_5]\rightarrow \R$. Right. Persistence diagram of $f$, $D(f)$
	obtained from a sublevel set filtration of $f$.  The set
	$D(f)$ is $\{(f(t_1), \infty), (f(t_3), f(t_2)), (f(t_5), f(t_4))\}$. The first coordinate of each point in 
	$D(f)$ is the height of a local minimum, while the second coordinate is the height of a local maximum or $\infty$.  }
    \label{fig:PD}
\end{figure}

As mentioned before, sublevel set persistence has been previously studied to characterize
prominent features of data. In \cite{BendichHomology13}, well groups
that are defined using $\varepsilon>0$ and sublevel sets of a continuous
function between topological spaces are used to measure the robustness of
homology groups from sublevel sets of a function under any
$\varepsilon$-perturbation. Unlike $\varepsilon$-extremal intervals, well groups
are algebraic groups and are not specific to a local extremum. 

\subsection{Persistence of Extrema}
\label{sec:node-lives}

Given a function $f \colon C \to \R$ and the domain coordinate $t$ of a local
extremum, any continuous function $g\in N_\varepsilon(t)$ is guaranteed to have the same type of 
extremum in the interval~$\varphi^f_{\varepsilon}(t)$ for small enough $\varepsilon$. At some
value of $\varepsilon$, however, this is no longer guaranteed.
We use the persistence diagram $D(f)$ to assist us with understanding when this
occurs.

\begin{defn}[Birth-Death Pairing Map]
    Let $f:C\rightarrow \R$ be a nicely tame function.
    Let~$\{t_i\}_{i=1}^n$ be the
    set of domain coordinates of the local minima of $f$. Define the
        \emph{birth-death pairing map} to be~$\zeta_f: \{t_i\}_{i=1}^n \rightarrow \R_{>0}$ by 
    \[ \zeta_f(t_i) = \begin{cases}
        \max(f) & \text{if } t_i \text{ represents the essential component}\\
        f(t_j) & \text{ otherwise,}
    \end{cases}
    \]
    where $f(t_j)$ is the merge height of the minimum at $t_i$.
    \label{def:birth-death-map}
\end{defn}

Observe the minima for $f$ are the maxima of $-f$ and vice-versa. Additionally, the absolute difference in heights between extrema of $f$ remain the same in both $f$ and $-f$. Hence, we can study the prominence of maxima of $f$ by studying the minima of $-f$. This follows from \cite{BendichPersistent10} which discusses the symmetry between persistence diagrams computed from height filtrations that are ascending versus descending.

\begin{defn}[Persistence of Extrema]
    Let $f:C\rightarrow \R$ be a nicely tame function, and
    let \mbox{$(b,d) \in D(f)$}.  The persistence of $(b,d)$
    is the difference between the birth and death heights,
    $d-b$. Suppose $t$ is the domain coordinate such that $f(t)=b$
    and $(t,f(t))$ is the local minimum of $f$ representing the pair $(b,d)$.
    We define the \emph{persistence of the extremum}  $(t,f(t))$,
    denoted~$\pers_f(t)$, as
    \[ \pers_f(t) :=
        \begin{cases}
            \max(f)-f(t), & \text{if } (t, f(t)) \text{ is the global minimum of } f\\
            d-b, & \text{if } (t, f(t)) \text{ is a local (and not global) minimum of } f\\
            \pers_{-f}(t), & \text{if } (t, f(t)) \text{ is a local maximum of } f.\\
        \end{cases}
    \]
\end{defn}

See~\figref{nodelife} for an example of computing the persistence of local extrema.

\begin{defn}[Node Life]\label{def:nodelife}
    Let $f:C \rightarrow \R$ be a nicely tame function with a local extremum at domain coordinate $t$. The \emph{node life} of $t$ is $\pers_f(t)/2$.
\end{defn}

We sometimes omit the subscript $f$ from $\pers_f(t)$ when the function we are
computing the node life from is clear.
Proposition 2 and Corollary 1 from~\cite{BerryUsing20} states that $\varphi_\varepsilon^f(t)$ is
the smallest interval for which any nicely tame $\varepsilon$-perturbation of $f$ is guaranteed to have at least one local extremum of the same type as $t$, as long as $\varepsilon$ is less than the node life.

\begin{figure}[htp]
    \centering
    {\includegraphics[width=.7\textwidth]{images/nodelife}}
   {\includegraphics[width=.7\textwidth]{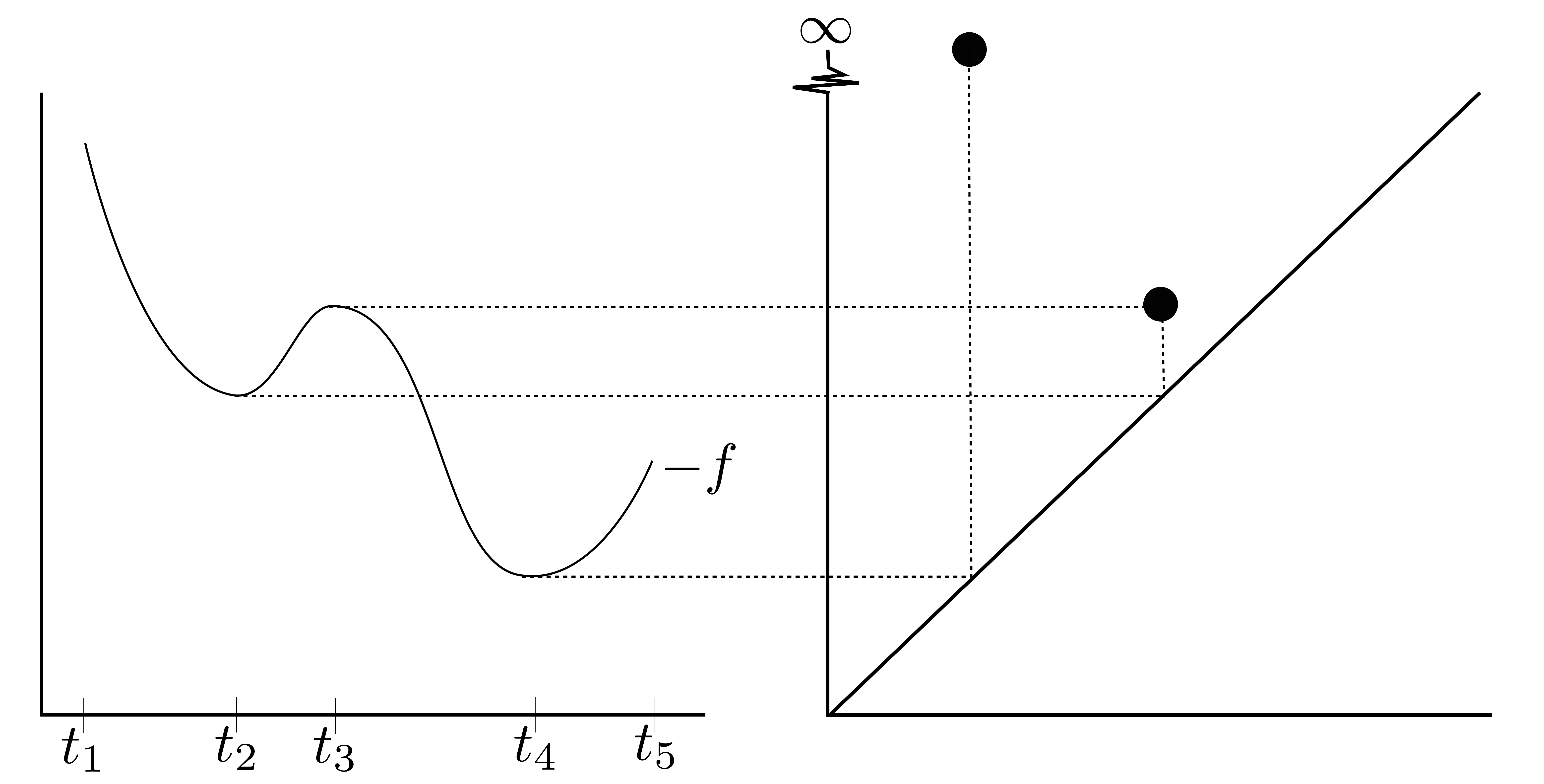}}
    \caption{Top. A continuous function $f$ and its persistence diagram from a sublevel set filtration.
In this example, $\pers_f(t_1) = \max(f)-f(t_1)$, $\pers_f(t_3)=f(t_2)-f(t_3)$, and $\pers_f(t_5)=f(t_4)-f(t_5)$.
Bottom. $-f$ and its persistence diagram from a sublevel set filtration. Now we can compute the persistence of the local maxmima of $f$ as 
$\pers_f(t_4) = f(t_4)-\min(f)$ and $\pers_f(t_2)=f(t_2)-f(t_3)$.}
    \label{fig:nodelife}
\end{figure}
\section{Extremal Event DAG}

To define the extremal event DAG, that will capture information about extrema of multiple time series, we need a notion of comparability of local extrema.

\begin{defn}[Comparability of Extrema]
Let $f,g: C\rightarrow \mathbb{R}$ be nicely tame functions.
Let $t_f, t_g$ be local extrema of $f$ and $g$ respectively. Let $\varepsilon > 0$. We declare $t_f \prec_\varepsilon t_g$ if for
every nicely tame $\varepsilon$-perturbation of $f$ and $g$ there exists $\varepsilon$-perturbed extrema
$t_f'$ and $t_g'$ such that $t_f' \in \varphi^f_{\varepsilon}(t_f)$, $t_g' \in \varphi^g_{\varepsilon}(t_g)$, and  $t_f'<t_g'$. We say $t_f$ and $t_g$ are
\emph{comparable at $\varepsilon$} if all of the following hold.
\begin{enumerate}
\item $\pers_f(t_f) > 2\varepsilon$
\item $\pers_g(t_g) > 2\varepsilon$
\item $t_f\prec_\varepsilon t_g$ or $t_g \prec_\varepsilon t_f$
\end{enumerate}
If at least one of these conditions does not hold, then $t_f$ and $t_g$ are \emph{incomparable at $\varepsilon$}. \label{def:comparable}
\end{defn}

\defref{comparable} relates order of extrema to possible $\varepsilon$-perturbations of functions. Using this definition, we are ready to define the extremal event DAG.

\begin{defn}[Extremal Event DAG]\label{def:extremal-DAG}
    Let $F=\{f_i:C \rightarrow \R\}_{i=1}^n$ be a collection of nicely tame
    functions. For $i \in [n]$, let $t_1^i<t_2^i<\dots <t_{n_i}^i$ be the domain coordinates for
    the local extrema of $f_i$.
    The \emph{extremal event DAG of $F$} is the directed graph,
    $\DAG(F):=(V, E,\omega_V,\omega_E)$, where
    \begin{itemize}
        \item $V:= \{ v(i,j) ~|~ i \in [n] \text{ and } j\in [n_i] \}$. In
            particular $v(i,j)\in V$ corresponds to the extremum of $f_i$ at~$t_j^i$.
        \item $E := \{ (v(i,j), v(r,s)) \mid t_j^i < t_s^r \}$.
        \item $\omega_V \colon V \to \R_{\geq 0}$
            is defined by the node life $\omega_V(v(i,j)):=\frac{1}{2}\pers_{f_i}(t_j^i)$.
            We call $\omega_V$ the \emph{node weights}.
        \item $\omega_E \colon E \to \R_{\geq 0}$ is defined by $\omega_{E}(v(i,j), v(r,s))
            := \inf\{\varepsilon\mid t_j^i \text{
                and } t_s^r \text{ are incomparable}\}$. We call~$\omega_E$ the \emph{edge weights}.
    \end{itemize}
\end{defn}

We note that we use the terms vertices and nodes interchangeably throughout this paper.
Given $\varepsilon>0$, we can easily recover $\varepsilon$-$\DAG(F)$ from
$\DAG(F)$ where $\varepsilon$-$\DAG(F)$ is defined in \cite{BerryUsing20}.
Specifically, $\varepsilon$-$\DAG(F)$ is the subgraph of the
$\DAG(F)$ that consists of vertices and edges with a weight less than or equal
to $\varepsilon$. Hence, $\DAG(F)$ is a stronger descriptor since it is not
dependent on $\varepsilon$.

Computing the vertices and directed edges of the extremal event DAG can be done
directly from the graphs of the functions in $F$. Computing the
weights requires more information. Expanding upon earlier observation about node lives, we 
note that we choose to define the node weights as the node lives for the following reason. 
Proposition 2 of~\cite{BerryUsing20} states that every nicely tame $g \in
N_{\varepsilon}(f)$ has a local extremum of the same type as $t$, say $t' \in
\varphi^f_{\varepsilon}(t)$ as long as $\varepsilon < \frac{1}{2}\pers_f(t)$. Furthermore,
Proposition 1 of~\cite{BerryUsing20} states that for any two local extrema at
$(s, f(s))$, and $(t, f(t))$ of $f$ of the same type, we have $\varphi^f_{\varepsilon}(t) \cap
\varphi^f_{\varepsilon}(s) = \emptyset$.  Hence, when $\varepsilon <\frac{1}{2}\pers_f(t)$,
we guarantee a relative ordering of extrema for $\varepsilon$-perturbations of
$f$. If $\varepsilon > \frac{1}{2}\pers_f(t)$, Proposition 1 of~\cite{BerryUsing20} does not apply 
and we lose the association between the extrema of the perturbed function of $g \in N_{\varepsilon}(f)$ and the extremum of $f$ at $t$.

\subsection{Computing Edge Weights}

In \appref{ext-intervals-prop}, we prove a few properties of $\varepsilon$-extremal intervals that are needed to prove the condition for computing edge weights. We state a condition for checking that requirement  3 of~\defref{comparable} is met.

\begin{thm}[Computing Edge Weights]

    Let $F=\{f_i:C \rightarrow \R\}_{i=1}^n$ be a collection of nicely tame
    functions where $t_1^i<t_2^i<\dots <t_{n_i}^i$ are all the domain coordinates of the local extrema of
    $f_i$. Let \mbox{$\DAG(F) := (V,E,\omega_V,\omega_E)$} be the extremal event DAG
    of~$F$. For all edges~$(v(i,j), v(c,d)) \in E$, the following statements hold

    \begin{enumerate}
        \item If $i=c$, then
            $$\omega_{E}(v(i,j), v(c,d)) = \min\{\omega_V(v(i,j)), \omega_V(v(c,d))\}.$$ \label{stmt:edge-weights-same}
        \item If $i\neq c$, then
            $$\omega_{E}(v(i,j), v(c,d))=\min\{\omega_V(v(i,j)), \omega_V(v(c,d)), \varepsilon^*(t_j^i, t_d^c)\},$$
            where
            $$\varepsilon^*(t_j^i, t_d^c) := \inf\{\varepsilon \mid \varphi^{f_i}_{\varepsilon}(t_j^i) \cap \varphi^{f_c}_{\varepsilon}(t_d^c) \neq \emptyset\}.$$
            \label{stmt:edge-weights-dif}
    \end{enumerate}
    \label{thm:edge-weights}
\end{thm}

The proof of \thmref{edge-weights} is technical and involves analyzing several cases. 
We provide the proof of \thmref{edge-weights} in \appref{edgeweights}.

\subsection{Example of Extremal Event DAG Construction}

We give an example of constructing an extremal event DAG.

\begin{example}
    We construct the extremal event DAG for \mbox{$\sin(x):[0,2\pi]\rightarrow \R$},
    $\cos(x):[0,2\pi]\rightarrow \R$ as illustrated in~\figref{extremalDAGsine}. First we
    compute the persistence diagram from a sublevel set filtration of
    $\sin(x)$, $-\sin(x)$, $\cos(x),$ and $-\cos(x)$.
    This computes the node lives of all local extrema in $\sin(x)$
    and $\cos(x)$. These node lives are the node weights in the extremal event DAG. Next we
    compute edge weights between nodes based on \thmref{edge-weights}(1).
    
    For an illustration of this, consider the local extrema at
    $x=0$ and $x=\frac{\pi}{2}$~of~$\sin(x)$. The node lives of these two local
    extrema are $0.5$ and $1$ respectively. The edge weight between the two corresponding
    vertices in the extremal event DAG is the minimum of these two node lives: $0.5$.

    Computing
    the edge weights between two vertices corresponding to different function is more involved.
    We need to apply~\thmref{edge-weights}(2). To illustrate this, consider the local extrema at
    $x=\frac{\pi}{2}$ of~$\sin(x)$ and $x=\pi$ of $\cos(x)$. Since $\sin(x)$ and
    $\cos(x)$ are translations of one another, the
    $\varepsilon$-extremal intervals grow at the same rate for both $\sin(x)$ and
    $\cos(x)$. We know that $\varphi^{\sin}_{\varepsilon}(\pi/2)$ and $\varphi^{\cos}_{\varepsilon}(\pi)$
    first intersect at the half-way point of the domain coordinates, which is $\frac{3\pi}{4}$.
    Using the definition of the $\varepsilon$-extremal  intervals we find

    \begin{align*}
        \sin(\pi/2)-\varepsilon &= \sin(3\pi/4)+\varepsilon\\
        \varepsilon &= \frac{1}{4}(2-\sqrt{2}) \approx 0.14
    \end{align*}

    \begin{align*}
        \cos(\pi)+\varepsilon &= \cos(3\pi/4)-\varepsilon\\
        \varepsilon &= \frac{1}{4}(2-\sqrt{2}) \approx 0.14
    \end{align*}

    The epsilon value computed is the infimum $\varepsilon$ for which $\varphi^{\sin}_\varepsilon(\pi/2)$
    and $\varphi^{\cos}_\varepsilon(\pi)$ both contain $\frac{3\pi}{4}$. Hence this is also the infimum
    $\varepsilon$ for which $\varphi^{\sin}_\varepsilon(\pi/2)\cap \varphi^{\cos}_\varepsilon(\pi)\neq \emptyset$.
    Since $0.14$ is less than the node life of either extremum, then by~\thmref{edge-weights}(2), the
    edge weight between the vertices corresponding to $x=\pi/2$
    of $\sin(x)$ and $x=\pi$ of $\cos(x)$ is $0.14$. Applying a similar process to all edges, we get
    the extremal event DAG.
    
\begin{figure}[htp]
    \begin{subfigure}[b]{0.45\textwidth}
         \centering
         \includegraphics[width=\textwidth]{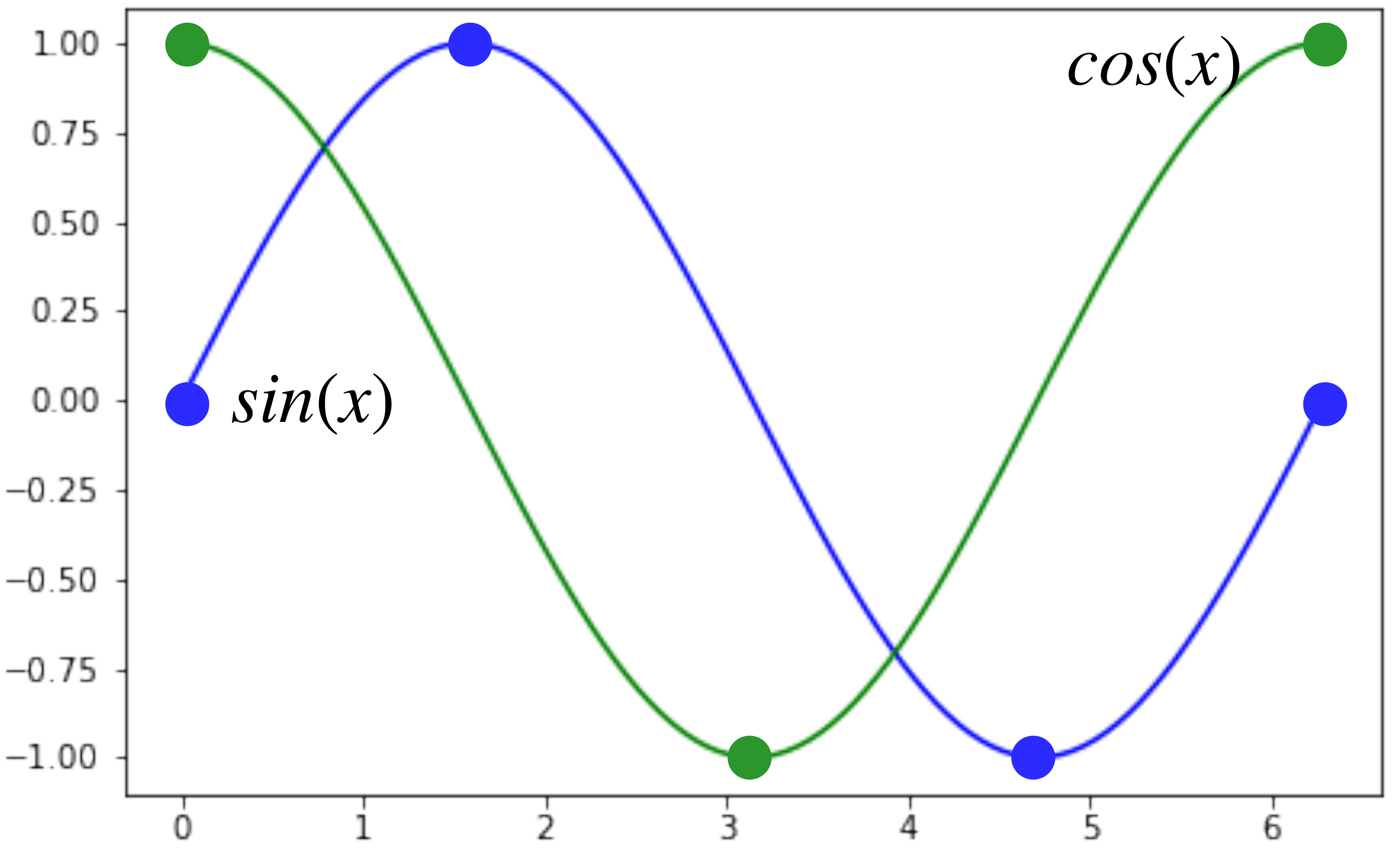}
         \caption{Functions.}
         \label{fig:fxns}
     \end{subfigure}
     \hfill
     \begin{subfigure}[b]{0.4\textwidth}
         \centering
         \includegraphics[width=\textwidth]{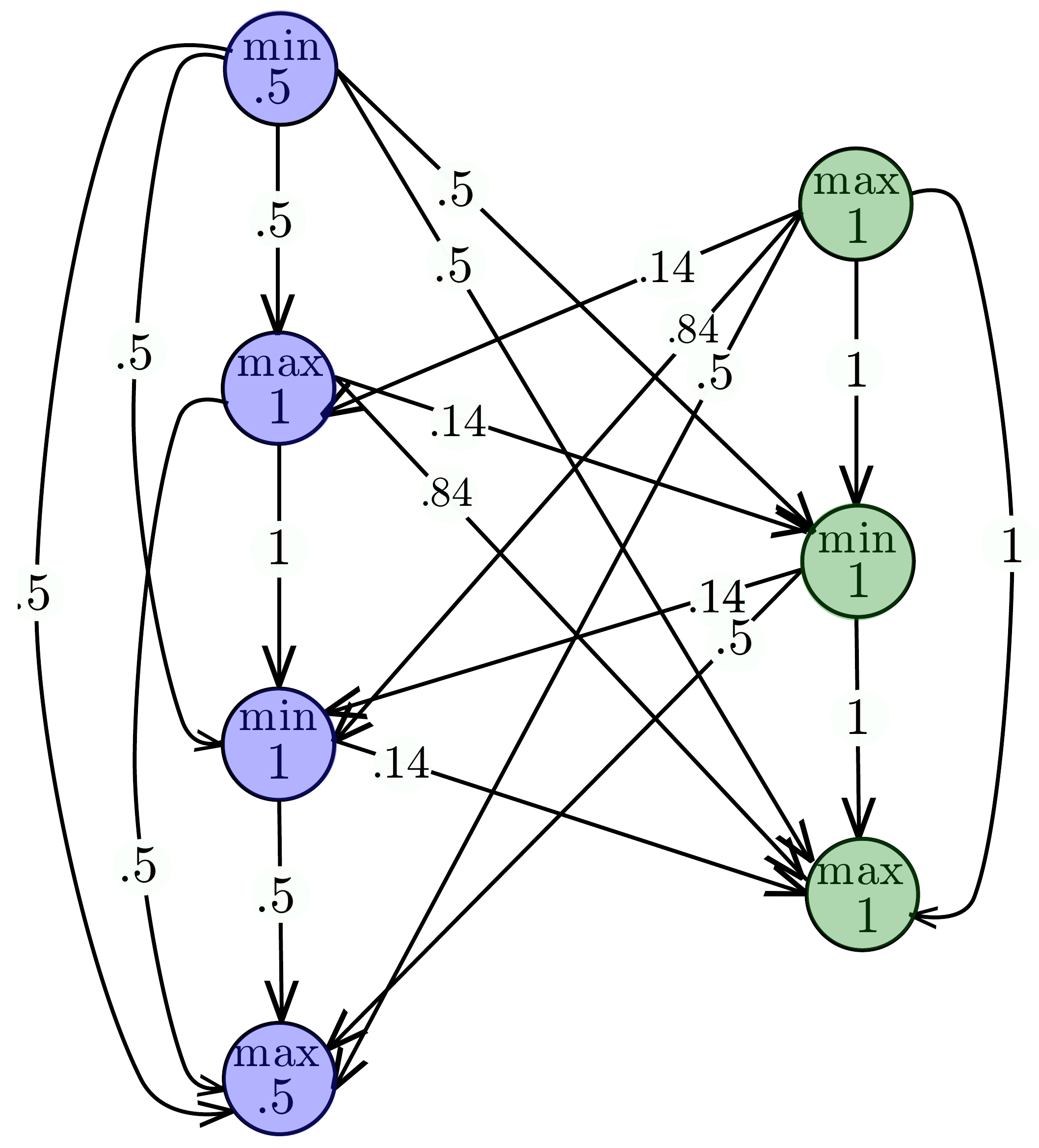}
         \caption{Extremal Event DAG.}
         \label{fig:DAG}
     \end{subfigure} \\
       \caption{Extremal Event DAG for $\sin(x):[0,2\pi]\rightarrow \R$ and $\cos(x):[0,2\pi]\rightarrow \R$.
        The vertices on the left and highlighted in blue represent the local extrema of $\sin(x)$ while the 
        vertices on the right and highlighted in green represent the local extrema of $\cos(x)$. The vertices 
        highlighted in blue from top to bottom correspond to the local extrema of $\sin(x)$ in ascending order 
        by domain coordinate. For example, the top blue vertex with label and weight $(\min, .5)$ corresponds 
        to the local extremum $(0, 0)$, the second blue vertex $(\max, 1)$ corresponds to the local extremum $(\frac{\pi}{2}, 1)$, 
        etc. Similarly the green vertices correspond to the local extrema of cosine in ascending order by domain coordinate. 
        Directed edges indicate the ordering of the domain coordinates of the local extrema. The vertex weights are the 
        node lives of the corresponding local extrema while the edge weights are computed using~\thmref{edge-weights}.}
        \label{fig:extremalDAGsine}
    \end{figure}
\end{example}

\section{Extremal Event DAG Distance}

In this section, we define a distance between extremal event DAGs
representing different collections of time series, or \textit{datasets}. In
particular, we first discuss the alignment of nodes between different
collections of time series and then the alignment of edges. We call the result
an \textit{extremal event supergraph}. The weights on the nodes and edges of the
extremal event supergraph are determined by the weights on the two extremal event DAGs. The
distance between the extremal event DAGs is then computed from the weights on the
extremal event supergraph.

\subsection{Backbone Distance}

From~\figref{extremalDAGsine}, one can see that each time series is translated
into an ordered linear sequence of alternating minima and maxima. These linear
sequences greatly simplify the comparison between datasets assuming that 
there is a one-to-one correspondence between the identifications of each of the time series
in each dataset. For example, consider two
gene expression datasets under different experimental conditions. In this case,
each time series has a unique identity corresponding to the gene that it
represents. Our primary task in this section is to perform a matching operation
between the extrema of two time series with matching identities. To perform the
matching, we use a modified version of the edit distance. Throughout this section, we refer to
\figref{time-series-and-dags} for illustrations of definitions.

\begin{figure}[htp]
     \centering
     \begin{subfigure}[b]{0.4\textwidth}
         \centering
         \includegraphics[width=\textwidth]{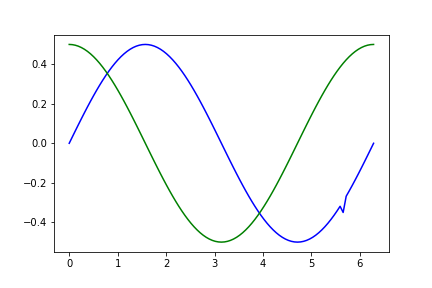}
         \caption{Dataset 1}
         \label{fig:timeseries1}
     \end{subfigure}
     \hfill
     \begin{subfigure}[b]{0.4\textwidth}
         \centering
         \includegraphics[width=\textwidth]{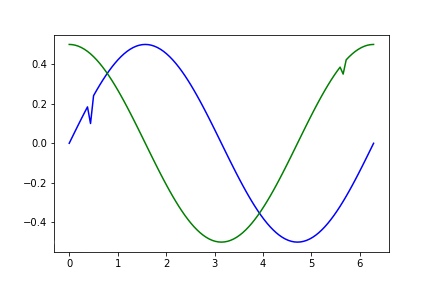}
         \caption{Dataset 2}
         \label{fig:timeseries2}
     \end{subfigure} \\
     \begin{subfigure}[b]{0.4\textwidth}
         \centering
         \includegraphics[width=\textwidth]{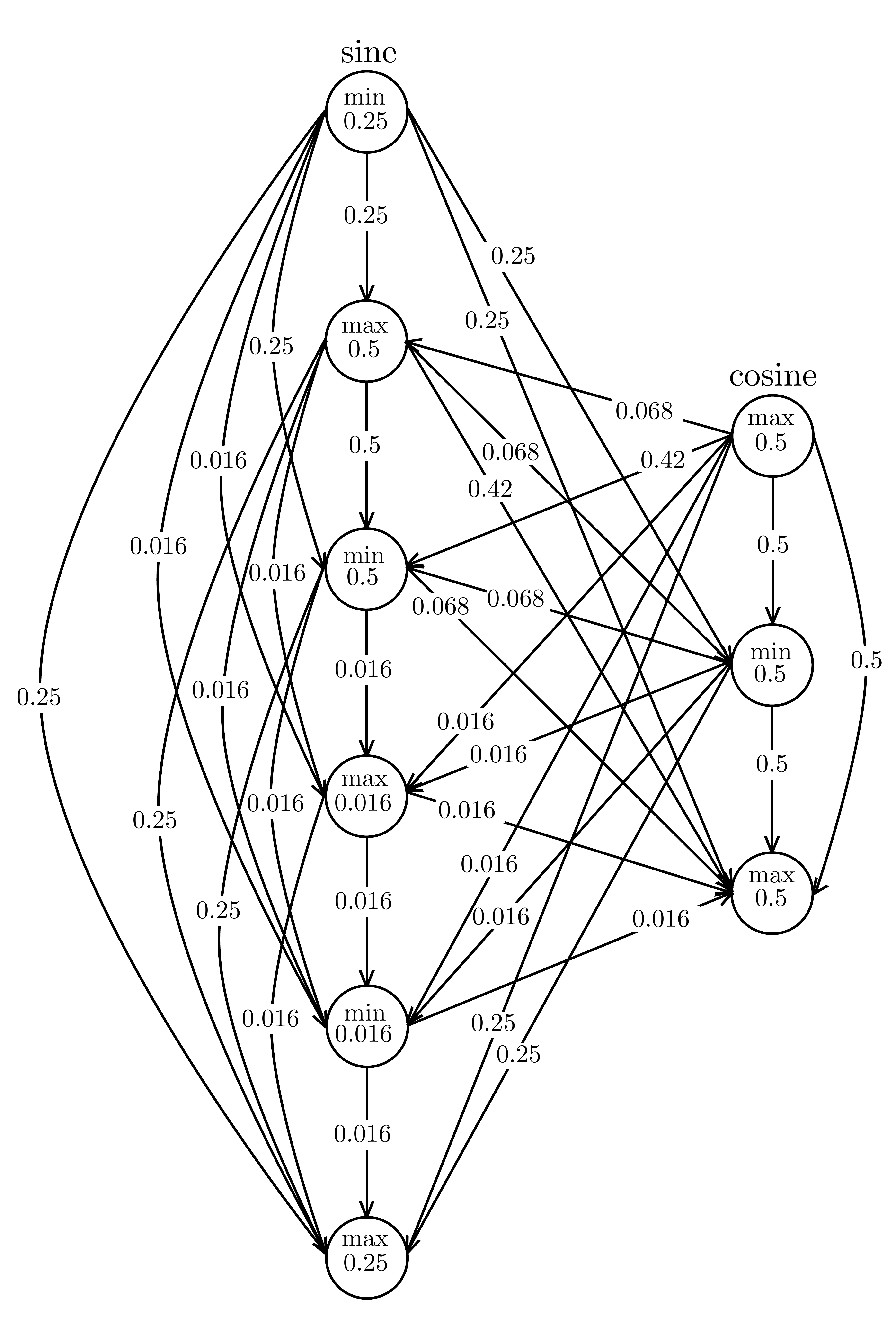}
         \caption{Extremal Event DAG 1}
         \label{fig:DAG1}
     \end{subfigure}
     \hfill
          \begin{subfigure}[b]{0.4\textwidth}
         \centering
         \includegraphics[width=\textwidth]{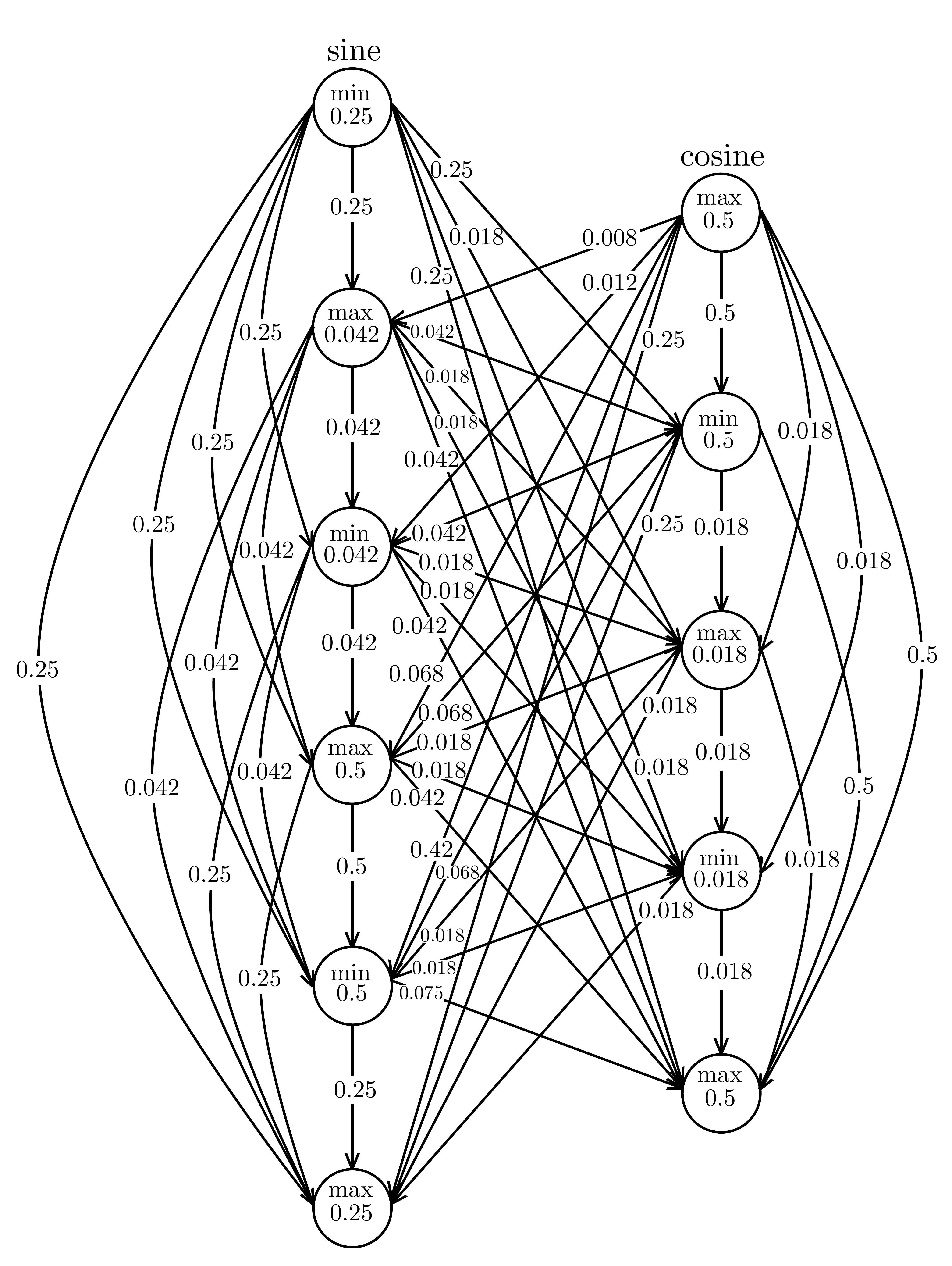}
         \caption{Extremal Event DAG 2}
         \label{fig:DAG2}
     \end{subfigure}

        \caption{Time Series Data and Corresponding Extremal Event DAGs. We consider two 
        datasets consisting of two functions, $\frac{1}{2}\sin(x)$ and
        $\frac{1}{2}\cos(x)$ over $[0, 2\pi]$ with some added noise. In
        \figref{timeseries1} and \figref{timeseries2}, we label the blue curve
        as ``sine" and green curve as ``cosine". \figref{DAG1} is the
        extremal event DAG for Dataset 1 while \figref{DAG2} is
        extremal event DAG for Dataset 2.}
        \label{fig:time-series-and-dags}
\end{figure}

\begin{defn}[Backbones] A \emph{backbone} is a finite sequence $\x = (x_1, x_2,
    \dots , x_n)$, where each $x_i$ is a tuple~$x_i=(s_i, w_i)$ with $s_i$ a string, and $w_i \in
    \R_{\geq 0}$. The empty string is denoted by 0.
    The \emph{length} of $\x$ is denoted $\length(\x)$, and is equal to the
    number of elements in the sequence (here, $\length(\x)=n$).
    We call each $x_i$ a
    \emph{node} and we denote the first $k$ terms of $\x$ by~$\xtok{k}.$
\label{def:backbone}
\end{defn}

\begin{rem}[Constructing Backbones from Nicely Tame Functions]
    In what follows, we construct a backbone from a nicely tame
    function $f \colon C \to \R$ by computing $\DAG(\{
        f\})=(V,E,\omega_V,\omega_E)$ and removing all
    edges and edge weights; the nodes are ordered
    by their corresponding domain coordinates. Then, the data associated to each
    node $v \in V$ is a string representing which type of local maxima (min or
    max) along with the node weight $\omega_V(v)$.  This backbone for
    $f$ is denoted as $B(f)$.
    See~\figref{backbones} for an example. 
\end{rem}

\begin{figure}[htp]
    \centering
    \begin{subfigure}[b]{\textwidth}
        \centering
        \includegraphics[width=\textwidth]{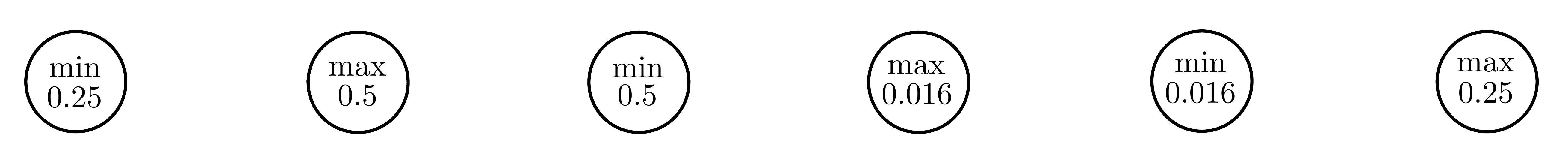}
        \caption{Sine Backbone 1}
        \label{fig:backbone1}
    \end{subfigure} \\
    \begin{subfigure}[b]{\textwidth}
        \centering
        \includegraphics[width=\textwidth]{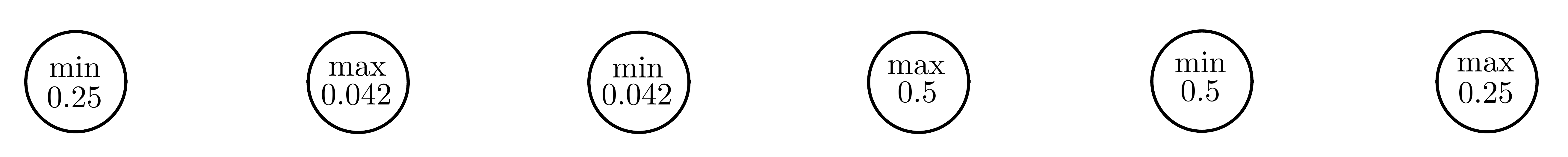}
        \caption{Sine Backbone 2}
        \label{fig:backbone2}
    \end{subfigure}
    \caption{Extracting Sine Backbones from Extremal Event DAG 1 and Extremal Event DAG 2.
        \figref{backbone1} illustrates the backbone where each node corresponds
        to a local extremum of the sine labeled curve from Dataset 1 (see
        \figref{timeseries1}). Mathematically, this backbone is the sequence
        $(\min, 0.25)$, $(\max, 0.5)$, $(\min, 0.5)$, $(\max, 0.016)$, $(\min, 0.016)$,
        $(\max, 0.25)$.  \figref{backbone2} illustrates the backbone where each
        node corresponds to a local extremum of the sine labeled curve from Dataset 2
         (see \figref{timeseries2}). Mathematically, this backbone is
        the sequence $(\min, 0.25)$, $(\max, 0.042)$, $(\min, 0.042)$, $(\max, 0.5)$,
        $(\min, 0.5)$, $(\max, 0.25)$. }
        \label{fig:backbones}
\end{figure}

\begin{rem}[Backbones as Sets]
    We consider functions over backbones to other spaces. For these settings, we
    think of a backbone as an ordered multiset,~$\x = \{x_1, x_2, \dots, x_n\}$,
    (i.e., repeated elements are allowed) equipped  with an injective index function,
    $\iota_\x: \x \rightarrow [n]$ where $\iota_\x(x_i) = i$.
    Let $\zero := (0,0)$. We
    also define $\tilde{\x} = \{\zero\} \cup \x$, where $\zero$ is the
    \emph{empty node.} The function $\iota$ is not extended to $\tilde{\x}$.
\end{rem}

Next we discuss alignments and how to compute a distance between two backbones using an optimal alignment.

\begin{defn}[Alignment]\label{def:alignment}
    Let $\x = (x_1,x_2,\dots, x_m)$ and $\y = (y_1,y_2,\dots, y_n)$ be
    backbones. An \emph{alignment} is a totally ordered correspondence between $\tilde{\x}$ and
    $\tilde{\y}$ that does not repeat elements of~$\x$ or $\y$ and respects the
    labels (or strings) of the backbones. We say that the
    number of pairs in the correspondence is the \emph{length} of the alignment.
    In particular, we represent an alignment of length $k$ between $\x$ and
    $\y$ as a  function $\alpha: [k] \to \tilde{\x} \times \tilde{\y}$,
    where $\alpha(i)$ can be written as two coordinate
    functions~$\alpha(i):= (\alpha_{\x}(i), \alpha_{\y}(i))$, such that
    \begin{enumerate}
        \item \textbf{No Null Alignments.} The pair $(\zero,\zero)$ is not in the image of
            $\alpha$, which we denote by $\im(\alpha)$. \label{property:nullalignments}
        \item \textbf{Preserves Order of Backbones.} The coordinate functions $\alpha_\x: [k] \to \tilde{\x}$, $\alpha_\y:
            [k]\rightarrow \tilde{\y}$ are \emph{partially monotone}. The
            function $\alpha_\x$ is partially monotone if for every
            $i,j \in [k]$ such that~$\alpha_\x(i)\neq \zero$ and $\alpha_\x(j) \neq
            \zero$, we have
            $$
                \iota_\x(\alpha_\x(i)) < \iota_\x(\alpha_\x(j))
                \text{ if and only if } i < j.
            $$
            An analogous definition applies to $\alpha_\y$. \label{property:preservesbackbones}
        \item \textbf{No Misalignments.}
            For each $\left((s_x,w_x),(s_y,w_y)\right) \in \im(\alpha)$,
            we either have equality in strings \mbox{$s_x = s_y$}, or one of $(s_x, w_x)$, $(s_y, w_y)$ is equal to $\zero$.  
            \label{property:nomisalignments}

        \item \textbf{Restriction to Matching.} Each element of $\x$ and $\y$ appears in the image of $\alpha_{\x}$ and $\alpha_{\y}$ exactly once. 
        That is, for each $x_i \in \x$, there exists exactly one $j \in [k]$ for which $\alpha_{\x}(j)=x_i$. The analogous statement holds 
        for each $y_i \in \y$. \label{property:restrictiontomatching}
    \end{enumerate}
    If $\alpha(i)=(\alpha_{\x}(i), \zero)$, we say that $\alpha_{\x}(i)$ is aligned
    with an \emph{insertion}; similarly for $\alpha(i)=(\zero,\alpha_{\y}(i))$. We
    denote the restriction of $\alpha$ to the first $h$ integers, $[h] =\{1, 2,
    \dots, h\} \subset [k]$ as $\alphtok{h}$.
\end{defn}

\begin{notn}[Elements in $\im(\alpha)$]\label{not:notation-image-alpha}
 When we use the notation $(x,y) \in \im(\alpha)$, we always assume that $x\not = \zero$ 
 and $y \not = \zero$. We also use notation $(x,\zero) \in \im(\alpha)$ and 
 $(\zero,y) \in \im(\alpha)$ to denote that $x$ or $y$ is aligned with an insertion.
\end{notn}

Note that the restriction of $\im(\alpha) \cap (\x \times \y)$ is a partial
matching (that is, each element in $\x \times \y$, if not aligned with an
insertion, is aligned with a distinct element of the other backbone). We call
any pair~$(x,y) \in \im(\alpha) \cap (\x \times \y)$ a \emph{nontrivial match}.
An example of two different alignments of the sine backbones shown
in~\figref{backbones} is given in~\figref{alignments}. \figref{alignment1} is
aligned without insertions, while~\figref{alignment2} has two insertions in each
of the backbones. Notice that the insertions occur at the small noisy extrema in
each of the time series.

\begin{figure}[htp]
    \centering
    \begin{subfigure}[b]{\textwidth}
        \centering
        \includegraphics[width=\textwidth]{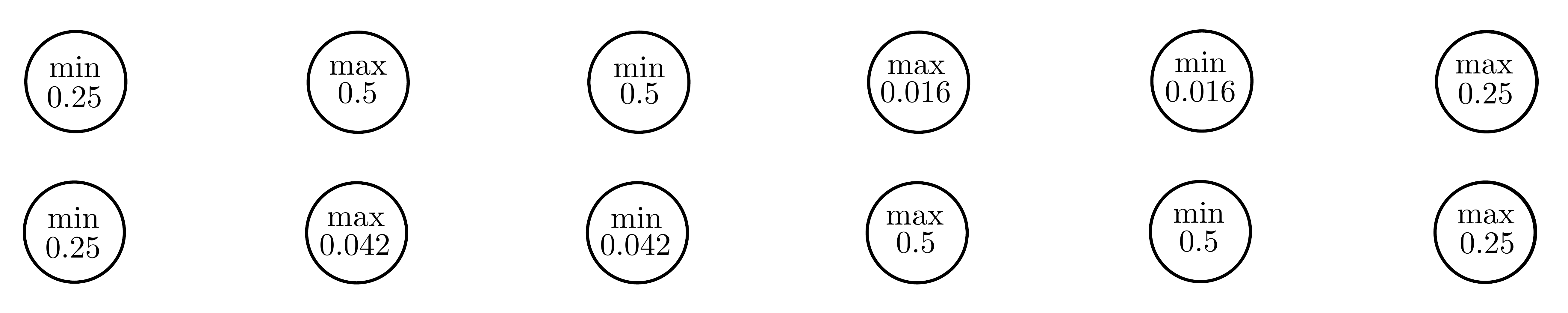}
        \caption{Alignment 1}
        \label{fig:alignment1}
    \end{subfigure} \\
    \begin{subfigure}[b]{\textwidth}
        \centering
        \includegraphics[width=\textwidth]{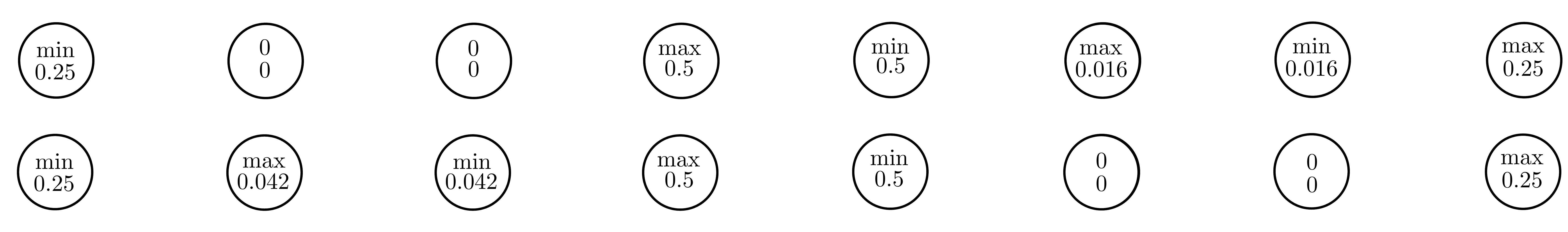}
        \caption{Alignment 2}
        \label{fig:alignment2}
    \end{subfigure}
    \caption{Two Possible Alignments of Sine Backbones. We consider the
        backbones shown in \figref{backbones}. Call these $\x$ and $\y$
        respectively. The top row consists of nodes from $\x$ while the bottom row consists of nodes from $\y$.
        \figref{alignment1} gives an alignment, $\alpha_1: \{1, 2,
        \dots, 6\} \rightarrow \tilde{\x} \times \tilde{\y}$ of the two
        backbones where $\alpha_1(i) = (x_i, y_i)$. \figref{alignment2} gives an
        alignment $\alpha_2: \{1, 2, \dots, 8\}\rightarrow \tilde{\x} \times
        \tilde{\y}$ where $\alpha_2(1) = (x_1, y_1),$ $\alpha_2(2) = (\zero, y_2),$
        $\alpha_2(3) = (\zero, y_3),$ $\alpha_2(4) = (x_2, y_4),$
        $\alpha_2(5)=(x_3, y_5),$  $\alpha_2(6) = (x_4, \zero),$ $\alpha_2(7) =
        (x_5, \zero),$ and $\alpha_2(8) = (x_6, y_6)$.} \label{fig:alignments}
\end{figure}

\begin{defn}[Cost of Alignment]\label{def:cost}
    Let $\x$ and $\y$ be backbones and $\alpha:[k] \rightarrow \tilde{\x} \times \tilde{\y}$ be an
    alignment of length $k$. The \textit{cost} of $\alpha$ is defined as
    $$
        \cost(\alpha) := \sum_{(x,y), (x, \zero), (\zero, y) \in \im(\alpha)}
        |w_x-w_y|,
    $$
    where $x=(s_x,w_x)$ and $y=(s_y,w_y)$.
    We define the cost of the partial alignment $c_{\x, \y}(i,j)$ to be the minimum cost of aligning $\xtok{i}$
    with $\ytok{j}$, that is,
    $$
        c_{\x, \y}(i, j)
        :=
        \min\{ \cost(\alpha) \mid \alpha \text{ is an alignment of } \xtok{i} \text{ and } \ytok{j} \}.
    $$
\end{defn}

Referring again to~\figref{alignments}, we compute that the alignment
in~\figref{alignment2} has a lower cost, 0.116, than that
in~\figref{alignment1}, 0.932.

\begin{defn}[Optimal Alignment]
    Let $\x=(x_1, x_2, \dots, x_m)$ and $\y=(y_1, y_2, \dots, y_n)$ be
    backbones. We call an alignment $\alpha:[k] \rightarrow \tilde{\x} \times
    \tilde{\y}$ \emph{optimal} if $\cost(\alpha)=c_{\x, \y}(m,n)$.
    \label{def:optimal-cost}
\end{defn}

An optimal alignment minimizes cost. We note that there could be multiple
alignments that minimize cost and so an optimal alignment is not necessarily
unique. 

We define the distance between two backbones~$\x$ and~$\y$ using an
optimal alignment. To do this, we need to identify nontrivial matches, i.e.,
those alignment pairs that do not involve insertions.

\begin{defn}[Backbone Distance]\label{def:backbone-dist}
    The \emph{backbone distance} between backbones $\x$ and $\y$ is defined
    as
    \begin{equation}\label{backbone-dist}
    d_{\B}(\x, \y) = \inf_{\alpha}
        \left (\sum_{(x,y) \in \im(\alpha)} |w_x - w_y|
        + \sum_{(x,\zero) \in \im(\alpha)} w_x
        + \sum_{(\zero,y) \in \im(\alpha)} w_y \right )
    \end{equation}
    where $\alpha$ ranges over all alignments between $\x$ and $\y$.
\end{defn}

The backbone distance finds the best alignment between $\x$ and $\y$, then
defines  the distance to be the $L_1$ norm between a vector consisting of the
node weights in $\tilde{\x}$ and a vector consisting of the matching node weights
in $\tilde{\y}$. The first term of Equation \ref{backbone-dist} accounts for the
cost of the nodes in $\x$ that are aligned with nodes in $\y$, the second term
accounts for the cost of the nodes in $\x$ that are aligned with an insertion,
and the third term accounts for the cost of the nodes in $\y$ that are aligned
with an insertion. We show that this distance is in fact a metric in \appref{backbone}. 


\subsection{Extremal Event DAG Distance}
Using the backbone distance, we define a distance between two extremal event DAGs, $D$
and $D'$, constructed from comparable datasets. In other words, the number and
identity of the time series are the same between the two datasets so that the
choice of which backbones to align is clear. Once the alignments have been
computed, we construct a supergraph based on $D$ and $D'$ where there is a vertex
for each node pair from each alignment. We add a directed edge between two
vertices if the edge exists between the two vertices in either $D$ or $D'$. After
we construct the supergraph, we impose two weight functions on the vertices and
nodes given by the weights of the nodes and edges in $D$ and $D'$ respectively.
The difference between these weight vectors is the extremal event DAG distance between
$D$ and $D'$.

\begin{defn}[Extremal Event Supergraph]\label{def:supergraph}\label{def:supergraph-weights}
    Let $D = (V, E, \omega_{V}, \omega_{E})$ and $D' = (V', E',
    \omega'_{V}, \omega'_{E})$ be two extremal event DAGs with $n$ pairs of aligned
    backbones. Let $\x_1, \x_2, \dots, \x_n$ be the backbones of $D$, and~$\y_1,
    \y_2, \dots, \y_n$ be the backbones of $D'$ and, for each $i \in [n]$, let
    $\alpha^{(i)}: [k_i]
    \to \tilde{\x}_i \times \tilde{\y}_i$ be the corresponding alignments where $k_i = \length(\alpha^i)$.
    Just as we expanded $\alpha$ to two coordinate functions in
    \defref{alignment},
    $\alpha^{(i)}$ can be expanded into two coordinate
    functions~$\alpha^{(i)}:= (\alpha_{\x_i}, \alpha_{\y_i})$. The
    \emph{extremal event supergraph determined by the alignments $\{\alpha^{(i)}\}_{i=1}^n$} of $D$ and
    $D'$ is a doubly weighted directed graph~$(V_\alpha, E_\alpha,
    \omega_{\alpha}, \omega'_{\alpha})$, where
    \begin{itemize}
        \item $V_\alpha := \{ v(i,j) \mid i \in [n], j \in [k_i] \} $. That is, the vertices of $V_{\alpha}$ are in one-to-one correspondence with each element of every alignment. Note $V \cup V' \subset V_{\alpha}$.
        \item An ordered pair of vertices $(v(i,j), v(k,l)) \in V_{\alpha} \times V_{\alpha}$ is
            a directed edge in $E_\alpha$ if and only if either one or both of the following is true
            \begin{list}{$\circ$}{}
                \item $ (\alpha_{\x_i}(j), \alpha_{\x_k}(l)) \in E$
                \item $(\alpha_{\y_i}(j),  \alpha_{\y_k}(l)) \in E'$.
            \end{list}
            Note $E \cup E' \subset E_\alpha$. 
        \item   The weight function~$\omega_\alpha: V_\alpha \cup E_\alpha \rightarrow \R_{\geq 0}$ is
            defined by
            \begin{equation*}
                \omega_\alpha(x) =
                    \begin{cases}
                        \omega_{V}(v(i,j)), & x = v(i,j) \in V \subset V_{\alpha}\\
                        \omega_{E}(v(i,j), v(k,l)),  & x=(v(i,j),v(k,l)) \in
                            E \subset E_{\alpha} \\
                        0 & \text{otherwise}.
                    \end{cases}
            \end{equation*}
        \item  The weight function~$\omega'_\alpha: V_\alpha \cup E_\alpha \rightarrow \R_{\geq 0}$
            is defined by
            \begin{equation*}
                \omega'_\alpha(x) =
                    \begin{cases}
                        \omega'_{V}(v((i,j)), &  x=v(i,j) \in V' \subset V_{\alpha}\\
                        \omega'_{E}(v(i,j), v(k,l)),  & x=(v(i,j),v(k,l)) \in
                            E' \subset E_{\alpha} \\
                        0 & \text{otherwise}.
                    \end{cases}
            \end{equation*}
    \end{itemize}
\end{defn}

We give an example of an extremal event supergraph and its weights in \figref{superDAG}. 

\begin{figure}[htp]
    \includegraphics[width=.8\textwidth]{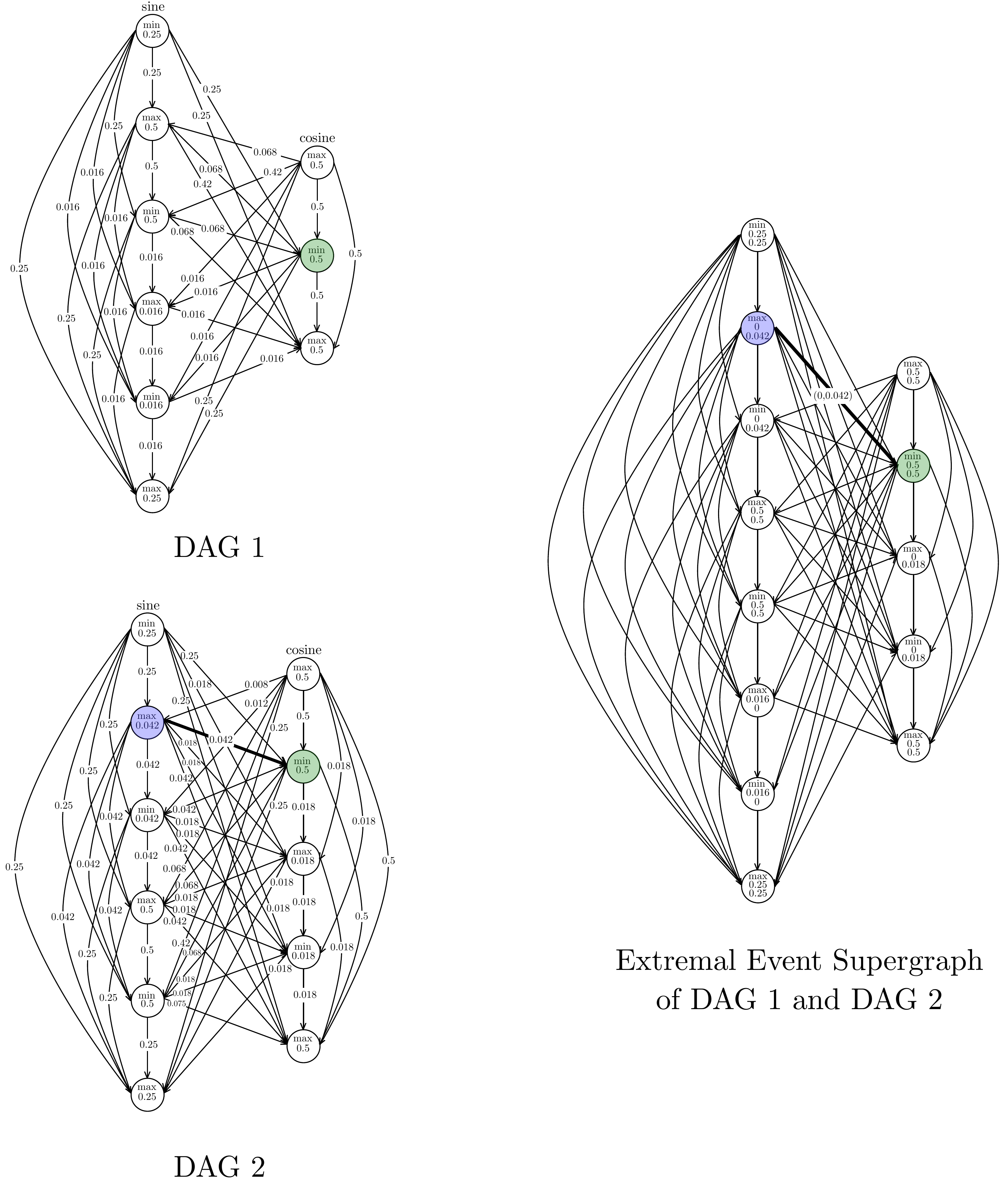}
    \caption{Extremal event supergraph of (extremal event) DAG 1 and DAG 2 from
    \figref{time-series-and-dags}. The nodes on the left represent to the
    optimal alignment between the sine backbones in DAG 1 and DAG 2. The nodes on
    the right represent to the optimal alignment between the cosine backbones
    in DAG 1 and DAG 2. The node weights are listed on the node where the upper
    node weight comes from the weight function for DAG 1 and the lower node
    weight comes from the weight function for DAG 2. The blue node in the extremal event supergraph 
    comes from aligning the blue node in DAG 2 with an insertion. The green node in the extremal event 
    supergraph comes from aligning the green nodes in DAG 1 and DAG 2. For readability, we present
    only one edge weight pair, associated to the bold edge. The edge weight on
    the left is equal to zero since the blue node in DAG 2 is aligned with an insertion. The edge weight on the 
    right is equal to 0.042 which is the edge weight between the blue and green node in DAG 2. }
    \label{fig:superDAG}
\end{figure}

We define the extremal event DAG distance to be the sum of absolute differences in
node and edge weights from the extremal event supergraph determined by the best
alignment we can easily compute.

\begin{defn}[Extremal Event DAG Distance]
    $D$ and $D'$ be two extremal event DAGs where $\x_1, \x_2, \dots \x_n$ are
    the backbones of $D$ and $\y_1, \y_2, \dots, \y_n$ are the backbones of
    $D'$. The \emph{extremal event DAG distance} is defined as:
    \[
        d_{ED}(D, D')
        =  \sum_{i=1}^n d_{\mathcal{B}}(\x_i, \y_i)
            + \inf_{\{\alpha_i\}_{i=1}^n}\sum_{(u,v)\in E_\alpha} |\omega^\alpha_{D}(u,v) - \omega^\alpha_{D'}(u,v)|.
    \]
    where $\{\alpha_i\}_{i=1}^n$ ranges over all sets of optimal alignments
    between the backbones.
    \label{def:EDD}
\end{defn}

We define extremal event DAG distance using optimal alignments between backbones
because of its computability. As we show in \secref{algorithms}, we use modified
edit distance alignment algorithms to efficiently compute backbone alignments.

An open conjecture is that the sum of differences of edge weights is minimized
only under an optimal alignment; that is,
\[
    d_{ED}(D, D')
    = \inf_{\{\alpha\}_{i=1}^n} \left(\sum_{u \in V_\alpha} |\omega^{\alpha}_{D}(u) - \omega^{\alpha}_{D'}(v)|
        + \sum_{(u,v)\in E_{\alpha}} |\omega^{\alpha}_{D}(u,v) - \omega^{\alpha}_{D'}(u,v)|\right)
\]
where $\{\alpha_i\}_{i=1}^n$ ranges over all sets of alignments between the backbones.

If this conjecture is true, then we can prove the triangle inequality for the
extremal event DAG distance using the same composition of alignments that we used for
showing the triangle inequality holds for the backbone metric. If the conjecture
is not true, then it is possible that the triangle inequality does not hold for
the extremal event DAG distance. For the biological applications we have in mind, the
key property that we desire is from a distance is \textit{stability}, which is the property that small changes in two
datasets does not cause a large jump in the distance between the associated
extremal event DAGs. We show this property holds in
\secref{stability} when the functions are ``close" to one another.

\section{Stability of Extremal Event DAGs}
\label{sec:stability}

In this section, we prove a Lipschitz stability result:
that small changes in functions that are sufficiently close result in small
distances between the corresponding extremal event DAGs.  Our results are
similar in flavor to stability for persistence
diagrams \cite{SteinerStability07}.

\subsection{Stability in Backbone Distance}

We begin by proving stability results for the backbone distance. The main result
of this section is \corref{backbone-stability}, which states the backbone
distance between backbones of two nicely tame real valued functions from a
closed interval is bounded by a constant times the $L_{\infty}$ distance
between the two functions. That is,
\[ d_{\mathcal{B}}(B(f), B(f')) \leq K \norm{f-f'}_{\infty}. \]

To get there, we show that the maximum difference in node weights arising from an optimal
alignment is bounded by the $L_\infty$ distance of the two corresponding
functions (\thmref{backbone-infinity-stability}). This leads us to comparing the
backbone distance to the following distance that looks at the maximum distance
between aligned node weights that arises from an optimal backbone alignment.

\begin{defn}[Backbone Infinity Distance]
    Let $\x$, $\y$ be backbones. We define the \emph{backbone infinity distance}
    between $\x$ and $\y$ as
    $$
        d_{\mathcal{B}_{\infty}}(\x, \y)
            = \inf_\alpha \max{|\omega_\x(\alpha_\x(i)) - \omega_\y(\alpha_\y(i))|}
    $$
    where $\alpha$ ranges over all alignments of $\x$ and $\y$.
    \label{def:backbone-infty-dist}
\end{defn}

In \appref{backboneinfty}, we prove the backbone infinity distance is a
metric. The proof is a simplified version of the proof that the backbone
distance is a metric (\propref{backbone-dist-is-dist}).

\subsubsection{Relationship Between Local Extrema, Points in Persistence Diagrams, and Backbone Nodes}

We prove backbone distance stability by moving between concepts of local extrema
of functions, points in persistence diagrams, and backbone nodes. We describe
the relationship between these three concepts next.

Let $f: C \rightarrow \R$ be a nicely tame function and $(t, f(t))$ be a local
minimum of $f$ that does not represent the essential component in $D(f)$. Recall
from \secref{node-lives} that at a height of $f(t)$ in the sublevel set
filtration, a new connected component is born. The death of this connected
component happens at the height of a local maximum denoted as $\zeta_f(t)$
(recall \defref{birth-death-map}). This implies existence of a point $(f(t), \zeta_f(t)) \in D(f)$ in the persistence diagram..
We then compute the node life, $\frac{1}{2}\pers_f(t) =
\frac{1}{2}(\zeta_f(t)-f(t))$. This shows that the node $(\text{min}, \frac{1}{2}\pers_f(t)) \in
B(f)$ is a backbone node. In summary, a local minimum $(t, f(t))$ corresponds to a point in the
persistence diagram $(f(t), \zeta_f(t)) \in D(f)$, and a vertex in a
backbone $(\text{min}, \frac{1}{2}\pers_f(t)) \in B(f)$.

Now suppose $t$ represents the essential component in $D(f)$, which means that
$t$ is a global minimum of $f$. Then, $(t, f(t))$ corresponds to $(f(t), \infty)
\in D(f)$ and $(\text{min}, \frac{1}{2}\pers_f(t)) \in B(f)$ where
$\frac{1}{2}\pers_f(t) = \frac{1}{2}(\max(f) - f(t))$.

There is the same type of correspondence for local maxima of $f$ by applying the
same process to $-f$. All the local maxima of $f$ become local minima of $-f$. To 
make the correspondence between nodes in backbones, extrema, and points in
persistence diagrams more precise, we define the following.

\begin{defn}[Truncated Persistence Points]
    Let $f:C \rightarrow \R$ be a nicely tame function. Let $(t, f(t))$ be a
    local minimum of $f$. The point $(f(t), \zeta_f(t))$ is the \emph{truncated
    persistence point} of $(t, f(t))$.
    \label{def:extrema}
\end{defn}

We often refer to the truncated persistence points as persistence points. 
We note that if $(t, f(t))$ is a local maximum of $f$, then we declare the point
$(-f(t), \zeta_{-f}(t))$ as the persistence point of $(t, f(t))$.
Because of the correspondence between extrema, persistence points, and
backbone nodes, we discuss pairings of extrema or persistence points to
get aligned pairs in backbone alignments.

\begin{figure}[]
    \centering
    \begin{subfigure}[b]{0.43\textwidth}
        \centering
        \includegraphics[width=\textwidth]{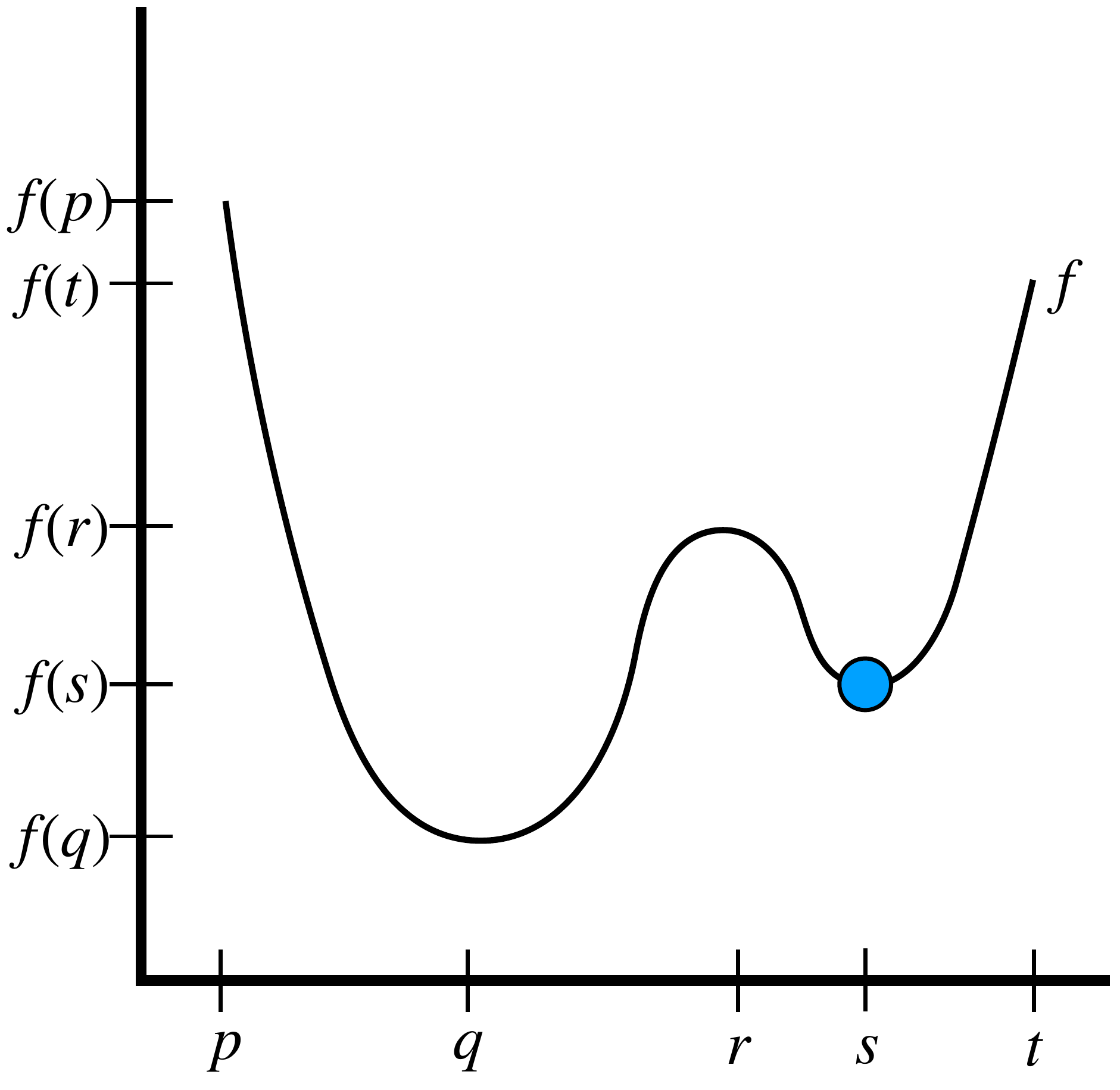}
        \caption{$f$}
        \label{fig:relationships-function}
    \end{subfigure}
    \hfil
    \begin{subfigure}[b]{0.43\textwidth}
        \centering
        \includegraphics[width=\textwidth]{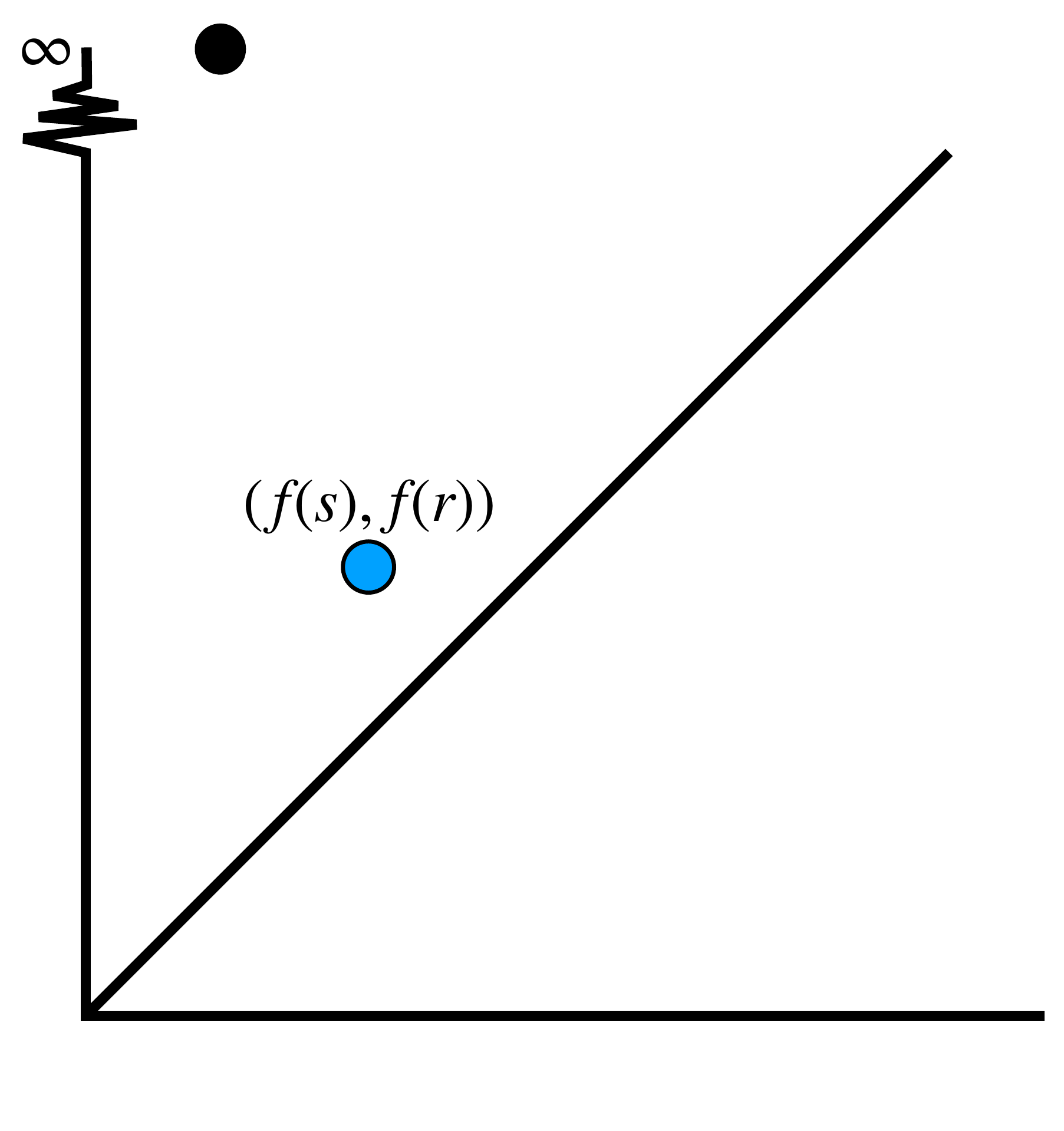}
        \caption{$D(f)$}
        \label{fig:relationships-persdgm}
    \end{subfigure}
    \hfil
    \begin{subfigure}[b]{0.1\textwidth}
        \centering
        \includegraphics[height=2.5in]{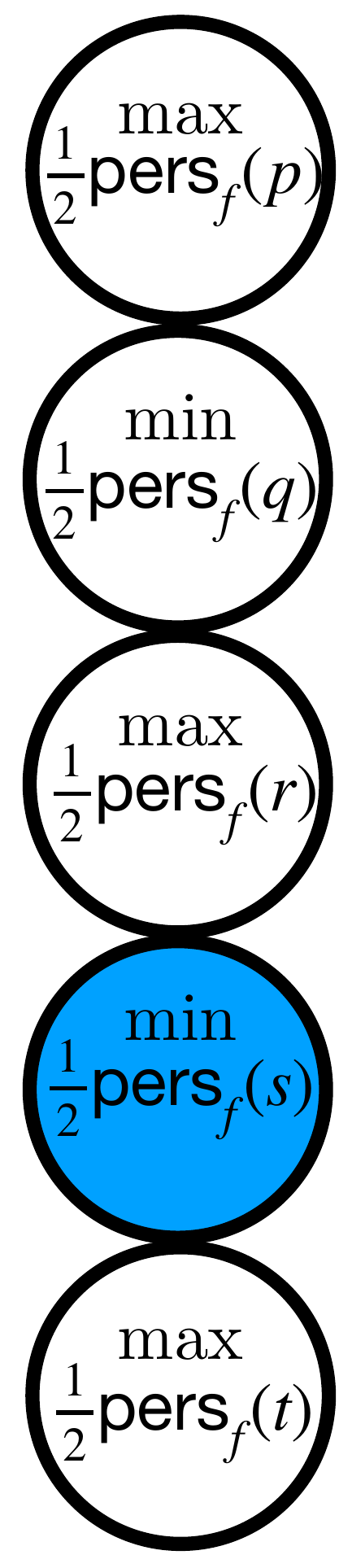}
        \caption{$B(f)$}
        \label{fig:relationships-backbone}
    \end{subfigure}
    \hfil
    \caption{Moving Between Local Minima, Persistence Diagrams, and Backbone
    Nodes. The local minimum, $(s, f(s))$ in light blue corresponds to the point
    $(f(s), f(r)) \in D(f)$ and $(\text{min}, \frac{1}{2}\pers_f(s)) \in B(f)$.}
    \label{fig:moving-between}
\end{figure}

\begin{rem}[Persistence Diagram Containing Diagonal]
Persistence diagrams are often defined as in \defref{pd}
along with a union of all points on the diagonal $\Delta = \{(x,x) \in
\R^2\}$ counted with infinite multiplicity. The addition of the diagonal is useful for
defining distances between persistence diagrams. For the remainder of this section, 
we assume that persistence diagrams contain all points on the diagonal counted with 
infinite multiplicity. This representation is useful for constructing alignments between backbones.
\label{rem:diagonal}
\end{rem}

\subsubsection{$d_{\mathcal{B}_{\infty}}$ is stable}
A key result that we use is the \emph{Box Lemma}, that is proved in \cite{SteinerStability07} to prove stability for persistence diagrams.

\begin{lem}[Box Lemma (\cite{SteinerStability07})]
    Let $X$ be a topological space,  $f,g: X\rightarrow \R$ be tame functions
    and let $\varepsilon = \norm{f-g}_{\infty}$. For $a<b<c<d$, let $R = [a,b]
    \times [c,d]$ be a box in the extended plane, $\overline{\R}^2$, and
    \mbox{$R_{\varepsilon} = [a+\varepsilon, b-\varepsilon] \times [c+\varepsilon,
    d-\varepsilon]$} be the box obtained by shrinking $R$ by $\varepsilon$ on all
    sides. Then,
    $$ \size{D(f)\cap R_{\varepsilon}} \leq \size{D(g) \cap R}.$$
\end{lem}

The next result is similar in flavor to the Easy Bijection Lemma from
\cite{SteinerStability07}.  We  first prove stability for the backbone infinity
distance in a special case. The result will depend on two constants.

\begin{defn}[Constants $\delta_{\min}$, $\delta_{\max}$ ]\label{def:constant-delta1}
    Let $X$ be a topological space and \mbox{$f:X\rightarrow \R$} be a tame function.
    Define $\delta_{\min}$ to be half of the smallest distance between two distinct  off-diagonal points, or a point in $D(f)$ and a point on the diagonal, that is,
    $$\delta_{\min} := \frac{1}{2} \min\{\norm{p-q}_\infty \mid p \in D(f)\setminus \Delta, q \in D(f), p\neq q\}.$$
   The constant $\delta_{\max}$ is defined analogously using $D(-f)$. 
\end{defn} 

Next, we note a relationship between \defref{constant-delta1} and the minimum node life of extrema of $f$.

\begin{lem}[Minimum of Node Lives is Bounded Below by $\delta_{\min}$]
Let $f:C\rightarrow \R$ be a nicely tame function. Let $\{t_i\}_{i=1}^n$ be the domain coordinates for local minima of $f$. Define $\delta$ to be half the minimum of the node lives of $t_i$, that is,
$$ \delta:= \frac{1}{2} \min \{\pers_f(t_i)\}_{i=1}^n.$$
Then $\delta_{\min} \leq \delta.$
\label{lem:min-node-lives}
\end{lem}

\begin{proof}
Let $t \in \{t_{i}\}_{i=1}^n$ such that $\frac{1}{2}\pers_f(t) = \delta$. Observe the point $(\frac{1}{2}\pers_f(t), \frac{1}{2}\pers_f(t))$ is the orthogonal projection of $(f(t), \zeta_f(t))$ onto the diagonal. In particular,\\ 
$(\frac{1}{2}\pers_f(t)+f(t), \frac{1}{2}\pers_f(t)+f(t))$ is the closest point on the diagonal to $(f(t), \zeta_f(t))$. Notice,
\begin{align*}
|\zeta_f(t)-(\frac{1}{2}\pers_f(t)+f(t))| &= |\zeta_f(t)-(\frac{1}{2}\zeta_f(t)-\frac{1}{2}f(t)+f(t))|\\
&= |\frac{1}{2}\zeta_f(t)-\frac{1}{2}f(t)|\\
&=\frac{1}{2}\pers_f(t).
\end{align*}
Furthermore,
\begin{align*}
|f(t)-(\frac{1}{2}\pers_f(t)+f(t))| &= |f(t)-(\frac{1}{2}\zeta_f(t)-\frac{1}{2}f(t)+f(t))|\\
&= |\frac{1}{2}f(t)-\frac{1}{2}\zeta_f(t)|\\
&=\frac{1}{2}\pers_f(t).
\end{align*}
Therefore,
$$\norm{(f(t), \zeta_f(t)) - (\frac{1}{2}\pers_f(t)+f(t), \frac{1}{2}\pers_f(t)+f(t))}_{\infty} = \frac{1}{2}\pers_f(t) = \delta.$$
Additionally, since for all $t_i \in \{t_i\}_{i=1}^n$ where $t_i \neq t$, we have $\frac{1}{2}\pers_f(t_i) \geq \delta$, it must be the case that
$$\norm{(f(t_i), \zeta_f(t_i)) - (\frac{1}{2}\pers_f(t_i)+f(t_i), \frac{1}{2}\pers_f(t_i)+f(t_i))}_{\infty} \geq \delta.$$
This implies that half the minimum distance between a point $p \in D(f)\setminus \Delta$ and a point on the diagonal is equal to $\delta$, that is,
$$ \frac{1}{2} \min\{\norm{p - q}_{\infty} \mid p \in D(f) \setminus \Delta, q \in \Delta \} = \delta.$$
Lastly, since 
$$ \{\norm{p - q}_{\infty} \mid p \in D(f) \setminus \Delta, q \in \Delta \} \subset \{\norm{p-q}_\infty \mid p \in D(f)\setminus \Delta, q \in D(f), p\neq q\}$$
we conclude $\delta_{\min} \leq \delta$.
\end{proof}

Using the same proof but with $-f$ and $D(-f)$, we find $\delta_{\max} \leq \delta$. We use the two constants $\delta_{\min}$ and $\delta_{\max}$ to determine when functions are ``close".

\begin{defn}[Very Close]
    Let $f:C\rightarrow \R$ be a nicely tame function. Let $\delta_f =
    \min\{\delta_{\min}, \delta_{\max}\}$. A nicely tame function $f': C\rightarrow \R$ is
    \emph{very close} to $f$ if $\norm{f-f'}_\infty < \delta_f.$
    \label{def:very-close}
\end{defn}

Next we prove an analogue of the Easy Bijection Lemma~\cite{SteinerStability07} for backbones. 
We start by constructing an alignment between two backbones arising from nicely tame functions $f$ and $f'$ where $f'$
is very close to $f$. \figref{direct-alignment} shows an example on how to
construct the direct alignment between very close functions.

\begin{con}[Direct Alignment]\label{con:direct-alignment}
    Let $f, f': C\rightarrow \R$ be nicely tame functions such that $f'$ is very
    close to $f$.
    Let~$\varepsilon = \norm{f-f'}_{\infty}$. Note, that since $f,f'$ are very close, we have $\varepsilon <
    \delta_f$ . The \emph{direct alignment}
    construction consists of two steps:
    \begin{enumerate}
        \item \orgemph{Pairing nodes in $B(f)$ with $B(f')$}.\label{step:pairing}
            Recall that each node in $B(f)$ and $B(f')$ corresponds to a local extremum
            of $f$ and $f'$, respectively.
            We begin by pairing local minima of $f$ with local minima of~$f'$.
            By definition of persistence diagrams, there is
            a one-to-one correspondence between the local minima of $f$ and the
            points in~$D(f)$. Thus, we can pair local minima of $f$ and $f'$ by pairing off diagonal
            points in $D(f)$ and $D(f')$, respectively.

            Let $p = (p_1, p_2) \in D(f)\setminus \Delta$. We describe what point in $D(f')$ is paired with $p$.  
            Since $p \in \square_{\varepsilon}(p)$, the Box Lemma tells us that the multiplicity
            $\mu(p)$ of~$p$ satisfies
            \[
                \mu(p) \leq \size{D(f') \cap \square_{\varepsilon}(p)}
                \leq \size{D(f) \cap \square_{2\varepsilon}(p)}.
            \]
            By definition of $\delta_f$ and the assumption $\varepsilon < \delta_f$, 
            we know that $p$ is the only point contained in the set $D(f) \cap
            \square_{2\varepsilon}(p)$. Therefore, $\size{D(f') \cap
            \square_{\varepsilon}(p)} = \mu(p)$. Furthermore, since $p \in D(f) \cap
            \square_{\varepsilon}(p)$, there is the same number of points, with
            multiplicity, in $D(f') \cap \square_{\varepsilon}(p)$ and in $D(f) \cap
            \square_{\varepsilon}(p)$. 
            
            We explain how to define a bijection by pairing the points in the squares $D(f) \cap \square_{\varepsilon}(p)$ and $D(f') \cap \square_{\varepsilon}(p)$. Let~$n=\mu(p)$ and
            let~$\{t_i\}_{i=1}^n$
            be the set of the domain coordinates of the local minima
            of $f$ for which $f(t_i) = p_1$. Let~$q=(q_1, q_2) \in
            D(f') \cap \square_{\varepsilon}(p)$. Observe $q_1 = f'(t)$ for some local
            minimum $(t, f'(t))$ of~$f'$. Because $p, q\in \square_{\varepsilon}(p)$, we
            have $\norm{p-q}_{\infty} \leq \varepsilon<\delta_f$. In particular,
            $$ |p_1-q_1|=|f(t_i)-f'(t)| < \delta_f \text{, for all }i \in [n]. $$
            This implies that
            $$ f(t_i)-\delta_f < f'(t) <f(t_i)+\delta_f \text{, for all }i \in [n].$$
            This inequality, the fact that $\delta_f\leq \delta_{\min}$, and $f(t_i)=p_1$ for
            all $i \in [n]$ implies
            $$t \in A := (f-\delta_f)^{-1}(-\infty, p_1+\delta_f).$$
            By \lemref{min-node-lives}, $\delta_f \leq \frac{1}{2}\min\{\pers_f(t_i)\}_{i=1}^n$. Applying Proposition 1 of \cite{BerryUsing20}, we find $A$ is a disjoint union of intervals and each contains
            exactly one~$t_i$,
            i.e., $A = \bigcup_{i=1}^n \varphi_{\delta_f}(t_i)$. Let~$t_i^*\in
            \{t_i\}_{i=1}^n$ such that $t \in \varphi_{\delta_f}(t_i^*)$. For
            our alignment, we pair the local minima $(t_i^*,
            f(t_i^*))$ with $(t, f'(t))$. Iterating this process for all points~$q\in
            D(f') \cap \square_{\varepsilon}(p)$ results in a bijection between points in 
            the squares $D(f) \cap \square_{\varepsilon}(p)$ and $D(f') \cap \square_{\varepsilon}(p)$.

            Iterating the above procedure for all points $p \in D(f)\setminus \Delta$,
            we pair all local minima of $f$ with local minima of $f'$.
            All remaining local minima of $f'$ are paired with an empty node. The order of
            the alignment is given by the domain coordinates of $f'$.

            What remains are the local maxima of $f$ and $f'$. To pair these extrema, we
            apply the same exact process to minima of $-f$ and $-f'$ since they are
            local maxima of $f$ and $f'$.

        \item \orgemph{Indexing the pairs so that order of the backbones for $B(f)$ and
            $B(f')$ are preserved.}
            Let $\x=B(f)$ and $\x'=B(f')$.
            Since extrema in $f$ (and $f'$) are in one-to-one
            correspondence with nodes in $\x$ (and $\x'$, respectively), we know
            that each pair of aligned extrema in Step~(\ref{step:pairing}) correspond
            to a pair of nodes in~$\tilde{\x} \times \tilde{\x}'$.
            To construct the \emph{direct alignment} $\alpha:[k] \rightarrow
            \tilde{\x}\times \tilde{\x}'$, we order the pairs found in Step~(\ref{step:pairing}) based on the order
            of the domain coordinates of the local extrema of $f'$. It follows
            that $k$ is the number of local extrema of $f'$.  Let
            $(t'_i, f'(t_i'))$ be the $i^{th}$ extremum based on order of domain
            coordinates of $f'$, and assume, without loss of generality, that
            this extremum is a local minimum. Then, $ \alpha(i) =
            (\alpha_{\x}(i), \alpha_{\x'}(i)) $ is given by $\alpha_{\x'}(i) =
            (\min, \frac{1}{2}\pers_{f'}(t'_i))$ and $\alpha_{\x}(i)$ is the
            paired node from Step~(\ref{step:pairing}).
            
           \end{enumerate}
           \end{con}
           
           \begin{lem}[Direct Alignment is an Alignment]
           Let $f, f': C\rightarrow \R$ be nicely tame functions such that $f'$ is very
           close to $f$. Let $\x = B(f)$ and $\x' = B(f')$. The direct alignment, $\alpha:[k] \rightarrow \tilde{\x} \times \tilde{\x}'$ constructed in \conref{direct-alignment} is an alignment.
           \end{lem}
           
            \begin{proof}
            We show \defref{alignment} holds. By \conref{direct-alignment}, we immediately see we have no null alignments, misalignments, and have a restriction to matching. Hence, \propertyref{nullalignments}, \propertyref{nomisalignments}, and \propertyref{restrictiontomatching} hold. What remains is showing the alignment preserves order of backbones. 
            
            Since the nodes of
            $\x'$ are ordered by domain coordinates of local extrema of~$f'$,
            the alignment $\alpha$ preserves the order of nodes of  $\x'$.
            We now show that the alignment $\alpha$ also preserves the order of nodes in
            $B(f)$.
            Consider nodes~$\alpha_{\x}(i)$, $\alpha_{\x}(j) \in \x=B(f)$ such that $i<j$
            and  $\alpha_{\x}(i)$,~$\alpha_{\x}(j)$ map to local extrema $(t, f(t))$ and
            $(s, f(s))$, respectively. Trivially, $t, s$ are
            contained in~$\varphi_{\delta_f}(t)$ and~$\varphi_{\delta_f}(s)$, respectively. Since $\norm{f-f'}_{\infty} < \delta_f \leq \min\{\frac{1}{2}\pers_f(t_i)\}_{i=1}^n$ by \lemref{min-node-lives}, then if $t$ and $s$ are not adjacent, $\varphi_{\delta_f}(t)\cap \varphi_{\delta_f}(s) = \emptyset$ by Proposition 1 of \cite{BerryUsing20}. Otherwise $t$ and $s$ are adjacent and by \lemref{properties}, \stmtref{nodelife-extrema-containment}, we have $t \notin \varphi_{\delta_f}(s)$ and $s \notin \varphi_{\delta_f}(t)$. Either way, we find
            $\iota(\alpha_{\x'}(i))<\iota(\alpha_{\x'}(j))$ implies
            $\iota(\alpha_{\x}(i))<\iota(\alpha_{\x}(j))$. The same argument in reverse
            can be used to show that if~$i<j$, then
            $\iota(\alpha_{\x}(i))<\iota(\alpha_{\x}(j))$. Therefore, the order of the
            backbones $B(f)$ and $B(f')$ is preserved and we have constructed an alignment between
            $B(f)$ and $B(f')$.
\end{proof}

\begin{figure}[htp]
    \centering
    \begin{subfigure}[b]{0.48\textwidth}
        \centering
        \includegraphics[width=\textwidth]{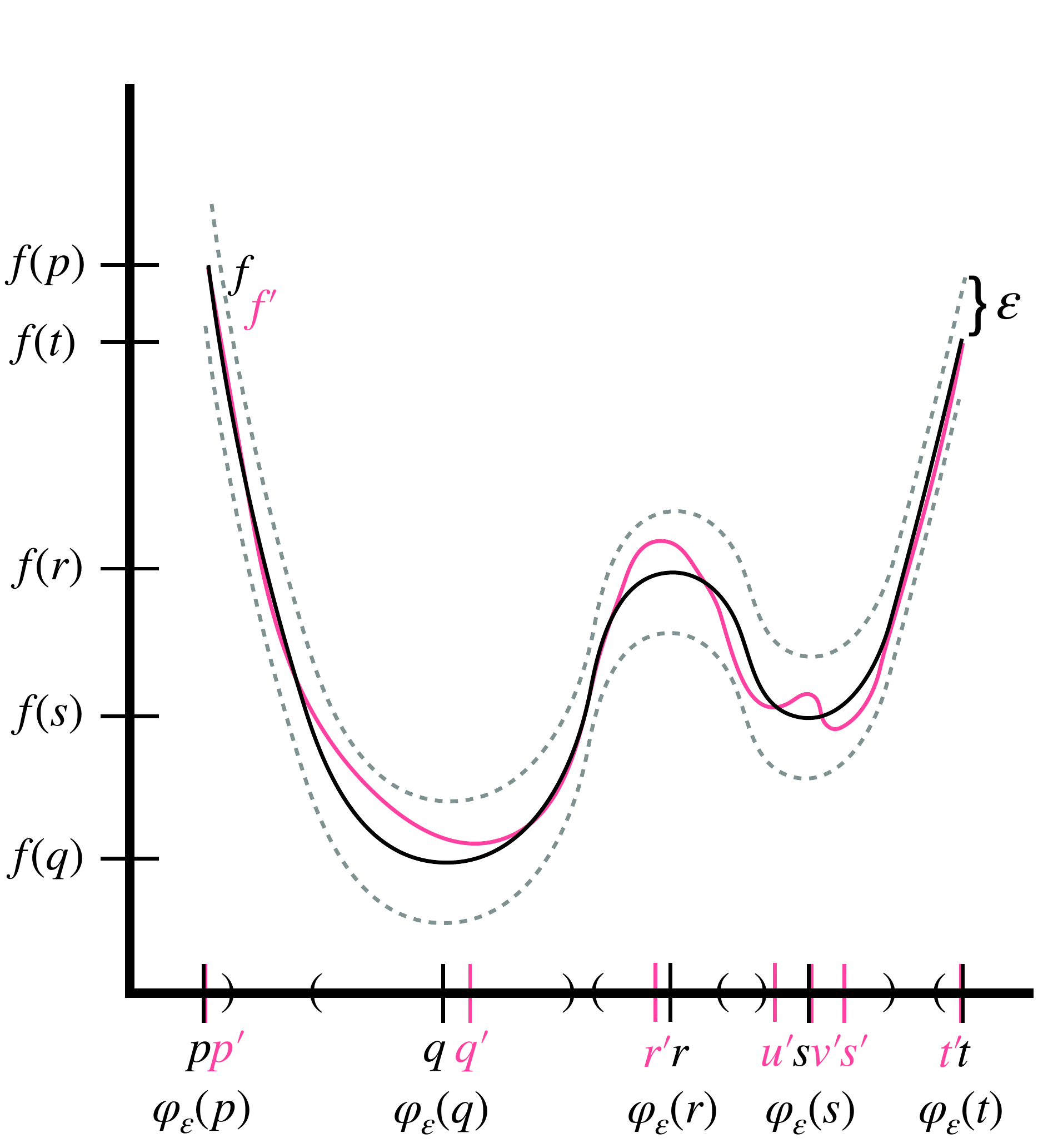}
        \caption{$f, f'$}
        \label{fig:veryclose-function}
    \end{subfigure}
    \hfil
    \begin{subfigure}[b]{0.48\textwidth}
        \centering
        \includegraphics[width=\textwidth]{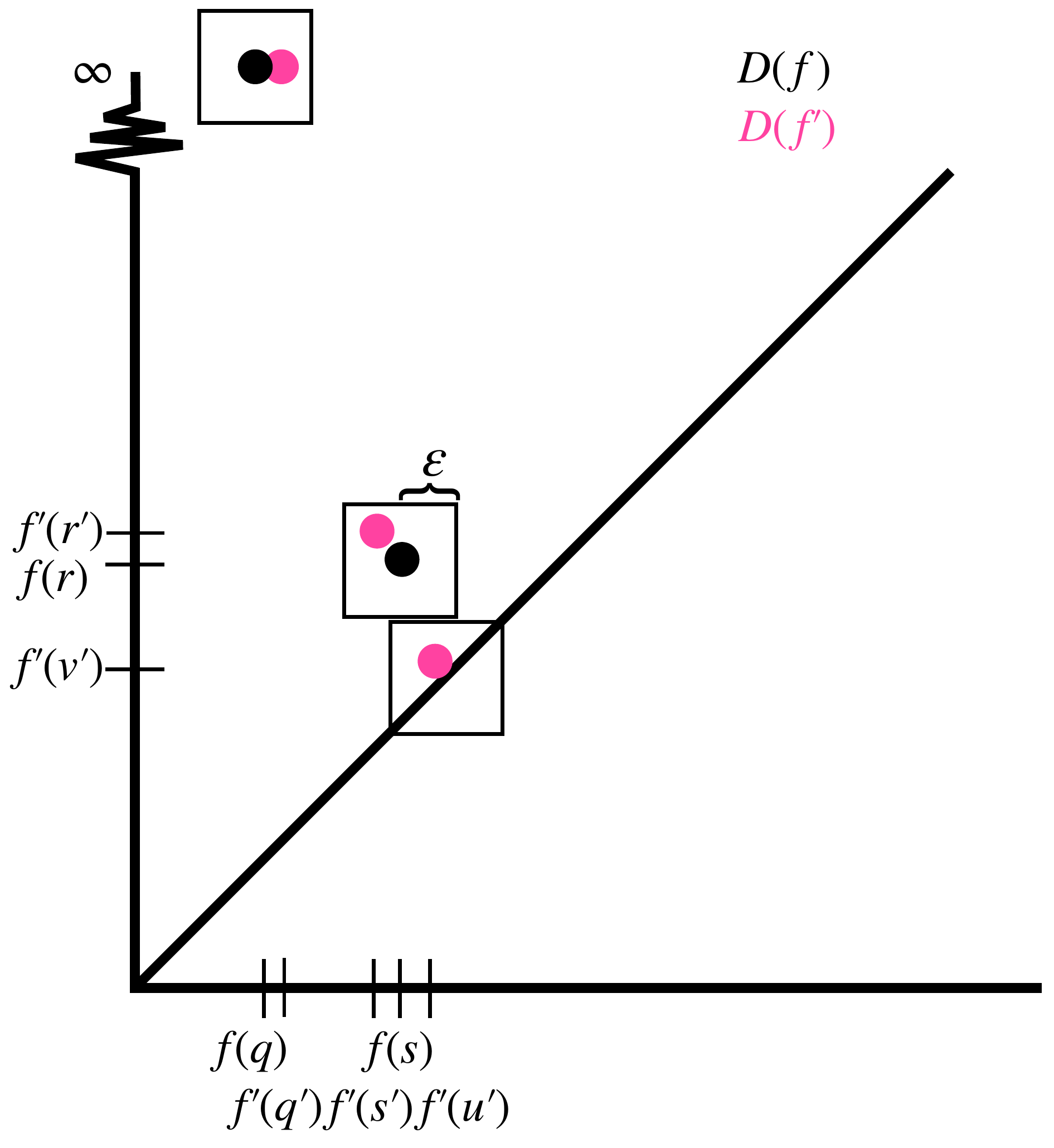}
        \caption{$D(f), D(f')$}
        \label{fig:veryclose-persdgm}
    \end{subfigure}\\
    \hfil
    \begin{subfigure}[b]{.96\textwidth}
        \centering
        \includegraphics[width=\textwidth]{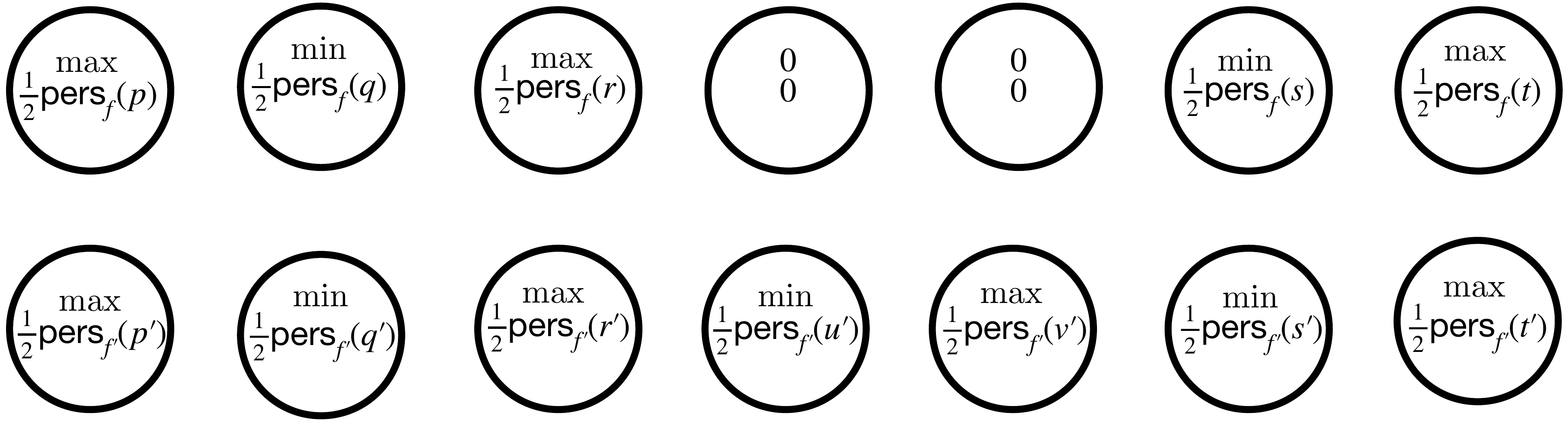}
        \caption{Direct Alignment between $B(f)$ and $B(f')$}
        \label{fig:veryclose-backbone}
    \end{subfigure}
    \hfil
    \caption{Construction of Direct Alignment. In \figref{veryclose-function}, $f$
    is the black function while $f'$ is the pink function. The black labeled
    ticks denote the domain coordinates of the local extrema of $f$ and the pink
    labeled ticks denote the domain coordinates of the local extrema of $f'$.
    The $\varepsilon$-extremal intervals for the local extrema of $f$ are
    illustrated. Since any two points in $D(f)$ where one is not a diagonal point, have a distance
    of at least $\varepsilon$, and $f'\in N_{\varepsilon}(f)$, we have $f'$ is
    very close to $f$. Applying the Direct Alignment Lemma, we get pairings of
    points in $D(f)$ and $D(f')$ as shown in \figref{veryclose-persdgm}. From
    the pairings in $D(f)$ and $D(f')$, we get pairings of nodes with the label
    ``min" that preserve order. The preservation of order comes from how the
    $\varepsilon$-extremal intervals for minima of $f$ are disjoint. We apply an analogous process to pair nodes with the label ``max".
    The alignment that is constructed
    in the Direct Alignment Lemma for $f$ and $f'$ is shown in
    \figref{veryclose-backbone}. } \label{fig:direct-alignment}
\end{figure}

Next, we prove a bound on the absolute difference between aligned weights in the
direct alignment.

\begin{lem}[Bound in Difference in Node Weights in Direct Alignment]
    Let $f,f': C\rightarrow \R$ be nicely tame functions such that $f'$ is very
    close to $f$.  Let \mbox{$\varepsilon := \norm{f-f'}_{\infty}$}. Let $\x=B(f)$,
    $\x'=B(f')$, and $\alpha:[k] \to \tilde{\x} \times \tilde{\x}'$ be the direct
    alignment as defined in \conref{direct-alignment}. Then, the absolute
    difference in weights between aligned nodes is bounded by $\varepsilon$;
    that is, for all $(x,x') \in \im(\alpha)$,
    $|w_x-w_{x'}| \leq \varepsilon$.
\label{lem:nodeweights-easyalignment}
\end{lem}

\begin{proof}

    Let $(x,x') \in \im(\alpha)$. Either both represent extrema from
    $f$ and $f'$, or the node $x'$ is the empty node.

    \orgemph{First, assume that both nodes represent extrema.} For this proof, we assume
    they are minima and note an analogous argument holds for maxima.  Let
    $(t, f(t))$ and $(t', f'(t'))$ be the local minima corresponding to nodes
    $x$ and $x'$, respectively.
    Either both of these extrema do not represent the essential component
    in~$D(f)$ and $D(f')$ or at least one of them does.
    \orgemph{Suppose neither represents
    the essential component.} Consider the persistence points~$p=(f(t),
    \zeta_f(t))$, $q=(f'(t'), \zeta_{f'}(t'))$ in $D(f)\setminus \Delta$ and $D(f')\setminus \Delta$
    respectively. From \conref{direct-alignment}, we know that
    both~$p, q \in \square_{\varepsilon}(p)$. This implies $\norm{p-q}_\infty
    \leq \varepsilon$. Hence,
    \begin{equation*}
    \begin{split}
        |w_x-w_{x'}| &= \frac{1}{2}|\pers_f(t)-\pers_{f'}(t')|\\
        &= \frac{1}{2}|(\zeta_f(t)-f(t))-(\zeta_{f'}(t')-f'(t'))|\\
        &= \frac{1}{2}|(\zeta_f(t)-\zeta_{f'}(t'))+(f'(t')-f(t))|\\
        &\leq \frac{1}{2}|\zeta_f(t)-\zeta_{f'}(t')|+\frac{1}{2}|f'(t')-f(t)|\\
        &\leq \frac{\varepsilon}{2} + \frac{\varepsilon}{2} = \varepsilon .
    \end{split}
    \end{equation*}
    \orgemph{Next, consider the case at least one of  $(t, f'(t))$ or $(t', f'(t'))$
    represents the essential component.}  By the Box Lemma, we can infer that both
    points have to represent the essential component. Then we know these points
    are global minima of $f$ and $f'$.
    Since~$(f(t), \infty)$ and $(f'(t'), \infty)$ are both contained in the
    square of radius $\varepsilon$ centered at $(f(t), \infty)$ in the extended
    plane,~$|f(t)-f'(t')|\leq\varepsilon$. Additionally, by \conref{direct-alignment}
    we know that $t_{\max}$, a global maximum of $f$, is paired with
    $t'_{\max}$, a global maximum of $f'$ such that $t'_{\max} \in
    \varphi_{\delta_f}(t_{\max})$. This implies $|f(t_{\max})-f'(t'_{\max})| \leq
    \varepsilon$. Applying the same computation as above, we see
    $|w_x-w_{x'}| \leq \varepsilon$.

    \orgemph{Lastly, consider the case $(t', f'(t'))$ is paired with an empty
    node.} 
    Consider the point $(f'(t'), \zeta_{f'}(t')) \in D(f')$. The Box Lemma implies
    $\square_{\varepsilon}((f'(t'), \zeta_{f'}(t')))$ must contain at least one
    point from $D(f)$. By assumption $(t', f'(t'))$ is paired with an empty node,
    and so $\square_{\varepsilon}((f'(t'), \zeta_{f'}(t'))$ must contain a point on
    the diagonal. Therefore, $(f'(t'), \zeta_{f'}(t'))$ is within an
    $L_{\infty}$  distance of
    $\varepsilon$ from a point on the diagonal. Because the point~$(\frac{1}{2}\pers_{f'}(t')+f(t'),
    \frac{1}{2}\pers_{f'}(t')+f(t'))$ is the orthogonal projection of $(f'(t'),
    \zeta_{f'}(t'))$ onto the diagonal, it is the closest point on the diagonal in $\R^2$ to $(f'(t'),
    \zeta_{f'}(t'))$. Therefore
    \begin{equation}\label{square}
    \norm{(f'(t'), \zeta_{f'}(t')) - (\frac{1}{2}\pers_{f'}(t')+f(t'),
    \frac{1}{2}\pers_{f'}(t')+f(t'))}_{\infty} \leq \varepsilon.
    \end{equation}
    Comparing the $x$-coordinates of pair of points in \eqref{square}, we find
    $(\frac{1}{2}\pers_{f'}(t')+f(t'))-f'(t')\leq \varepsilon$ and comparing the
    $y$-coordinates we find
    $\zeta_{f'}(t')-(\frac{1}{2}\pers_{f'}(t')+f(t'))\leq \varepsilon$.
    Observe
    $$(\frac{1}{2}\pers_{f'}(t')+f(t'))-f'(t') = \zeta_{f'}(t')-(\frac{1}{2}\pers_{f'}(t')+f(t')) = \frac{1}{2}\pers_f(t).$$
    %
    %
    We find  $\frac{1}{2}\pers_{f'}(t') \leq \varepsilon$.

    We conclude that for all paired extrema, \[ |w_x-w_{x'}| \leq \varepsilon. \]
\end{proof}

We can now prove local stability for the backbone infinity distance.

\begin{lem}[Local Backbone Infinity Stability]
    Let $f,f': C\rightarrow \R$ be nicely tame functions such that $f'$ is very close to $f$. Then,
    $$d_{\mathcal{B}_\infty}(B(f), B(f')) \leq \norm{f-f'}_{\infty}.$$
\label{lem:localbbinfty-stability}
\end{lem}

\begin{proof}
    In \lemref{nodeweights-easyalignment}, we showed that using the direct
    alignment between $B(f)$ and $B(f')$, the absolute difference in aligned
    node weights is bounded by $\norm{f-f'}_{\infty}$. Since the backbone
    infinity distance is defined by using an optimal alignment, we get
    $$d_{\mathcal{B}_{\infty}}(B(f), B(f')) \leq \norm{f-f'}_{\infty}.$$
\end{proof}

To remove the assumption that $f$ and $f'$ are very close and thus 
to globalize the backbone infinity stability result, we construct a
straight-line homotopy between $f$ and $f'$ and consider a finite number of
functions within this homotopy for which every two successive functions are very
close. For each such pair of functions, \lemref{localbbinfty-stability} applies
and we are able to apply almost the same argument as the proof of the
Interpolation Lemma in \cite{SteinerStability07}. Because we sample
functions between $f$ and $f'$ in the homotopy, we are able to conclude that
$d_{\mathcal{B}_\infty}(B(f), B(f')) \leq \norm{f-f'}_{\infty}$.

\begin{thm}[Backbone Infinity Stability]
    Let $f,f': C\rightarrow \R$ be nicely tame functions. Then,
    $$d_{\mathcal{B}_\infty}(B(f), B(f')) \leq \norm{f-f'}_{\infty}.$$
    \label{thm:backbone-infinity-stability}
\end{thm}

\begin{proof}
    Let $c:=\norm{f-f'}_{\infty}$. Define $h_{\lambda} := (1-\lambda)f+ \lambda
    f'$ where $\lambda \in [0,1]$. This is the  family of convex combinations of
    $f$ and $f'$ forms a linear interpolation between the two functions,
    starting at $h_0 = f$ and ending at $h_1 = f'$. Furthermore, we define
    $\delta(\lambda) := \delta_{h_{\lambda}}$ as in \defref{very-close}.
    Consider the open cover $U$ of $[0,1]$ by open intervals  $J_{\lambda} =
    (\lambda - \delta(\lambda)/2c, \lambda+\delta(\lambda)/2c)$ for all $\lambda
    \in [0,1]$. The compactness of $[0,1]$ implies the existence of a finite
    subcover $U'$ of $U$. Let $\lambda_1<\lambda_2<\dots<\lambda_n$ be the
    midpoints of the open intervals in $U'$. Observe, that half the length of
    $J_\lambda$ is equal to $\delta(\lambda)/2c.$ Since any two consecutive
    intervals $J_{\lambda_i}$ and $J_{\lambda_{i+1}}$ have a non-empty
    intersection,
    \begin{equation*}
    \begin{split}
        \lambda_{i+1}-\lambda_i &\leq \delta(\lambda_i)/2c +\delta(\lambda_{i+1})/2c\\
        &\leq 2\max\{\delta(\lambda_i)/2c, \delta(\lambda_{i+1})/2c\}\\
        &= \max\{\delta(\lambda_i), \delta(\lambda_{i+1})\}/c
    \end{split}
    \end{equation*}
    Furthermore, note
    \begin{equation*}
    \begin{split}
        |h_{\lambda_i}-h_{\lambda_{i+1}}|
        &= |((1-\lambda_i)f+\lambda_i f') - ((1-\lambda_{i+1})f + \lambda_{i+1}f')|\\
        &= |f(\lambda_{i+1}-\lambda_i)-f'(\lambda_{i+1}-\lambda_i)|\\
        &= \norm{f-f'}_\infty (\lambda_{i+1}-\lambda_i).
    \end{split}
    \end{equation*}
    This implies
    $$
        \norm{h_{\lambda_i}-h_{\lambda_{i+1}}}_\infty
        =c(\lambda_{i+1}-\lambda_i) \leq \max\{\delta(\lambda_i),
        \delta(\lambda_{i+1})\}.
    $$

    Therefore, $h_{\lambda_i}$ is very close to $h_{\lambda_{i+1}}$ or
    vice-versa. Either way, \lemref{localbbinfty-stability} applies and we have
    \[ d_{\mathcal{B}_\infty}(h_{\lambda_i}, h_{\lambda_{i+1}}) \leq \norm{h_{\lambda_i} - h_{\lambda_{i+1}}}_{\infty} \]
    for all $1 \leq i \leq n-1$.  Setting $\lambda_0=0$ and $\lambda_{n+1}=1$,
    we see the inequality also holds for $i=0$ and $i=n$ because
    $h_{\lambda_1}$ is very close to $h_{\lambda_0}$, and $h_{\lambda_{n+1}}$ is
    very close to $h_{\lambda_{n}}$ Therefore,
    \begin{equation*}
    \begin{split}
        d_{\mathcal{B}_{\infty}}(B(f),B(f'))
        & \leq \sum_{i=0}^{n} d_{\mathcal{B}_\infty}(B(h_{\lambda_i}), B(h_{\lambda_{i+1}}))\\
        &\leq \sum_{i=0}^{n} \norm{h_{\lambda_i}-h_{\lambda_{i+1}}}_{\infty}\\
        &= \norm{f-g}_{\infty}.
    \end{split}
    \end{equation*}
    The first inequality follows from the triangle inequality of the backbone
    infinity distance (\lemref{backbone-infty-triangle}). The
    last equality follows from how the collection $h_{\lambda_i}$ samples the
    straight line homotopy from $f$ to $f'$. Thus $d_{\mathcal{B}_\infty}(B(f),
    B(f')) \leq \norm{f-f'}_{\infty}.$
\end{proof}

\subsubsection{Backbone Stability Results}

Using \thmref{backbone-infinity-stability}, we also get stability results for
the backbone distance. If the alignment that realizes the backbone infinity
distance between two backbones, $B(f)$ and $B(f')$ is of length $K$, then the
sum of absolute differences of node weights is bounded by $K
\norm{f-f'}_{\infty}$. This is because the backbone distance is bounded by the
sum of absolute differences in node weights from the alignment realizing the
backbone infinity distance.

\begin{cor}[Backbone Stability]
    Let  $f,f': C\rightarrow \R$ be nicely tame functions. Let $K$ be the length
    of the alignment realizing the backbone infinity distance between $B(f)$ and
    $B(f')$. Then,
    $$d_{\mathcal{B}}(B(f), B(f')) \leq K\norm{f-f'}_{\infty}.$$
    \label{cor:backbone-stability}
\end{cor}

If we are unable to compute $K$, note that we can bound $K$ by the number of
extrema of $f$ plus the number of extrema of $f'$ because that is the longest
possible length of an alignment between $B(f)$ and $B(f')$.

\subsection{Local Extremal Event DAG Stability}
We showed stability between backbones. In this section, we extend those results
to the entire extremal event DAG in a local case (when $f'$ is extremely close to
$f$). We start by proving that the direct alignment for functions we call
\textit{extremely close} is the optimal backbone alignment for the backbones of
those two functions.

\begin{defn}[Extremely Close]
    Let $f:C\rightarrow \R$ be a nicely tame function. Let $\delta_f$ be as defined in \defref{constant-delta1}. 
    A nicely tame function $f': C\rightarrow \R$ is \emph{extremely
    close} to $f$ if $\norm{f-f'}_\infty < \delta_f/2.$
\end{defn}

The difference between functions that are very close and extremely close is that the 
constant is $\delta_f$ is divided by two for functions that are extremely close. 
This is needed to show that the direct alignment between two functions that are 
extremely close is the unique optimal alignment. Proving this involves several technical lemmas 
which we prove in \appref{stability}.  


\subsubsection{Bounding Differences in Aligned Node and Edge Weights}
We next prove a few lemmas that bound differences in node weights of aligned extrema.

For the rest of this subsection, we use the following assumptions and notation:

\begin{assumps}[Local Stability Assumptions]
Let  $F=\{f_i\}_{i=1}^n$ and \mbox{$F'=\{f'_i\}_{i=1}^n$} be collections of nicely tame
functions from $C$ to $\R$.
Furthermore, suppose $f_i'$ is extremely close 
to $f_i$ for each $i\in [n]$. Let $D=(V, E,\omega_V,\omega_E)$ and \mbox{$D'=(V',
E',\omega'_V,\omega'_E)$}
be the extremal event DAGs of $F$ and $F'$, respectively.  Let~$S_\alpha = (V_\alpha, E_\alpha,
    \omega_{\alpha}, \omega'_{\alpha})$ be the extremal event supergraph arising from the set of alignments $\alpha
= \{\alpha_i\}_{i=1}^n$ that is used to compute the extremal event DAG distance
between $D$ and $D'$.
\label{assumps:edd-stability}
\end{assumps}

\begin{lem}[Bound on Difference in Node Lives]
    Assume \assumpsref{edd-stability}. Let $v(i,j) \in V_{\alpha}$. Then,
    \[ |\omega_{\alpha}(v(i,j)) - \omega'_{\alpha}(v(i,j))| \leq \norm{f_i-f'_i}_{\infty}. \]
    \label{lem:node-weights-bound}
\end{lem}

\begin{proof}
    By \lemref{direct-alignment-optimal}, the alignment $\alpha_i$ between $B(f_i)$ and
    $B(f'_i)$ is the direct
    alignment. In
    \lemref{nodeweights-easyalignment}, we showed the absolute difference in node weights
    between aligned nodes is bounded by $\norm{f_i-f'_i}_{\infty}$.
    Therefore,
    \[ |\omega_{\alpha}(v(i,j)) - \omega_{\alpha}'(v(i,j))| \leq \norm{f_i-f'_i}_{\infty}. \]
\end{proof}

From \lemref{node-weights-bound}, we can conclude that if $f'$ is extremely
close to $f$, then the backbone distance between $B(f)$ and $B(f')$ is bounded
by the number of extrema in $f'$ multiplied by $\norm{f-f'}_{\infty}$. This is
because the direct alignment has a length of $B(f')$ and the absolute difference in
aligned node weights for each pair is bounded by $\norm{f-f'}_{\infty}$.
\begin{cor}[Bound on Backbone Distance for Extremely Close Functions]
    Let $f, f': C \rightarrow \R$ be nicely tame functions such that $f'$ is
    extremely close to $f$. Let $k$ be the number of extrema of $f'$. Then,
    \[ d_{\mathcal{B}}(B(f), B(f')) \leq k\norm{f-f'}_{\infty}. \]
    \label{cor:local-backbone-stability}
\end{cor}

Next we give a bound on the absolute difference in heights of aligned extrema
when every pair of functions is extremely close. In the following lemma we simplify notation as
follows:
\begin{itemize}
    \item $u_j := \alpha_{B(f_i)}(j)$
    \item $u'_j :=
            \alpha_{B(f'_i)}(j)$.
\end{itemize}
\begin{lem}[Bound on Difference in Heights of Aligned
    Extrema]\label{lem:aligned-heights-bound}
    Assume \assumpsref{edd-stability}. Let $v(i,j)$ be a vertex in $S_\alpha$.
    Let $\{ \alpha_i\}_{i=1}^n$ be the set of backbone alignments between $B(f_i)$ and
    $B(f'_i)$ that determines $S_{\alpha}$.
    Let $u_j \in B(f_i)$ and $u'_j \in B(f'_i)$. Let
    $(t, f_i(t))$ and $
    (t', f_i(t'))$ be the local extrema corresponding to $u_j$ and $u_j'$,
    respectively. Then,
    \[ |f_i(t) - f'_i(t')| \leq \norm{f_i-f'_i}_{\infty}. \]
\end{lem}

\begin{proof}
    By \lemref{direct-alignment-optimal}, $\alpha_i$ is the direct alignment for all $i \in [n]$. Recall
    from \conref{direct-alignment}, that both $(f_i(t), \zeta_{f_i}(t))$ and
    $(f'_i(t'), \zeta_{f'_i}(t'))$ are contained in the square centered at
    $(f_i(t), \zeta_{f_i}(t))$ of radius $\norm{f_i-f'_i}_{\infty}$. Hence, \mbox{$|f_i(t) -
    f'_i(t')| \leq \norm{f_i-f'_i}_{\infty}$}.
\end{proof}

We now have a bound on the maximum difference between node weights in extremal event
DAGs. What remains is bounding the difference in edge weights between extremal event
DAGs when  each pair of functions is extremely close. Let
$(v(i,k), v(j,m))$ be an edge in the extremal event
supergraph. We show
\[ |\omega_{\alpha}(v(i,k), v(j,m)) - \omega'_{\alpha}(v(i,k), v(j,m))|
    \leq  \max\{ \norm{f_i-f'_i}_\infty, \norm{f_j - f'_j}_\infty \}.
\]

For \lemref{intersection-difference-bound}, recall from \thmref{edge-weights} 
that $\varepsilon^*(t,s)$ is the infimum $\varepsilon$ for which $\varphi_{\varepsilon}(t)\cap \varphi_{\varepsilon}(s) \neq \emptyset$. 
Furthermore, we simplify notation as follows:
\begin{itemize}
    \item $u_k := \alpha_{B(f_i)}(k)$
    \item $u'_k := \alpha_{B(f'_i)}(k)$
    \item $s_m := \alpha_{B(f_j)}(m)$
    \item $s'_m := \alpha_{B(f'_j)}(m)$.
\end{itemize}
\begin{lem}[Bound on Difference of Extremal Interval Intersection Values]
    Assume \assumpsref{edd-stability}.
    Let
    $(v(i,k), v(j,m)) \in E_\alpha$ such that $i \neq j$, and all four nodes
    defining these two edges, $u_k$, $u'_k$, $s_m$, $s'_m$ are not empty nodes.
    Suppose the extrema these nodes represent are $(t,
    f_i(t))$, $(s, f_j(s))$, $(t', f'_i(t'))$, and
    $(s', f'_j(s'))$, respectively. Then,
    \[ |\varepsilon^*(t, s) - \varepsilon^*(t', s')| \leq \varepsilon_{i,j}\]
    where $\varepsilon_{i,j} := \max\{ \norm{f_i-f'_i}_\infty, \norm{f_j - f'_j}_\infty \}$.
    \label{lem:intersection-difference-bound}
\end{lem}

\begin{proof}
    Consider the case that both $(t, f_i(t))$ and $(s, f_j(s))$ are local
    minima. In the case that one or both are local maxima, we replace one, or
    both $f_i, f_j$ by the corresponding negative function and convert the
    problem to a problem about two minima.  Hence, only considering the case
    that both are local minima is
    sufficient. Additionally, we omit superscripts on $\varepsilon$-extremal
    intervals to avoid notational clutter. An input of $t$ or $s$ indicates the
    $\varepsilon$-extremal interval is computed from $f_i$ or $f_j$,
    respectively. An input of $t'$ or $s'$ indicates the $\varepsilon$-extremal
    interval is computed from $f'_i$ or $f'_j$, respectively. For convenience of exposition, let $\varepsilon_i =
    \norm{f_i-f'_i}_\infty$, $\varepsilon_j = \norm{f_j-f'_j}_\infty$, and
    \mbox{$\varepsilon_{ij}:= \max\{ \varepsilon_i, \varepsilon_j\}$}.

    Suppose $\varepsilon^*(t, s)<\varepsilon^*(t', s')$. Let $\varepsilon >
    \varepsilon^*(t, s)$. Then $\varphi_\varepsilon(t) \cap
    \varphi_\varepsilon(s) \neq \emptyset$.  By
    \lemref{aligned-heights-bound},
    \[ |f_i(t) - f'_i(t')| \leq \varepsilon_i \leq \varepsilon_{i,j}. \]
    Hence, $f_i(t) \leq f'_i(t')+\varepsilon_{i,j}$. Additionally, since
    $\norm{f_i-f'_i}_{\infty} \leq \varepsilon_{i,j}$,
    \[(f'_i-\varepsilon_{i,j})(x) \leq f_i(x) \text{ for all } x \in C. \]
    These two inequalities imply $f_i(t)+\varepsilon \leq
    f'_i(t')+\varepsilon+\varepsilon_{i,j}$ and
    $(f'_i-\varepsilon-\varepsilon_{i,j})(x) \leq f_i(x)-\varepsilon$ for all $x
    \in C$.
    Recall
    \begin{align*}
    & \varphi_{\varepsilon}(t) \text{ is the connected component of }
    (f_i-\varepsilon)^{-1}(f_i(t)+\varepsilon) \text{ containing } t, \\
    & \varphi_{\varepsilon+\varepsilon_{i,j}}(t') \text{ is the connected
    component of }
    (f'_i-\varepsilon-\varepsilon_{i,j})^{-1}(f'_i(t')+\varepsilon+\varepsilon_{i,j})
    \text{ containing } t'.
    \end{align*}
    Therefore, $\le(\varphi_{\varepsilon+\varepsilon_{i,j}}(t')) <
    \le(\varphi_\varepsilon(t))$ and
    $\re(\varphi_{\varepsilon+\varepsilon_{i,j}}(t')) > \re(\varphi_\varepsilon(t))$.
    We get $\varphi_\varepsilon(t) \subset
    \varphi_{\varepsilon+\varepsilon_{i,j}}(t')$. Similarly, we get
    $\varphi_\varepsilon(s) \subset
    \varphi_{\varepsilon+\varepsilon_{i,j}}(s')$. The non-empty intersection of
    \mbox{$\varphi_\varepsilon(t) \cap \varphi_\varepsilon(s)$} implies
    \mbox{$\varphi_{\varepsilon+\varepsilon_{i,j}}(t') \cap
    \varphi_{\varepsilon+\varepsilon_{i,j}}(s')\neq \emptyset$.} This non-empty
    intersection holds true for all $\varepsilon > \varepsilon^*(t, s)$. Since
    $\varepsilon^*(t, s)<\varepsilon^*(t', s')$, we get 
    \[ \varepsilon^*(t', s')\leq\varepsilon_{i,j}+\varepsilon^*(t, s).\] Therefore,
    \[\varepsilon^*(t', s') - \varepsilon^*(t, s) \leq \varepsilon_{i,j}.\]
    In the case $\varepsilon^*(t', s') < \varepsilon^*(t, s)$, we get
    $\varphi_\varepsilon(t') \subset \varphi_{\varepsilon+\varepsilon_{i,j}}(t)$
    and $\varphi_\varepsilon(s') \subset
    \varphi_{\varepsilon+\varepsilon_{i,j}}(s)$ by symmetry. Therefore,
    \[ \varepsilon^*(t, s) - \varepsilon^*(t', s') \leq \varepsilon_{i,j}.\]
    Combining these two cases, we conclude
    \[ |\varepsilon^*(t, s) - \varepsilon^*(t',s')| \leq \varepsilon_{i,j}. \]
\end{proof}

\begin{figure}
    \includegraphics[width=0.6\textwidth]{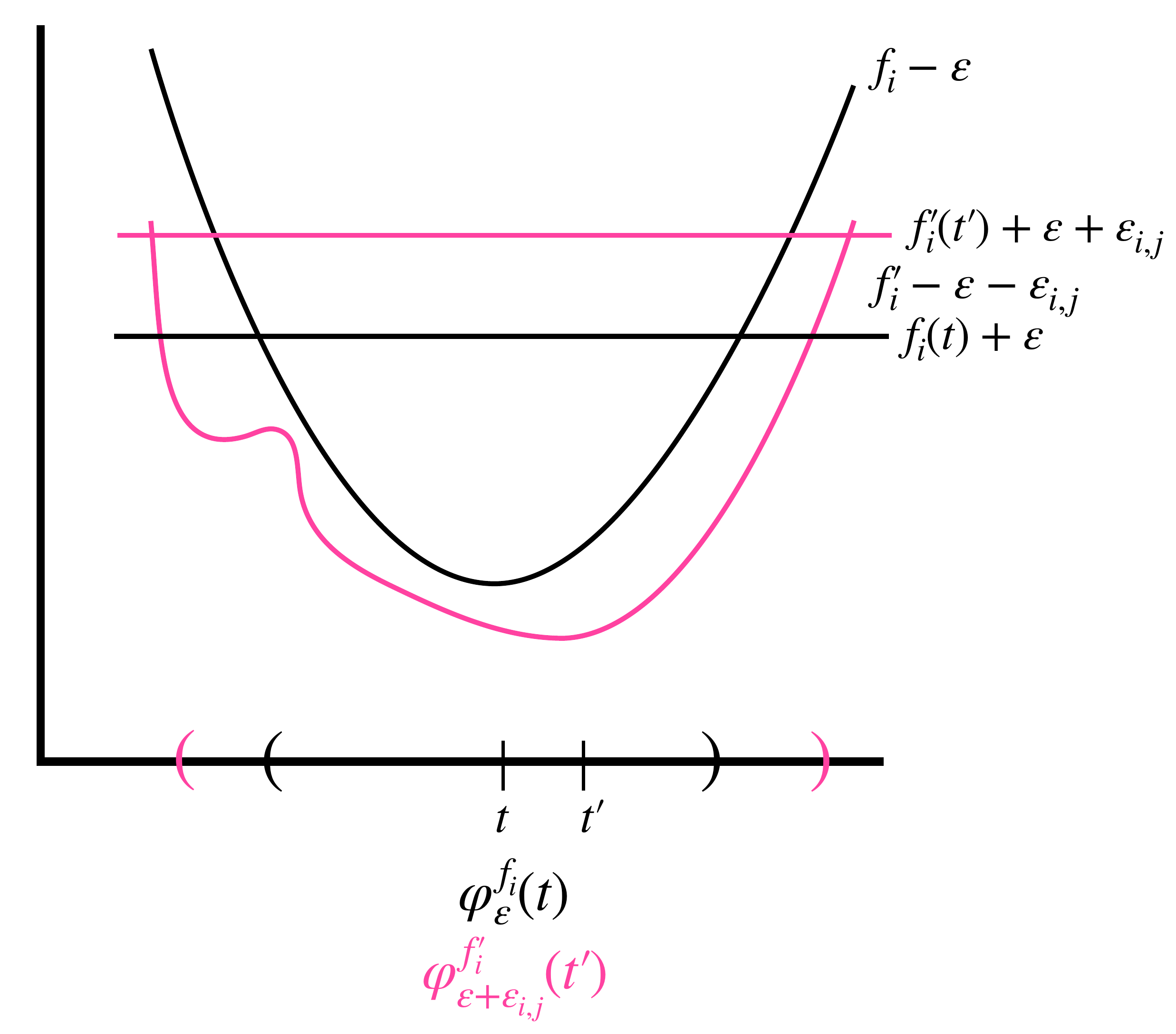}

    \caption{Nested $\varepsilon$-extremal intervals. We see that
        $\varphi_{\varepsilon}^{f_i}(t) \subset
        \varphi_{\varepsilon+\varepsilon_{i,j}}^{f'_i}(t')$. }
    \label{fig:intersection-bound}
\end{figure}

Next, we can bound the absolute difference in aligned edge weights. 

\begin{lem}[Bound on Differences in Edge Weights]\label{lem:edge-bound}
Assume \assumpsref{edd-stability}. Then,
    \[ |\omega_\alpha(v(i,k), v(j, m)) -
        \omega'_\alpha(v(i,k), v(j,m))|
        \leq  \max\{ \norm{f_i-g_i}_\infty, \norm{f_j -
        g_j}_\infty \}.
    \]
\end{lem}

\begin{proof}
    Let~$(t, f_i(t))$, $(s, f_j(s))$,
    $(t', f'_i(t'))$, and $(s', f'_j(s')) $ be the local extrema corresponding
    to nodes $\alpha_{B(f_i)}(k)$, $\alpha_{B(f_j)}(m)$, $\alpha_{B(f'_i)}(k)$,
    and~$\alpha_{B(f'_j)}(m)$, respectively. Additionally, we omit superscripts
    on~$\varepsilon$-extremal intervals to avoid notational clutter. An input of
    $t$ or $s$ indicates the $\varepsilon$-extremal interval is computed from
    $f_i$ or $f_j$, respectively. An input of $t'$ or $s'$ indicates the
    $\varepsilon$-extremal interval is computed from $f'_i$ or $f'_j$,
    respectively.

    We prove this lemma by discussing several cases. First, we assume that $\omega_\alpha(v(i,k),v(j,m))$ and
    $\omega'_\alpha(v(i,k),v(j,m))$ are non-zero. Then, by definition,
    \begin{align*}
    E_{\diff} &:= |\omega_\alpha(v(i,k), v(j,m)) - \omega'_\alpha(v(i,k), v(j,m))|\\
    &= |\min\{\frac{1}{2}\pers(t), \frac{1}{2}\pers(s), \varepsilon^*(t, s)\} - \min\{\frac{1}{2}\pers(t'),\frac{1}{2}\pers(s'), \varepsilon^*(t', s') \}|.
    \end{align*}

    Let $\varepsilon_i = \norm{f_i-g_i}_\infty$, $\varepsilon_j = \norm{f_j -
    g_j}_\infty$, and $\varepsilon_{i,j} = \max\{\varepsilon_i,
    \varepsilon_j\}$. Now we begin to go through the cases. Note that
    $E_{\diff}$ can be one of nine absolute differences depending on which value the minimum is achieved. In the cases where the
    difference comes from node weights of the same node or the extremal
    intersection values, we can apply either \lemref{node-weights-bound} or
    \lemref{intersection-difference-bound}. In all other cases we split the equality
    $ E_{\diff} = |U_1-U_2|$ into two cases $E_{\diff} = U_1-U_2$ or $E_{\diff}
    = U_2-U_1$. We replace the larger term with one of the possible values from $E_{\diff}$ so that we can apply \lemref{node-weights-bound} or \lemref{intersection-difference-bound}. For example, in one of the cases, if we assume $E_{\diff} = \frac{1}{2}(\pers(s') - \pers(t))$. Then, $\pers(s') \leq \pers(t')$. Applying \lemref{node-weights-bound}, we find
            \[ E_{\diff} = \frac{1}{2}(\pers(s') - \pers(t)) \leq \frac{1}{2}(\pers(t') - \pers(t)) \leq \varepsilon_i \leq \varepsilon_{i,j}. \]
    All together, we have 15 cases. We explicitly write these 15 cases out in \appref{bound-edges}.
Based on these bounds, we conclude in the case that $\omega_\alpha(v(i,k),v(j,m))$ and
    $\omega'_{\alpha}(v(i,k),v(j,m))$ are non-zero, $E_{\diff} \leq \varepsilon_{i,j}$.

    Now assume one of $\omega_\alpha(v(i,k), v(j,m))$ or
    $\omega'_{\alpha}(v(i,k), v(j,m))$ is equal to zero. Without loss of
    generality, suppose $\omega'_\alpha(v(i,k), v(j,m)) = 0$. Then,
    \[ E_{\diff} = \min\{ \frac{1}{2}\pers(t), \frac{1}{2}\pers(s), \varepsilon^*(t, s) \}. \]
    Applying \lemref{node-weights-bound}, we find
    \[ \frac{1}{2}\pers(t) \leq \varepsilon_i \leq \varepsilon_{i,j}, \quad \frac{1}{2}\pers(s)
        \leq \varepsilon_j \leq \varepsilon_{i,j}. \]
    If $E_{\diff} = \varepsilon^*(t, s)$, then $\varepsilon^*(t, s) \leq \frac{1}{2}\pers(t)
    \leq \varepsilon_i \leq \varepsilon_{i,j}$. Hence, $E_{\diff} \leq
    \varepsilon_{i,j}$.

    Combining all the cases, we can conclude
    \[ |\omega_{\alpha}(v(i,k), v(j,m)) - \omega'_{\alpha}(v(i,k), v(j,m))|
    \leq \varepsilon_{i,j}.\]
\end{proof}

Using the bounds we established between aligned node and edge weights in
the extremal event supergraph arising from extremely close functions, we can
bound the extremal event DAG distance.

\begin{thm}[Extremal Event DAG Stability]
    Assume \assumpsref{edd-stability}.
    Let $k_i$ be the number of extrema in $f'_i$. Let
    $\varepsilon_i := \norm{f_i-f'_i}_{\infty}$ and $\varepsilon_{i,j} :=
    \max\{\varepsilon_i, \varepsilon_j\}$. Let $P$ be the set of unordered pairs
    between the first $n$ positive integers.  Let $S_{\alpha}|_{\alpha_i}$ be the 
    restricted subgraph of $S_{\alpha}$ that is induced by $\alpha_i$. Furthermore, let $i\neq j$ and
    denote $E_{i,j}$ to be the set of cross edges in $S_{\alpha}$, that is,
    $(u,v) \in E_{i,j}$ if $u \in S_{\alpha}|_{\alpha_i}$ and $v \in
    S_{\alpha}|_{\alpha_j}$. Then,
    \[ d_{ED}(D, D')
        \leq \sum_{i=1}^n k_i\varepsilon_i+\sum_{i=1}^n\binom{k_i}{2}\varepsilon_i
        + \sum_{(i,j) \in P} |E_{i,j}| \varepsilon_{i,j}. \]
    \label{thm:extremal-DAG-stability}
\end{thm}

\begin{proof}
    Let $E_i$ be the set of edges in the extremal event supergraph restricted to the
    subgraph induced by $\alpha_i$. Hence,
    for each edge $(u,v) \in E_i$, we have $u$, $v \in S_{\alpha}|_{\alpha_i}$.  The
    extremal event DAG distance between $D$ and $D'$ can be expressed as
    the sum of three terms.

    \begin{equation*}
    \begin{split}
        d_{ED}(D, D')
            &= \sum_{i=1}^n d_{\mathcal{B}}(B(f_i), B(f'_i)) + \sum_{(u,v) \in
    E_i}|\omega_\alpha(u,v) - \omega'_\alpha(u,v)|\\
        &\quad + \sum_{(u,v) \in E_{i,j}}|\omega_\alpha(u,v) - \omega'_\alpha(u,v)|.
    \end{split}
    \end{equation*}

    The first term is the sum of backbone distances between $B(f_i)$ and
    $B(f'_i)$ for each $i \in [n]$. The second term is the sum of the absolute
    difference in edge weights where the nodes defining each edge are from the
    same backbone alignment. The third term is the sum of the absolute
    difference in edge weights where the nodes defining each edge are contained
    in different backbones.

    Applying \corref{local-backbone-stability}, we know that
    $d_{\mathcal{B}}(B(f_i), B(f'_i)) \leq k_i \varepsilon_i$. Hence, we can
    bound the first term
    \[ \sum_{i=1}^n d_{\mathcal{B}}(B(f_i), B(f'_i)) \leq \sum_{i=1}^n
    k_i\varepsilon_i.\]

    Applying \lemref{edge-bound}, and noting both $u, v \in S_{\alpha}|_{\alpha_i}$, we know that if $(u,v) \in E_i$, then
    $|\omega_\alpha(u,v) - \omega'_\alpha(u,v)| \leq
    \varepsilon_i$. There are $\binom{k_i}{2}$ edges in $E_i$. Hence, we can
    bound the second term by
    \[\sum_{(u,v) \in E_i}|\omega_\alpha(u,v) - \omega'_\alpha(u,v)|
        \leq \sum_{i=1}^n \binom{k_i}{2}\varepsilon_i. \]

    Let $|E_{i,j}|$ be the cardinality of the set $E_{i,j}$. Applying
    \lemref{edge-bound}, we can bound the third term
    \[ \sum_{(u,v) \in E_{i,j}}|\omega_\alpha(u,v) - \omega'_\alpha(u,v)|
        \leq \sum_{(i,j) \in P} |E_{i,j}| \varepsilon_{i,j}. \]

    Combining the three bounds we find that

    \[ d_{ED}(D, D')
        \leq \sum_{i=1}^n k_i\varepsilon_i+\sum_{i=1}^n\binom{k_i}{2}\varepsilon_i
        + \sum_{(i,j) \in P} |E_{i,j}| \varepsilon_{i,j}.  \]
\end{proof}

Note that Theorem~\ref{thm:extremal-DAG-stability} requires that in collections
$\{f_i\}_{i=1}^n$ and $\{ f_i'\}_{i=1}^n$ each pair $(f_j, f_j')$ is extremely
close. Therefore, this is a local stability result.
\thmref{backbone-infinity-stability} offers an approach to use a local
result (Lemma~\ref{lem:localbbinfty-stability}) to prove a global result. This
globalization approach uses a homotopy that is sampled sufficiently densely so
that each consecutive pair satisfies the assumptions of the local result.
However, the key ingredient used to aggregate the local results to a global
estimate is a  triangle inequality. It remains an important open question
whether extremal event DAG distance satisfies the triangle inequality. If so,
than a similar globalization process would yield a version of Theorem
~\ref{thm:extremal-DAG-stability} without restrictions on closeness of $f_j$ and
$f_j'$.

\section{Applications}
\label{sec:applications}

We apply the extremal event DAG construction and distance to two applications:
(1) quantifying similarity in replicate experiments of microarray yeast cell cycle
data and (2) providing quantitative evidence that an intrinsic oscillator drives the
blood stage cycle of the malaria parasite \textit{Plasmodium falciparum}. We discuss how to apply 
the extremal event DAG construction and distance to the discrete setting in the supplementary materials.
These two datasets were analyzed in \cite{BerryUsing20} and \cite{SmithAn20}, respectively, 
using a directed maximal common edge subgraph (DMCES) metric that compared $\varepsilon$-DAGs (recall $\varepsilon$-DAGs described after \defref{extremal-DAG})
with a sequence of fixed $\varepsilon$. Because of the computational complexity of computing the 
DMCES metric, these calculations were done on a limited number of time series with a limited number 
of extrema per time series, i.e., on less noisy data. Additionally, since the $\varepsilon$-DAGs specify 
a value of a parameter $\varepsilon$, the experiments were performed over a range of $\varepsilon$ 
between $0$ and $0.15$. The construction of extremal event DAGs does not need the value of 
$\varepsilon$ to be specified. Additionally, both the extremal event DAG construction and distance 
can be computed much more quickly than the DMCES metric which has an exponential time complexity. 
This all means we can compute distances over much larger sets of genes in a significantly 
shorter amount of time.

\subsection{Yeast Cell Cycle Data}

The first dataset consists of microarray time series transcriptomics from the
yeast \textit{Saccharomyces cerevisiae},  published in \cite{OrlandoGlobal08}. The yeast cell cycle is 
well studied and has experimental validation \cite{GuntherMandatory14, ChoThe19, HaaseTopology14, SimonSerial01}.
The amplitude of the data has been normalized between -0.5 and 0.5 and its phase
has been shifted by alignment using CLOCCS analysis,  see ``Appendix A: Yeast
Data Analysis" in \cite{BerryUsing20}. Using the CLOCCS analysis, the replicate
experiments were aligned so that the time series start at the same point in the
yeast cell cycle. Furthermore, the data were truncated to one period so that the
data analysis focuses on the extrema from a synchronized cell population, since
the production of daughter cells causes increasing levels of cell division
asynchrony that reduces the periodic signal.
We analyze two collections of time series data $\mathcal{D}_1$ and
$\mathcal{D}_2$ that each consists of 16 genes and 265 time points.

We perform three different comparison computations:

\begin{enumerate}

    \item We focus on a subset of $\mathcal{D}_1$ and $\mathcal{D}_2$ that
        consists of the time series for four genes: SWI4, YOX1, NDD1, and
        HCM1. We denote these sub-datasets as $\mathcal{D}_1'$ and
        $\mathcal{D}_2'$ respectively. We then compute the extremal event DAG distance
        between the extremal event DAGs of $\mathcal{D}_1'$ and $\mathcal{D}_2'$,
        $$d_{ED}(\DAG(\mathcal{D}_1'), \DAG(\mathcal{D}_2')).$$

    \item We consider dataset $\mathcal{D}_2'$ but switch labels between time series for $CLB2$ and $YOX1$. We call this
        mislabeled dataset $\mathcal{D}'_3$. Then we compute
        $$d_{ED}(\DAG(\mathcal{D}'_1), \DAG(\mathcal{D}'_3)).$$ The comparison
        between $d_{ED}(\DAG(\mathcal{D}_1'), \DAG(\mathcal{D}_2'))$ to
        $d_{ED}(\DAG(\mathcal{D}'_1), \DAG(\mathcal{D}'_3))$ indicates the impact of the replacement of one time series by another on  the extremal event DAG distance.

    \item Lastly, we assess the distance between the full datasets $\mathcal D_1$ and $\mathcal D_2$ by constructing a baseline distribution for the expected distance. We do this by first by
        scrambling the  gene names in $\mathcal{D}_2$  to create a dataset
        $\hat{\mathcal{D}}_2$. We then compute $$d_{ED}(\DAG(\mathcal{D}_1),
        \DAG(\hat{\mathcal{D}}_2)).$$ We repeat this computation 100 times for 100 random name assignments. This
        experiment gives us an idea on the range of possible distances between $\mathcal{D}_1$ and $\mathcal{D}_2$. We then compare this distribution to the actual distance
        $$d_{ED}(\DAG(\mathcal{D}_1), \DAG(\mathcal{D}_2)).$$

\end{enumerate}

Since the extremal event DAG distance can be any non-negative number, it can be
difficult to discern how similar $\mathcal{D}_1$ and $\mathcal{D}_2$ are solely
based on computing $d_{ED}(\DAG(\mathcal{D}_1), \DAG(\mathcal{D}_2)).$ To gain a
better understanding of how similar $\mathcal{D}_1$ and $\mathcal{D}_2$ are, we
perform computation (3) to get a baseline distribution of distances between the
time series in $\mathcal{D}_1$ and time series in $\mathcal{D}_2$. This
distribution can then be used as a null hypothesis $H_0$ for testing $H_1$ that
$\mathcal{D}_1$ and $\mathcal{D}_2$ measure gene expression in the identically
behaving cell in the same environmental condition.

\vspace{1ex}

\noindent \textbf{Computations 1 \& 2.}
We computed
\begin{align*}
    d_{ED}(\DAG(\mathcal{D}_1'), \DAG(\mathcal{D}_2'))) &= 10.34\\
    d_{ED}(\DAG(\mathcal{D}_1'), \DAG(\mathcal{D}_3'))) &= 15.48.
\end{align*}
The mismatched gene dataset causes a 50\% increase in distance even though only 25\% of
the dataset was perturbed, a substantial change. This result is consistent with
the result from numerical experiment 3 in \cite{BerryUsing20} where the same data were
analyzed using $\varepsilon$-DAGs and the DMCES metric. Specifically, DMCES
similarity was computed between the $\varepsilon$-DAGs at $\varepsilon$ values
ranging between 0 and 0.15. A similarity score of one indicated that the
$\varepsilon$-DAGs are equal whereas a similarity score of 0 indicates the
$\varepsilon$-DAGs are very dissimilar. The similarity ranged between 0.7 and 1
for $\mathcal{D}_1'$ and $\mathcal{D}_2'$, whereas the similarity ranged between
0.4 and 0.6 for $\mathcal{D}_1'$ and $\mathcal{D}_3'$. The same qualitative conclusion can be drawn from our results and the earlier work; namely that replacing one time series with another decreases similarity, or increases distance, between datasets.

\vspace{1ex}

\noindent \textbf{Computation 3.}
After computing $d_{ED}(\DAG(\mathcal{D}_1)), \DAG(\hat{\mathcal{D}}_2))$ 100
times we get the distribution of distances shown in \figref{yeast-scramble} with
the following statistics
\begin{itemize}
    \item Maximum Distance = 384.30
    \item Mean Distance = 341.90
    \item Minimum Distance = 273.95
    \item Standard Deviation = 23.14
\end{itemize}
We found $d_{ED}(\DAG(\mathcal{D}_1), \DAG(\mathcal{D}_2)) = 150.44$.
 \begin{figure}
     \includegraphics[width=.8\textwidth]{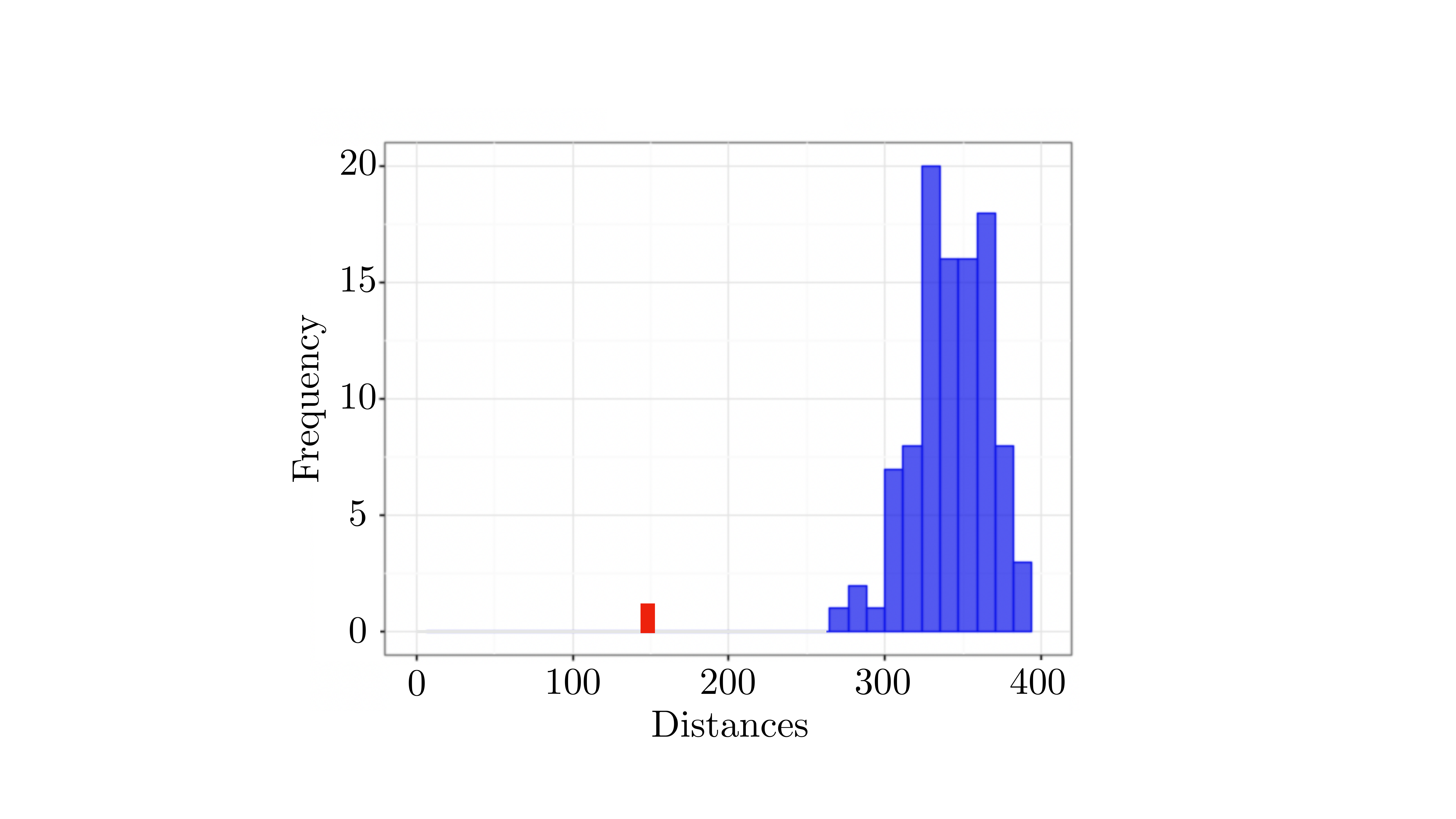}
        \caption{Extremal Event DAG Distances in Experiment 3. The red bar is $d_{ED}(\DAG(\mathcal{D}_1), \DAG(\mathcal{D}_2))$, the distance between the two yeast datasets without any scrambling of genes.}
        \label{fig:yeast-scramble}
 \end{figure}
Therefore  $d_{ED}(\DAG(\mathcal{D}_1),
\DAG(\mathcal{D}_2))$ is roughly eight standard deviations (see \figref{yeast-scramble}) below the mean of the
estimated null distribution, which suggests that there is a significant amount
of similarity between $\mathcal{D}_1$ and $\mathcal{D}_2$.

In numerical experiment 4 in~\cite{BerryUsing20}, the same goal of measuring replicate similarity was approached using a different technique. Subsets of 4 and 8 genes out of the 16 total were used to construct $\varepsilon$-DAGs over the range of $\varepsilon \in [0,0.15]$. No baseline was calculated, but the computations showed a mean similarity that was usually over the relatively high value of 80\% and frequently over 90\%.

Overall, the results of these three computations are consistent with the numerical
experiments in \cite{BerryUsing20}. We offer the significant improvement that we did
not need to run each of these computations over a range of $\varepsilon$ because
the extremal event DAG construction does not depend on $\varepsilon$. Because of the vast increase in computational efficiency, we were able to perform computation 3 using all 16 genes in the two datasets instead of randomly sampled subsets. It was shown in~\cite{BerryUsing20} that using subsets of increasing size decreased variance in the subsampled computations rather than substantially changing mean performance, which was taken to be evidence that subsampling was a good proxy for the distance between $\mathcal D_1$ and $\mathcal D_2$. In other words, if the DMCES method could have computed the distance between $\mathcal D_1$ and $\mathcal D_2$, then increasing the size of the subsamples would show convergence to that value. We offer a method that does not require a convergence argument, but rather can compute the value directly.

\subsection{Malaria Parasite Data} \label{sec:parasite-data}

In this application, we seek to show that oscillatory genes in strains of
\textit{Plasmodium falciparum} maintain as much of a phase ordering as well known
circadian genes across various mouse tissues. This provides circumstantial
evidence that malaria parasites have an internal clock that is at least as
conserved as that of the circadian oscillator, as shown in~\cite{SmithAn20}.
The mouse data comes from \cite{ZhangA14} that contains the circadian
transcriptomes of 12 mouse organ tissues every two hours for 48 hours.
In~\cite{SmithAn20}, similarity between datasets within malaria or mouse was
determined by choosing ``reference datasets'' to which the other strains and
tissues were compared, as opposed to computing all pairwise comparisons. The 3D7
strain and the liver tissue were chosen as the reference datasets in \textit{Plasmodium
falciparum} and the mouse tissues respectively.

In both collections of time series, subsets of periodic genes were selected that peak at similar times across parasite strains or mouse tissues. These genes are called ``in-phase" subsets. After this subset of genes was found, all datasets were interpolated with piecewise cubic Hermite interpolating polynomial spline to one hour intervals. The data were wrapped so that there could be a common starting point. Furthermore, \textit{Plasmodium falciparum} was down-sampled by removing every odd datapoint. More details on the experiments and data preprocessing can be found in the supplementary materials of \cite{SmithAn20}. After the pre-processing steps, our time series contain 119 parasite genes and 107 mouse genes.

To gain a better understanding of how similar our datasets are, we
created a baseline distribution for each strain or tissue by randomly interchanging gene names and shifting each time series by a random amount, using the same random shifts as in \cite{SmithAn20}. Specifically, if we view the values of a time series as an ordered list $h_1, h_2, \dots , h_n$ of length $n$, we randomly select $1\leq m \leq n$ to create a new shifted time series \[h_m, h_{m+1}, \dots, h_n, h_1, h_2, \dots, h_{m-1}.\]
The phase shift operation preserves characteristics of the dataset except for the ordering of
extrema. For each strain and tissue, the baseline distance was computed between the unpermuted and unshifted reference dataset and the permuted and shifted datasets.

We use the following notation for the parasite datasets:

\begin{itemize}
\item Let $\mathcal{D}_1$ be the collection of time series from strain 3D7. This is the reference dataset.
\item Let $\mathcal{D}_2$ be the collection of time series from strain FVO-NIH.
\item Let $\mathcal{D}_3$ be the collection of time series from strain SA250.
\item Let $\mathcal{D}_4$ be the collection of time series from strain D6.
\item Let $\mathcal{D}_2'$ be the collection of shifted and permuted time series of $\mathcal{D}_2$.
\item Let $\mathcal{D}_3'$ be the collection of shifted and permuted time series of $\mathcal{D}_3$.
\item Let $\mathcal{D}_4'$ be the collection of shifted and permuted time series of $\mathcal{D}_4$.
\end{itemize}

We perform the following computations to study the parasite data.
\begin{enumerate}
    \item Pick 1500 random subsets of 15 genes. With the arbitrary choice of 15
        genes fixed, let $\hat{\mathcal{D}_i}$ denote the subset of the
        corresponding time series from $\mathcal{D}_i$ for $i \in \{1,2,3, 4\}$.
        Compute $$d_{ED}(\DAG(\hat{\mathcal{D}}_1), \DAG(\hat{\mathcal{D}}_j))$$
        for each random subset and each $j \in \{2,3,4\}$.
    \item  Pick a new set of 1500 random subsets of 15 genes, where $\hat{\mathcal{D}_i}$ for $i \in \{1,2,3, 4\}$ is defined as above for each subset. Let $\hat{\mathcal{D}'_j}$
        denote the subset of the shifted and permuted time series $\mathcal{D}'_j$ for $j \in \{2,3, 4\}$.
        Compute $d_{ED}(\DAG(\hat{\mathcal{D}}_1), \DAG(\hat{\mathcal{D}}_j'))$
        for each random subset and $j \in \{2,3,4\}$.
\end{enumerate}

We do analogous experiments to study the mouse gene data. We
use the following notation.

\begin{itemize}
\item Let $\mathcal{D}_a$ be the collection of time series recorded from liver tissue. This is the reference dataset.
\item Let $\mathcal{D}_b$ be the collection of time series recorded from kidney tissue.
\item Let $\mathcal{D}_c$ be the collection of time series recorded from lung tissue.
\item Let $\mathcal{D}_b'$ be the collection of shifted and permuted time series of $\mathcal{D}_b$.
\item Let $\mathcal{D}_c'$ be the collection of shifted and permuted time series of $\mathcal{D}_c$.
\end{itemize}

We perform the following experiments to study the mouse tissue data.

\begin{enumerate}
    \item Pick 1500 random subsets of 15 genes. Let $\hat{\mathcal{D}_i}$ denote
        the subset of the corresponding time series from $\mathcal{D}_i$ for $i
        \in \{a,b, c\}$. Compute
        $d_{ED}(\DAG(\hat{\mathcal{D}}_a), \DAG(\hat{\mathcal{D}}_j))$ for each
        random subset and $j \in \{b, c\}$.
    \item  Pick a new set of 1500 random subsets of 15 genes, where $\hat{\mathcal{D}_i}$ for $i \in \{1,2,3, 4\}$ is defined as above for each subset. Let $\hat{\mathcal{D}'_j}$
        denote the subset of the shifted and permuted time series $\mathcal{D}'_j$ for $j \in \{b, c\}$.
        Compute $d_{ED}(\DAG(\hat{\mathcal{D}}_a), \DAG(\hat{\mathcal{D}}_j'))$
        for each random subset and $j \in \{b,c\}$.
\end{enumerate}

We summarize the results of the experiments below and provide more statistics in
\appref{tables}. In \figref{parasite-histograms} and
\figref{mouse-histograms}, we see in general that the baseline distances are
larger than the distances between cell lines and reference cell line.

\begin{figure}[h]
    \centering
    \begin{subfigure}[b]{0.32\textwidth}
        \centering
        \includegraphics[width=\textwidth]{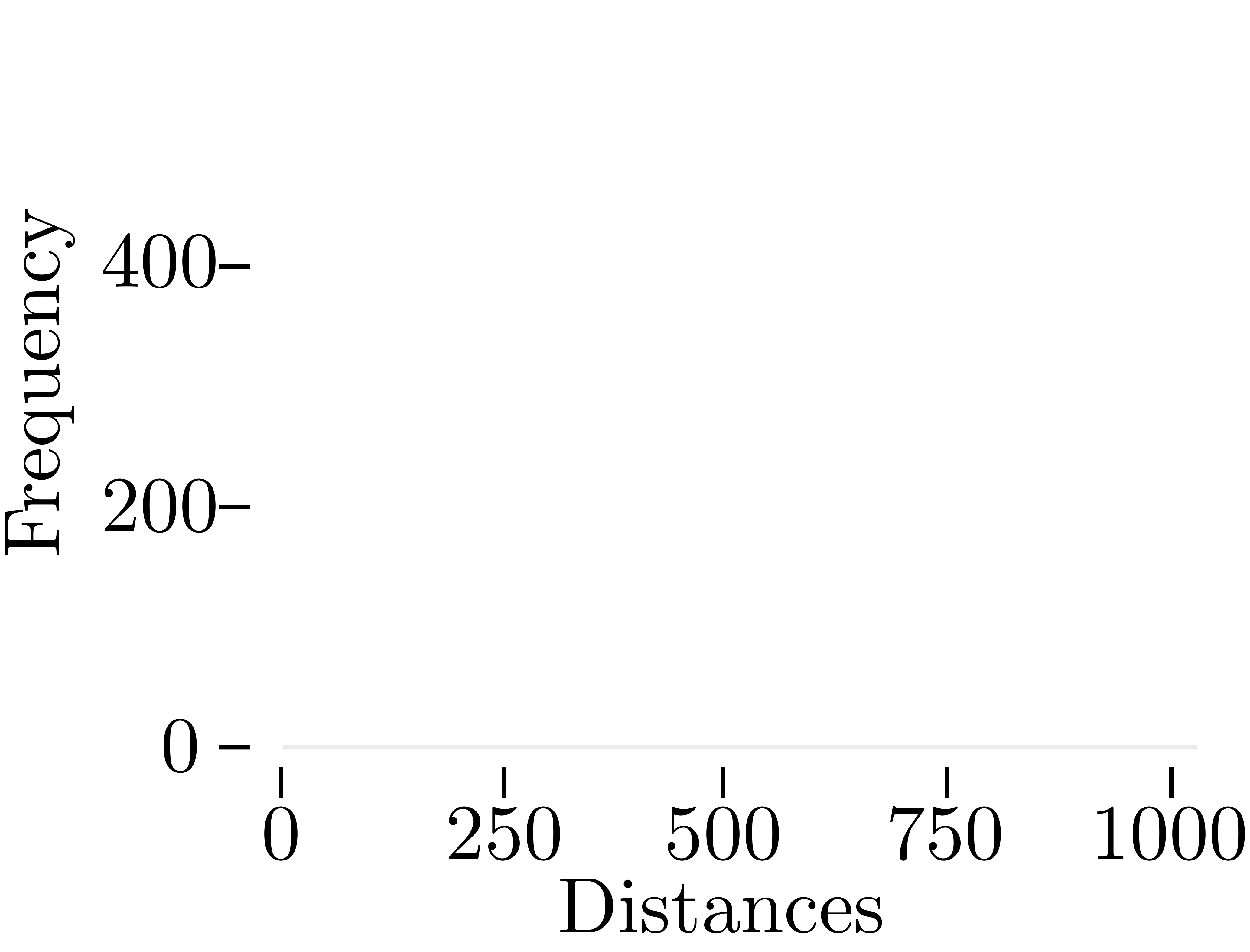}
        \caption{D6}
        \label{fig:D6}
    \end{subfigure}
    \begin{subfigure}[b]{0.32\textwidth}
        \centering
        \includegraphics[width=\textwidth]{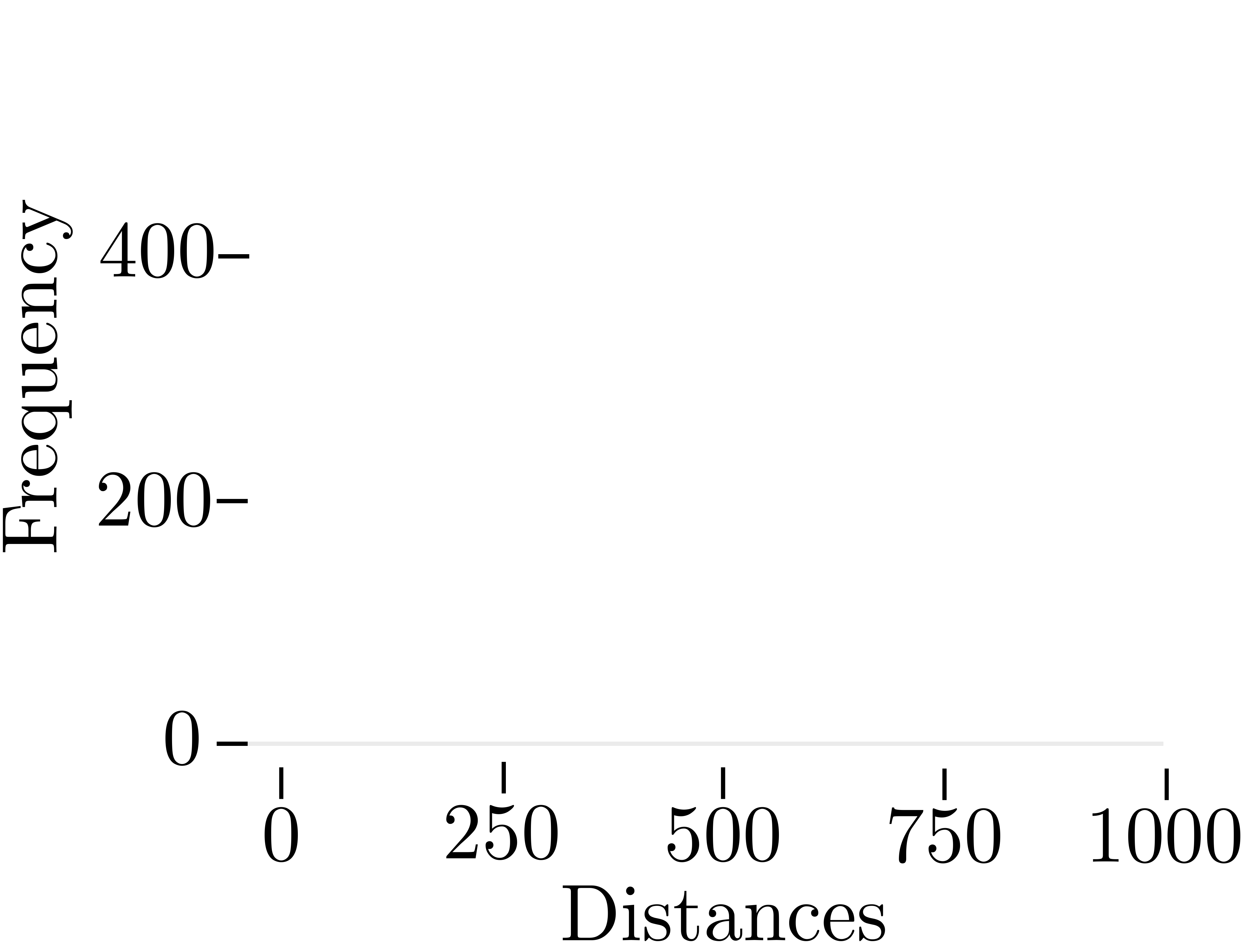}
        \caption{FVO-NIH}
        \label{fig:FVO}
    \end{subfigure}
    \begin{subfigure}[b]{0.32\textwidth}
        \centering
        \includegraphics[width=\textwidth]{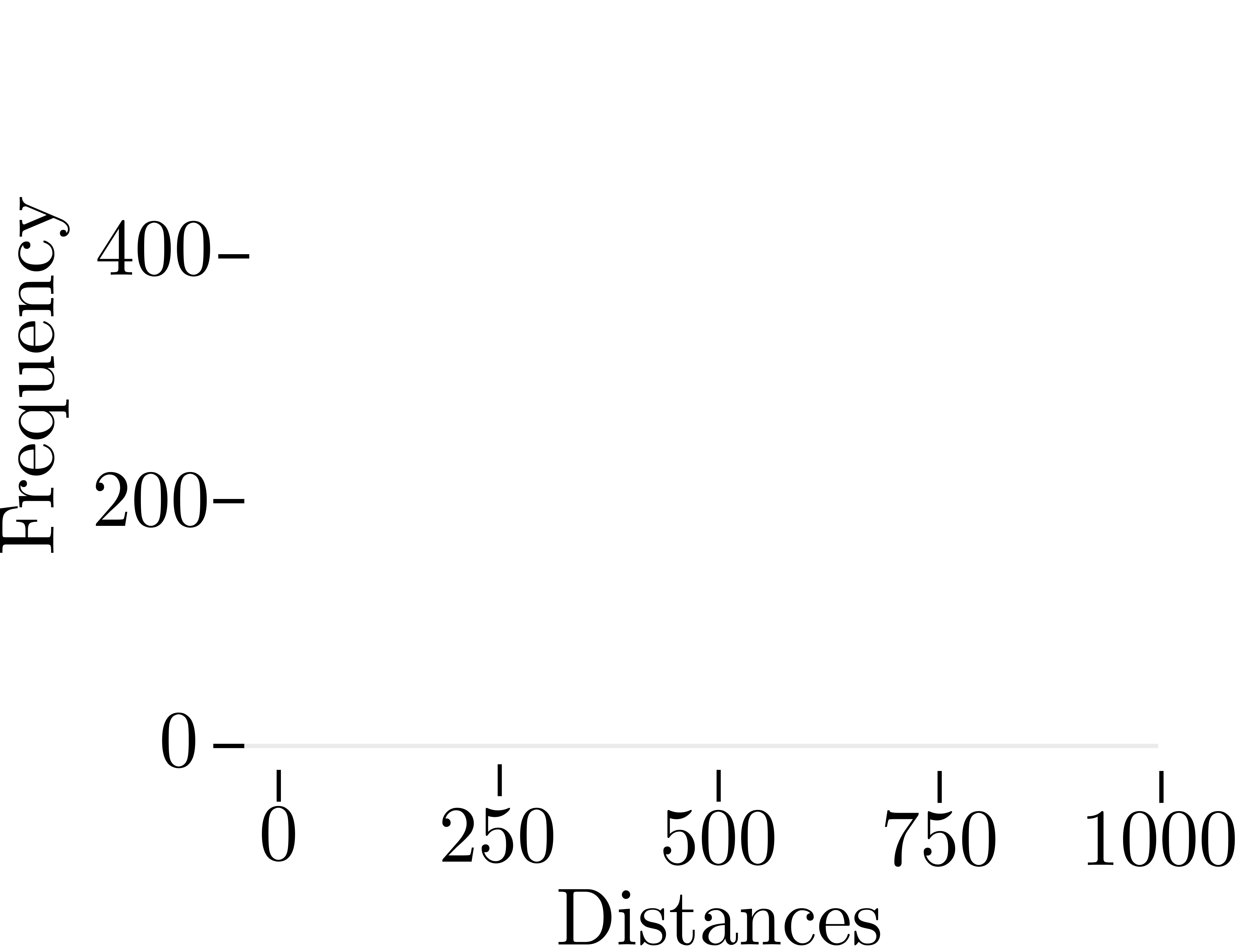}
        \caption{SA250}
        \label{fig:SA250}
    \end{subfigure}
    \caption{Histograms from \textit{Plasmodium falciparum} experiments. The reference
    strain was 3D7 for all experiments. The distribution of baseline extremal
    event DAG distances is shown in blue for each graph and the distribution of
    extremal event DAG distances is shown in purple. In all three plots,  the
    extremal event DAG distances are smaller than the corresponding baseline
    distances. Performing a paired t-test to the blue and purple distributions
    with a null hypothesis that the distributions are the same in all three
    plots resulted in a $p$-value below machine precision.}
    \label{fig:parasite-histograms}
\end{figure}

\begin{figure}[h]
     \centering
     \begin{subfigure}[b]{0.33\textwidth}
         \centering
         \includegraphics[width=\textwidth]{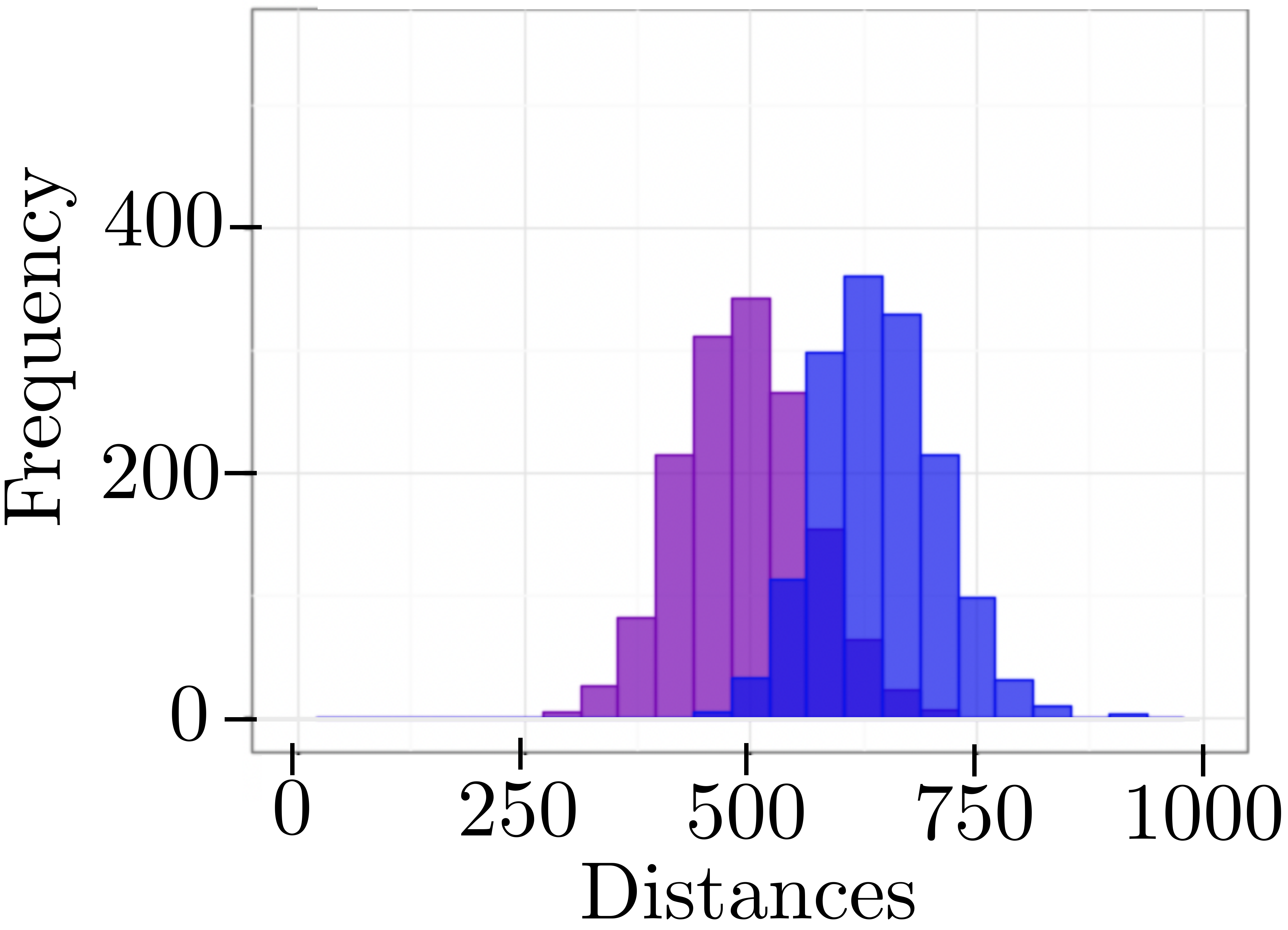}
         \caption{Lung}
         \label{fig:lung}
     \end{subfigure}
     \hfill
     \begin{subfigure}[b]{0.33\textwidth}
         \centering
         \includegraphics[width=\textwidth]{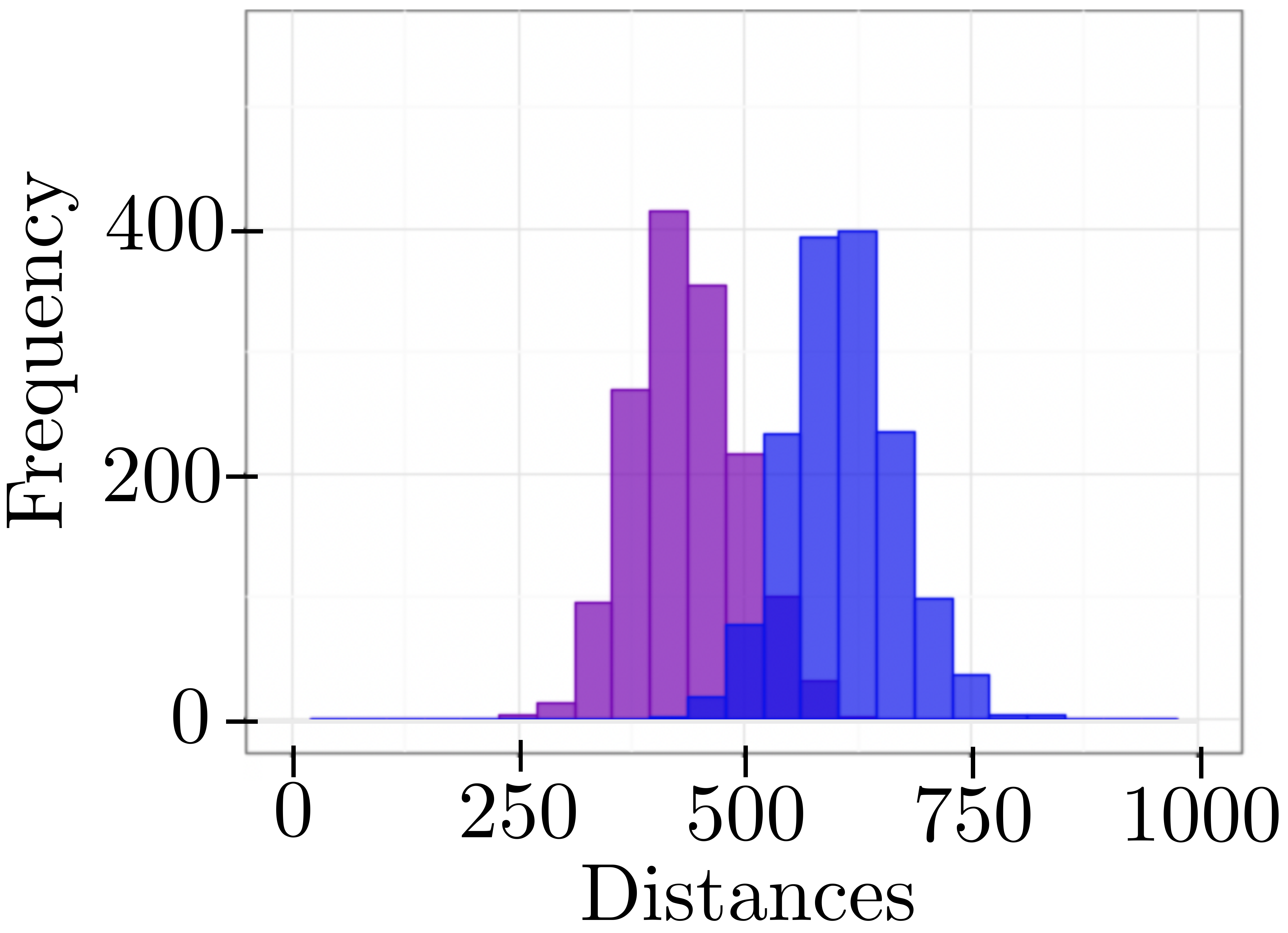}
         \caption{Kidney}
         \label{fig:kidney}
     \end{subfigure}
     \caption{Histograms from mouse experiments. The reference cell line was
     liver for all experiments. The distribution of baseline extremal event DAG
     distances is shown in blue for each graph and the distribution of extremal
     event DAG distances is shown in purple. In all three plots, we see the
     extremal event DAG distances are smaller than the corresponding baseline
     distances. Performing a paired t-test to the blue and purple distributions
     with a null hypothesis that the distributions are the same in all three
     plots resulted in a $p$-value below machine precision. } 
     \label{fig:mouse-histograms}
\end{figure}

Overall, we find the extremal event DAG distances from the parasite experiments to be smaller than the extremal event DAG distances from the mouse experiments, both in absolute value and in terms of distribution overlap, as found in \cite{SmithAn20}, indicating support for a malarial clock. Additionally, \figref{parasite-histograms} shows another pattern seen in~\cite{SmithAn20}, which is that the D6 data have the smallest distance to the 3D7 reference data and that SA250 has the largest distance to 3D7 compared both in absolute value and to their respective baselines.

In this computation, we used subsampling to computationally handle the size of dataset, as in~\cite{SmithAn20}, and computed 1500 samples instead of 5000. However, we used more than double the genes in our samples, 15 versus 6, and incorporated information about all levels of $\varepsilon$ instead of only a small fixed subset of $\varepsilon$.
The consistency of results found between the computations  validates the methodology of using extremal event DAGs in place of $\varepsilon$-DAGs.

\section{Conclusion}

We constructed a weighted directed graph descriptor of collections of time
series data that keeps track of the order and prominence of
extrema, and is robust to experimental noise that arises from taking discrete
time samples of a continuous process. Furthermore, we define a distance between
these extremal event DAGs that constructs an extremal event supergraph using a
modified version of the edit distance and then computes the distance by taking
the $L_1$ distance between aligned node and edge weights in the extremal
event supergraph. The benefit of this distance is that it can be computed via
dynamic programming and efficiently. We used this
distance to compare the similarity of experimental replicates in yeast cell
cycle data, the similarity of circadian gene expression in different mouse
tissues, and the similarity of gene expression across malaria parasite strains.
Our results are consistent with results from other literature
\cite{BerryUsing20, SmithAn20} that used directed maximal common edge subgraphs
of $\varepsilon$-DAGs \cite{NeremA19}. The benefit to using the extremal event
DAG methodology is that the savings in memory and computation speed facilitates
the analysis of significantly larger datasets.

Furthermore, we prove several stability results. In particular, the backbone
distance arising from two functions is bounded by the $L_{\infty}$ distance
of the two functions multiplied by the number of nodes in the backbone infinity
alignment. Using backbone stability, we prove the extremal event DAG distance is
stable in a local case. Local here means that the individual time series in one
collection differs from the corresponding time series in the other collection by
the amount that allows direct alignment of the minima and maxima
between the two time series. Additionally, one of the time series can have small
amplitude additional maxima and minima.
Extension of the local stability result to a stability between arbitrary
multivariate time series is challenging.  If the triangle inequality 
for the extremal event DAG distance holds, then the local
stability of the extremal event DAG can be used to prove a global stability
result using the same technique as in the proof of backbone stability
(\thmref{backbone-infinity-stability}). While in this paper we define the extremal
event DAG distance using the $L_1$ norm, we suspect that stability
also holds if we use any $L_p$ norm. We leave this generalization as future
work.

We focus on a descriptor that is robust to measurement error that arises from taking
discrete time samples of a continuous process. However, there can be other types
of uncertainty present in the data. One type is related to signal
processing and is seen as small peaks in the data. If we want to remove this
type of measurement errors from our analysis, we can apply a preprocessing step using
techniques from \cite{AudunSeparating20}. This technique applies sublevel set
persistence to a time series to determine a node life threshold. Nodes with a
node life below the threshold are classified as results of signal processing errors, or as noise.
Eliminating  nodes classified as noise and then computing extremal event DAGs
gives a smaller descriptor of a collection of time series that further increases
the size of computationally feasible datasets. Similar preprocessing steps to
remove small peaks can be made using Fourier transforms \cite{BollSuppression79,
KamathA02}.

In summary, extremal event DAGs are a new computational tool that can be used
alone or in combination with noise reduction algorithms to summarize and compare
collections of time series data.

\bibliographystyle{plain}
\bibliography{DAG}

\section{Appendix} \label{sec:appendix}

\subsection{Properties of $\varepsilon$-Extremal Intervals}\label{app:ext-intervals-prop}
We prove some properties of the $\varepsilon$-extremal intervals that are useful for computing the edge weights. 
For \lemref{nesting} and \lemref{properties}, we omit the superscript and subscript $f$ from $\varphi^f_{\varepsilon}$ and $\pers_f$ since $f$ is the only function we are considering.

\begin{lem}[Nesting of $\varepsilon$-Extremal Intervals]
    Let $f:C\rightarrow \R$ be a nicely tame function.
    Let $t$ be the domain coordinate of a local extremum of $f$. If $0 <
    \varepsilon_0 < \varepsilon_1$,  then,
    \[ \varphi_{\varepsilon_0}(t) \subset \varphi_{\varepsilon_1}(t). \]
\label{lem:nesting}
\end{lem}

\begin{proof}
Assume $t$ is the domain coordinate of a local minimum. Since sublevel sets of a function form a  filtration, we have
$$(f-\varepsilon_0)^{-1}(-\infty, f(t)+\varepsilon_0) \subset (f-\varepsilon_0)^{-1}(-\infty, f(t)+\varepsilon_1).$$
At the same time,  because $\varepsilon_0 <\varepsilon_1$,
 $$(f-\varepsilon_0)^{-1}(-\infty, f(t)+\varepsilon_1)\subset (f-\varepsilon_1)^{-1}(-\infty, f(t)+\varepsilon_1).$$
Combining these two statements we find that
  \begin{equation}\label{subset}
  (f-\varepsilon_0)^{-1}(-\infty, f(t)+\varepsilon_0)\subset (f-\varepsilon_1)^{-1}(-\infty, f(t)+\varepsilon_1).
  \end{equation}
Note that  $\varphi_{\varepsilon_0}(t) \cap \varphi_{\varepsilon_1}(t) \not = \emptyset$ because they are connected components of\\
$(f-\varepsilon_0)^{-1}(-\infty, f(t)+\varepsilon_0)$ and \mbox{$(f-\varepsilon_1)^{-1}(-\infty, f(t)+\varepsilon_1)$} respectively that contain~$t$. Therefore (\ref{subset}) implies
  $\varphi_{\varepsilon_0}(t)\subset \varphi_{\varepsilon_1}(t)$. An analogous argument holds if $t$ is the domain coordinate of a local maximum.
\end{proof}

\begin{lem}[Properties of $\varphi_{\varepsilon}(t_i)$]
    Let $f:C\rightarrow \R$ be a nicely tame function.
    Let $t_1<t_2<...<t_n$ be the domain coordinates of the local extrema of $f$.
    The following statements hold.
    \begin{enumerate}
        \item The length $\length(\varphi^f_{\varepsilon}(t))$
            increases as a function of~$\varepsilon$. 
            \label{stmt:monotonicity}\label{stmt:properties-inc}
        \item For $i<n$, $\varepsilon \leq \frac{1}{2}|f(t_i)-f(t_{i+1})|$ if and only if
            $t_{i+1} \notin \varphi_{\varepsilon}(t_i)$ and $t_{i} \notin
            \varphi_{\varepsilon}(t_{i+1})$.
            \label{stmt:extrema-containment}\label{stmt:properties-epsiff}
        \item For $i<n$, if $\varepsilon \leq \frac{1}{2}\min\{\pers(t_i), \pers(t_{i+1})\}$,
            then $t_{i+1} \notin \varphi_{\varepsilon}(t_i)$ and $t_{i} \notin
            \varphi_{\varepsilon}(t_{i+1})$.
            \label{stmt:nodelife-extrema-containment}\label{stmt:properties-extrema-containment}
    \end{enumerate}
    \label{lem:properties}
\end{lem}

\begin{proof}

    Let $f$ and $T=\{t_i\}_{i=1}^n$ be defined as in the lemma statement.
    We prove the three statements for local minima first.
    Let $i \in [n]$. Assume that~$(t_i, f(t_i))$ is a local
    minimum.  Note that, since minimum and maximum alternate in $T$,
    we know that $t_{i+1}$ (if it exists) is a local maximum.

    \orgemph{Proof of \stmtref{properties-inc} for minima}.
    Consider two values $0 < \varepsilon_0 < \varepsilon_1$. By \lemref{nesting}, we find
    \[ \varphi_{\varepsilon_0}(t_i) \subset \varphi_{\varepsilon_1}(t_i).\]
    Therefore $\length(\varphi_{\varepsilon_0}(t_i)) \leq \length(\varphi_{\varepsilon_1}(t_i))$.

    \orgemph{Proof of \stmtref{properties-epsiff} for minima}.
    For the forward direction, we assume $i < n$ and $\varepsilon \leq \frac{1}{2}|f(t_i)-f(t_{i+1})|$.
    Since $t_i$ is a local minimum and~$t_{i+1}$ is a local maximum,
    we have~$\varepsilon \leq \frac{1}{2}(f(t_{i+1})-f(t_i))$, which~implies
        \begin{equation}\label{eqn:minmaxineq}
            f(t_{i+1})-\varepsilon \geq f(t_i)+\varepsilon.
        \end{equation}
    By definition of $\varepsilon$-extremal intervals
    (\defref{epsilon-intervals}) and since $t_i$ is a local minimum,
    any point $x \in
    \varphi_{\varepsilon}(t_i)$ satisfies $f(x)-\varepsilon <
    f(t_i)+\varepsilon$. Since we already established that
    $ f(t_{i+1})-\varepsilon \geq f(t_i)+\varepsilon$ in \eqnref{minmaxineq},
    we know that~$t_{i+1} \notin \varphi_{\varepsilon}(t_i)$.
    Similarly, since $t_{i+1}$ is a maximum, for any point~$y \in \varphi_{\varepsilon}(t_{i+1})$,
    we know that~$f(y)+\varepsilon > f(t_{i+1})-\varepsilon$. Along with
    \eqnref{minmaxineq}, we conclude $t_i \notin
    \varphi_{\varepsilon}(t_{i+1})$.

    Next, we prove the backward direction by contrapositive.
    Assume $i < n$ and  $\varepsilon >
    \frac{1}{2}|f(t_i)-f(t_{i+1})|$. Since~$t_i$ is a local minimum and
    $t_{i+1}$ is a local maximum, we have
    $$f(t_{i+1})-\varepsilon <
    f(t_i)+\varepsilon.$$
    Therefore, $t_{i+1} \in (f-\varepsilon)^{-1}(-\infty,
    f(t_i)+\varepsilon)$. In order for $t_{i+1} \in
    \varphi_{\varepsilon}(t_i)$, we need to show that $t_{i+1}$ is in
    the connected component of  $(f-\varepsilon)^{-1}(-\infty,
    f(t_{i})+\varepsilon)$ containing $t_i$.
    Recalling that $\le(*)$ denotes the left endpoint of an interval and since
    \mbox{$t_i \in \varphi_{\varepsilon}(t_i)$}, we have
    $\le(\varphi_{\varepsilon}(t_i))<t_i<t_{i+1}$. In addition, since $t_i$
    and~$t_{i+1}$ are adjacent,
    $\re(\varphi_{\varepsilon}(t_i))>t_{i+1}$. We conclude that $t_{i+1} \in
    \varphi_{\varepsilon}(t_i)$. Therefore,
    \stmtref{extrema-containment} holds for minima.

    \orgemph{Proof of \stmtref{properties-extrema-containment} for
    minima}.
   \stmtref{properties-extrema-containment} follows directly from
    \cite[Proposition 4]{BerryUsing20}.

    For the case where $(t_i, f(t_i))$ is a local  maximum,
    we substitute $-f$ for $f$ and follow the proofs above.
\end{proof}

\subsection{Computing Edge Weights Lemma} \label{app:edgeweights}
We prove \thmref{edge-weights}. We first need the following lemma.

\begin{lem}[Comparability of Extrema from Same function when $\varepsilon$ is Smaller than Node Weights]
    Let $F=\{f_i:C \rightarrow \R\}_{i=1}^n$ be a collection of nicely tame
    functions where $t_1^i<t_2^i<\dots <t_{k_i}^i$ are all the domain coordinates of the local extrema of
    $f_i$. Let \mbox{$\DAG(F) := (V,E,\omega_V,\omega_E)$} be the extremal event DAG
    of~$F$. Let $(v(i,j), v(c,d)) \in E$. Suppose $i=c$, $\varepsilon < \min\{\omega_V (v(i, j)), \omega_V (v(c, d))\}$, 
    and $\varphi^{f_i}_{\varepsilon}(t^i_j)\cap \varphi^{f_c}_{\varepsilon}(t^c_d) \neq \emptyset$. Then $t^i_j$ and $t^c_d$ are comparable.
    \label{lem:appedgeweights}
\end{lem}

\begin{proof}
    Since $i=c$ in this case, we omit the superscripts $i$ and $c$ of $t_j^i$
    and $t_d^c$. Additionally we set $f := f_i$. Furthermore, we omit the
    subscript and superscript $f$ from the functions $\pers_f$
    and~$\varphi_{\varepsilon}^f$.  
    
By Proposition 1 and Theorem 2 of \cite{BerryUsing20},
            one of $t_j$ or $t_d$ is the domain coordinate of a local maximum
            while the other is the domain coordinate of a local minimum, and
            these two extrema are adjacent, i.e. there are no extrema of $f$
            between $t_j$ and $t_d$. Without loss of generality, suppose~$(t_j,
            f(t_j))$ is a local minimum,~$(t_d, f(t_d))$ is a local maximum, and
            $t_j< t_d$. By way of contradiction, suppose $t_j$ and~$t_d$ are
            incomparable. Since we assume that the first two requirements of
            \defref{comparable} are satisfied, it is the third condition that is
            violated. Therefore it is not true that  $t_j\prec_\varepsilon t_d$
            nor it is true that  $t_d \prec_\varepsilon t_j$. Since $t_j < t_d$,
            there exists $g \in N_{\varepsilon}(f)$ such that for every $t_j'
            \in \varphi_{\varepsilon}(t_j)$ that is a domain coordinate of a
            local minimum of $g$, and $t_d' \in \varphi_{\varepsilon}(t_d)$ that
            is a domain coordinate of a local maximum of~$g$, we
            have~$t_j'> t_d'$. Consider such $t_j'$ and $t_d'$ that are adjacent.

            \orgemph{We claim that $g(t_j')>f(t_d)-\varepsilon$}. On the contrary, suppose
            $g(t_j') \leq f(t_d)-\varepsilon$. We show there exists a local
            maximum $(t_g', g(t_g'))$ such that $t_g' \in
            \varphi_{\varepsilon}(t_d)$ and $t_j' < t_g'$, which contradicts the
            assumption $t_j$ and~$t_d$ are incomparable.

            Since $g \in N_{\varepsilon}(f)$, we have $f(t_d)-\varepsilon <
            g(t_d)$. Hence, $g(t_j')<g(t_d)$. Suppose $C=[a_1, a_2]$ where $a_1,a_2 \in
            \R$. We discuss two cases: either $\re(\varphi_{\varepsilon}(t_d))=a_2$ or $\re(\varphi_{\varepsilon}(t_d))\neq a_2$.

                    \orgemph{We prove Case 1 that $\re(\varphi_{\varepsilon}(t_d)) \neq a_2$}. By definition of
                    $\varepsilon$-extremal intervals,
                    $f(\re(\varphi_{\varepsilon}(t_d)))+\varepsilon =
                    f(t_d)-\varepsilon$. Since $g \in N_{\varepsilon}(f)$,
                    \[
                        g(\re(\varphi_{\varepsilon}(t_d)))<f(\re(\varphi_{\varepsilon}(t_d)))+\varepsilon
                        = f(t_d)-\varepsilon.
                    \]
                    Recall, $f(t_d)-\varepsilon <
                    g(t_d)$. Therefore, $g(\re(\varphi_{\varepsilon}(t_d))) <
                    g(t_d)$. By \stmtref{nodelife-extrema-containment} of
                    \lemref{properties}, we have $t_d \notin
                    \varphi_{\varepsilon}(t_j)$. Hence, $t_j'<t_d$. Furthermore,
                    by assumption $t_d'<t_j'$ where $t_d' \in
                    \varphi_{\varepsilon}(t_d)$. Hence,~$t_j'>t_d'>\le(\varphi_{\varepsilon}(t_d))$ and
                    $t_j'<t_d<\re(\varphi_{\varepsilon}(t_d))$. Thus, $t_j' \in
                    \varphi_{\varepsilon}(t_d)$. All together we have, \[
                        t_j'<t_d<\re(\varphi_{\varepsilon}(t_d)),\quad g(t_j') <
                        g(t_d), \quad g(t_d)>g(\re(\varphi_{\varepsilon}(t_d))).
                    \] This and the assumption $g$ is nicely tame implies there
                    exists a local maximum $(t_g', g(t_g'))$ such that~$t_g' \in
                    (t_j', \re(\varphi_{\varepsilon}(t_d)))$. Hence, $t_g' \in
                    \varphi_{\varepsilon}(t_d)$ and $t_j'<t_g'$.
                    
                    \orgemph{We prove Case 2 $\re(\varphi_{\varepsilon}(t_d))=a_2$}. Using the same reasoning as in the case 
                    $(\re(\varphi_{\varepsilon}(t_d)) \neq a_2)$, we find 
                    	\[t_j'<t_d<\re(\varphi_{\varepsilon}(t_d)),\quad g(t_j') <
                        g(t_d). \]
                        	Furthermore, we must have either
	         \[ g(t_d) > g(\re(\varphi_{\varepsilon}(t_d))) \text{ or } g(a_2) = g(\re(\varphi_{\varepsilon}(t_d))) \geq g(t_d). \]
	         Either way, we can conclude there exists a local maximum $t_g' \in (t_j',a_2] \subset \varphi_{\varepsilon}(t_d)$ of $g$.  

            In both cases ($\re(\varphi_{\varepsilon}(t_d))=a_2$ and $\re(\varphi_{\varepsilon}(t_d))\neq a_2$) we find there exists a local
            maximum $(t_g', g(t_g'))$ such that $t_g' \in
            \varphi_{\varepsilon}(t_d)$ and $t_j'<t_g'$. This shows that $t_j
            \prec_{\varepsilon} t_d$, which is a contradiction.

            Therefore the claim
            $$g(t_j')>f(t_d)-\varepsilon$$
            holds.
            A similar argument can be used to show
            $$g(t_d')<f(t_j)+\varepsilon.$$
             By \stmtref{nodelife-extrema-containment} of
             \lemref{properties}, $t_j \notin
            \varphi_{\varepsilon}(t_d)$ since $\varepsilon < \min\{\omega_V(v(i,j)), \omega_V(v(c,d))\} = \frac{1}{2}\min\{\pers(t_j), \pers(t_d)\}$. Applying
             \stmtref{extrema-containment} of
            \lemref{properties}, we get $\varepsilon \leq
            \frac{1}{2}(f(t_d)-f(t_j))$. Hence,
            $$g(t_j')>f(t_d)-\varepsilon\geq \frac{1}{2}(f(t_d)+f(t_j)).$$
            $$g(t_d')<f(t_j)+\varepsilon\leq \frac{1}{2}(f(t_d)+f(t_j)).$$
            This implies that
            $$g(t_d')<g(t_j').$$
            Since $t_d'<t_j'$ and $(t_d', g(t_d'))$ is a local maximum while
            $(t_j', g(t_j'))$ is local minimum of $g$, there must exist domain
            coordinates of a local minimum and maximum of $g$ between $t_d'$ and
            $t_j'$. Hence, we reach a contradiction with the assumption that
            $t_d'$ and $t_j'$ are adjacent extrema. Therefore, we conclude that
            $t_j$ and $t_d$ are comparable.

\end{proof}

\noindent \textbf{\thmref{edge-weights}} (Computing Edge Weights). Let $F=\{f_i:C \rightarrow \R\}_{i=1}^n$ be a collection of nicely tame
    functions where $t_1^i<t_2^i<\dots <t_{k_i}^i$ are all the domain coordinates of the local extrema of
    $f_i$. Let \mbox{$\DAG(F) := (V,E,\omega_V,\omega_E)$} be the extremal event DAG
    of~$F$. For all edges~$(v(i,j), v(c,d)) \in E$, the following statements hold

    \begin{enumerate}
        \item If $i=c$, then
            $$\omega_{E}(v(i,j), v(c,d)) = \min\{\omega_V(v(i,j)), \omega_V(v(c,d))\}.$$ \label{stmt:edge-weights-same}
        \item If $i\neq c$, then
            $$\omega_{E}(v(i,j), v(c,d))=\min\{\omega_V(v(i,j)), \omega_V(v(c,d)), \varepsilon^*(t_j^i, t_d^c)\},$$
            where
            $$\varepsilon^*(t_j^i, t_d^c) := \inf\{\varepsilon \mid \varphi^{f_i}_{\varepsilon}(t_j^i) \cap \varphi^{f_c}_{\varepsilon}(t_d^c) \neq \emptyset\}.$$
            \label{stmt:edge-weights-dif}
    \end{enumerate}

\begin{proof}
    Assume all hypotheses.

    \orgemph{First, we prove \stmtref{edge-weights-same}}.
    Since $i=c$ in this case, we omit the superscripts $i$ and $c$ of $t_j^i$
    and $t_d^c$. Additionally we set $f := f_i$. Furthermore, we omit the
    subscript and superscript $f$ from the functions $\pers_f$
    and~$\varphi_{\varepsilon}^f$.  
    
    \orgemph{First, suppose $\varepsilon < \min\{\omega_V(v(i,j)), \omega_V(v(c,d))\}$.} We show that $t_j$ and $t_d$ are comparable.  We consider two~cases.

       \orgemph{Suppose $\varphi_{\varepsilon}(t_j) \cap \varphi_{\varepsilon}(t_d) =
            \emptyset$}. Without loss of generality, assume $t_j<t_d$. Let $g\in
            N_{\varepsilon}(f)$. Using Proposition 2 and Corollary 2
            of~\cite{BerryUsing20}, we get that $\varphi_{\varepsilon}(t_j)$ and
            $\varphi_{\varepsilon}(t_d)$ contain local extrema $t_j'$ and $t_d'$
            of the same type as~$t_j$ and $t_d$. Since
            $\varphi_{\varepsilon}(t_j) \cap \varphi_{\varepsilon}(t_d) =
            \emptyset$, $t_j < t_d$ implies $t_j' < t_d'$. Therefore, $t_j$ and
            $t_d$ are comparable.

        \orgemph{Next suppose $\varphi_{\varepsilon}(t_j) \cap \varphi_{\varepsilon}(t_d) \neq
            \emptyset$.} This is the content of \lemref{appedgeweights} in \appref{edgeweights}.
    Altogether we find that if
    $\varepsilon < \min\{\omega_V(v(i,j), \omega_V(v(c,d))\}$, then $t_j$ and $t_d$ are
    comparable. 
    
   \orgemph{ Lastly, we consider the case that $\varepsilon \geq
    \min\{\omega_V(v(i,j), \omega_V(v(c,d))\}$.} By \defref{comparable}, $t_j$ and $t_d$ are
    incomparable.

    We have shown that $t_j$, $t_d$ are comparable for all $\varepsilon <
    \min\{\omega_V(v(i,j)), \omega_V(v(c,d))\}$ and $t_j$, $t_d$ are incomparable for
    all~$\varepsilon \geq \min\{\omega_V(v(i,j)), \omega_V(v(c,d))\}$. Therefore, \[
        \min\{\omega_V(v(i,j)), \omega_V(v(c,d))\} =  \inf\{\varepsilon\mid t_j \text{ and }
        t_d \text{ are incomparable}\}. \] We conclude  $\omega_{E}(v(i,j),
    v(c,d)) = \min\{\omega_V(v(i,j)), \omega_V(v(c,d))\}$.

\vspace{1ex}

    \orgemph{Next, we prove \stmtref{edge-weights-dif}}.
    Let $(t_j^i, f_i(t_j^i))$ and $(t_d^c, f_c(t_d^c))$ be local extrema.  First,
    let
    $$\varepsilon\leq\min\{\omega_V(v(i,j)), \omega_V(v(c,d)),
    \varepsilon^*\}.$$
     Then, by definition of $\varepsilon^*$, the intervals
    $\varphi^{f_i}_{\varepsilon}(t_j^i)$ and
    $\varphi^{f_c}_{\varepsilon}(t_d^c)$ are disjoint. Additionally, from
    Proposition 2 and Corollary 2 of~\cite{BerryUsing20},  both intervals
    guarantee existence of local extrema of the appropriate type under
    any~$\varepsilon$-perturbation. Therefore, $t_j^i$ and $t_d^c$ are
    comparable.

    Next, if $\varepsilon^*< \varepsilon \leq \min\{\omega_V(v(i,j)), \omega_V(v(c,d))\}$, then $\varphi_{\varepsilon}^{f_i}(t_j^i) \cap
    \varphi^{f_c}_{\varepsilon}(t_d^c) \neq \emptyset$.  Since a local extremum
    of an $\varepsilon$-perturbation can happen anywhere in
    $\varphi_{\varepsilon}(t_j^i)$ and $\varphi_{\varepsilon}(t_d^c)$, then
    $t_j^i$ and $t_d^c$ are incomparable at $\varepsilon$.

    Lastly, if $\varepsilon\geq\min\{\omega_V(v(i,j)), \omega_V(v(c,d))\}$,
    then by \defref{comparable}, $t_j^i$ and $t_d^c$ are incomparable.

    We conclude that  $\omega_{E}(v(i,j),
    v(c,d))=\min\{\omega_V(v(i,j), \omega_V(v(c,d)), \varepsilon^*(t_j^i,
    t_d^c)\}$.
\end{proof}

\subsection{Backbone Distance is A Metric}\label{app:backbone}
In order to show the triangle
inequality, $d_{\mathcal{B}}(\x, \z) \leq d_{\mathcal{B}}(\x, \y) +
d_{\mathcal{B}}(\y, \z)$, we need a way of composing an alignment between $\x$
and $\y$ with an alignment between $\y$ and $\z$. This is the content of
\conref{composition-map} and \lemref{induced-alignment}.

\begin{con}[Composition of Alignments]\label{con:composition-map}
    Let $\x$, $\y$, $\z$ be backbones and $\alpha_1: [k] \rightarrow \tilde{\x}
    \times \tilde{\y}$, $\alpha_2: [l]
    \rightarrow \tilde{\y} \times \tilde{\z}$
    be alignment maps. The \emph{composition of $\alpha_1$ and
    $\alpha_2$} induces an ordered correspondence between $\tilde{\x}$ and $\tilde{\z}$
    whose nontrivial pairs are given by
$$
       \A := \left\{ (x,z) ~|~
            \exists y\in \y \text { s.t.\ } (x,y) \in \im(\alpha_1)
            \text{ and }
            (y,z) \in \im(\alpha_2)
        \right\}.
    $$
    The pairs in $\A$ are ordered using the order given in $\x$.  
    The set $A$, however does not form a complete alignment of $\x$ and $\z$ because $x_i \in \x$ with  $(x_i, \zero) \in \im(\alpha_1)$ and $z_j \in \z$ with $(\zero,z_j) \in \im(\alpha_2)$ are not accounted for in pairs in $\A$. We include all such pairs $(x_i,\zero)$ and $(\zero,z_j)$ to construct a function $\alpha_2\circ \alpha_1: [s] \to \tilde{\x} \times \tilde{\z}$ in such a way that each  pair satisfies
    \begin{itemize}
    \item[(a)] If $\iota_\x(x_i) < \iota_\x(x)$ ($\iota_\x(x) <\iota_\x(x_j))$  and $(x_i,z_i) \in A$ ($(x_j,z_j) \in A$) then for 
    \[ \alpha_2\circ \alpha_1(s_1) = (x_i,z_i),\quad \alpha_2\circ \alpha_1(s_2) = (x, \zero),\quad  \alpha_2\circ \alpha_1(s_3) (x_j,z_j) \] 
    we have the order  $s_1<s_2$ ($s_2 <s_3$).
    \item[(b)] If $\iota_\z(z_i) < \iota_\z(z)$ ($\iota_\z(z) <\iota_\z(z_j))$  and $(x_i,z_i) \in A$ ($(x_j,z_j) \in A$) then  for 
    \[ \alpha_2\circ \alpha_1(s_1) = (x_i,z_i),\quad \alpha_2\circ \alpha_1(s_2) = (\zero,z),\quad  \alpha_2\circ \alpha_1(s_3) (x_j,z_j) \] 
    we have the order $s_1<s_2$ ($s_2 <s_3$).
    \end{itemize} 
    An extension of $\A$ to an alignment $\alpha_2\circ \alpha_1: [s] \to \tilde{\x} \times \tilde{\z}$ exists, but also that it may not be unique since there may be multiple options on where to place consecutive insertions. We resolve this ambiguity arbitrarily in the following way. If there is a set of consecutive insertions $\{ (x_{i_1},\zero),(x_{i_2},\zero),\ldots, (x_{i_k},\zero), (\zero, z_{j_1}),\ldots, (\zero, z_{j_s})\} $ between a pair $(x_i, z_i)$ and a pair $(x_j, z_j)$, then we order them starting with the insertions $\{(x_{i_r}, \zero)\}_{r=1}^k$ ordered by the order given in $\x$, followed by the insertions $\{(\zero, z_{j_r})\}_{r=1}^s$ insertions ordered by the order given in $\z$. We hasted to point out that the following results (\lemref{induced-alignment}, \lemref{backbone-triangle}, and \lemref{backbone-infty-triangle}) hold for $\alpha_2\circ \alpha_1$ defined using any choice of order of such consecutive insertions. 
\end{con}

The ordered correspondence $\alpha_2 \circ \alpha_1$ is, in fact, an alignment:

\begin{lem}[Composition is an Alignment]\label{lem:induced-alignment}

    Let $\x$, $\y$, $\z$ be backbones and $\alpha_1: [k] \rightarrow \tilde{\x} \times
    \tilde{\y}$, $\alpha_2: [l]
    \rightarrow \tilde{\y} \times \tilde{\z}$
    be alignment maps.  Then,
    $\alpha_2 \circ \alpha_1$ is an alignment between $\x$ and $\z$.

\end{lem}

\begin{proof}
\orgemph{We show that all properties of \defref{alignment} hold.}
Since $\alpha_1, \alpha_2$ are alignments, they contain no null alignments or misalignments. This implies that $\alpha_2 \circ \alpha_1$ also contains no null alignments or misalignments. Furthermore, $\im(\alpha_2\circ \alpha_1)$ contains all $x \in \x$ that are nontrivially aligned with $y \in \y$ along with all $x \in \x$ that are aligned with an insertion. Hence all $x\in \x$ appear in the image of $(\alpha_2 \circ \alpha_1)_{\x}$ exactly once. The same can be said for all $z \in \z$. Therefore, \propertyref{nullalignments}, \propertyref{nomisalignments}, and \propertyref{restrictiontomatching} of \defref{alignment} hold. 

\orgemph{Lastly, we show that $\alpha_2 \circ \alpha_1$ preserves order of the backbones $\x$ and $\z$.} First consider the set of nontrivial pairs $\A$. We show the order of the (partial) backbones $\x$ and $\z$ is preserved in this set. By construction, we know the order of $\x$ is preserved. Replacing each pair $(x,z) \in A$ by $(x,y) \in \im(\alpha_1)$ where $(y,z) \in \im(\alpha_2)$, we see the order of $\y$ is preserved because $\alpha_1$ is an alignment. Next, replacing each of these pairs with $(y,z)$ where $(y,z) \in \im(\alpha_2)$, we see that order $\z$ is preserved in $A$ because $\alpha_2$ is an alignment. Hence, the order of the backbones $\x$ and $\z$ are preserved in $\A$. By \conref{composition-map}, each remaining trivial pair is added to the set $\A$ so that the order of the backbones $\x$ and $\z$ are preserved. Thus, \propertyref{preservesbackbones} of \defref{alignment} also holds and we have an alignment $\alpha_2 \circ \alpha_1$ between $\x$ and $\z$. 
\end{proof}

Next, we prove that the backbone distance satisfies the triangle inequality. We
consider three backbones $\x, \y,$ and $\z$. We use optimal alignments
between $\x$ and $\y$, and $\y$ and $\z$, to construct an alignment between $\x$
and $\z$. We show that the constructed alignment between $\x$ and $\z$ satisfies the
triangle inequality. Since the cost of the alignment between $\x$ and $\z$ that we find is
an upper bound of the cost of an optimal alignment between $\x$ and $\z$, the backbone
distance must also satisfy the triangle inequality.

\begin{lem}[Backbone Distance Satisfies Triangle Inequality]\label{lem:backbone-triangle}
    Let $\x$, $\y$, $\z$ be backbones. Then,
    $$d_{\mathcal{B}}(\x, \z) \leq d_{\mathcal{B}}(\x, \y) + d_{\mathcal{B}}(\y,
    \z).$$
\end{lem}

\begin{proof}

    Let $\alpha_1: [k] \rightarrow \tilde{\x}
    \times \tilde{\y}$, $\alpha_2: [m]
    \rightarrow \tilde{\y} \times \tilde{\z}$ be optimal alignments. Consider the composition alignment~\mbox{$\alpha_2 \circ \alpha_1$} from \conref{composition-map}. Define $\A$ as in \conref{composition-map}, 
    $$\Ax := \{(x, y) \in \im(\alpha_1)  \mid \exists z \in \z \st (y, z) \in \im(\alpha_2) \},$$ 
    and 
    $$\Az := \{(y,z) \in \im(\alpha_2)  \mid \exists x \in \x \st (x,y) \in \im(\alpha_1) \}.$$ 
    We start by considering and justifying all the relations we need in the
    triangle inequality computation.
    Since $d_{\mathcal{B}}(\x, \z)$ is computed
    using an optimal alignment between $\x$ and $\z$ and $\alpha_2 \circ
    \alpha_1$ is one alignment,
    \begin{equation}\label{eq:1}
    d_{\B}(\x, \z) \leq 
        \sum_{(x,z) \in \im(\alpha_2\circ \alpha_1)} |w_x - w_z|
        + \sum_{(x,\zero) \in \im(\alpha_2\circ \alpha_1)}  w_x
        + \sum_{(\zero,z) \in \im(\alpha_2\circ \alpha_1)} w_z 
    \end{equation}

    We apply the $L_1$ norm triangle inequality to the first term in \eqref{eq:1} to get
    \begin{equation}\label{eq:2}
    \sum_{(x,z) \in \im(\alpha_2\circ \alpha_1)} |w_x - w_z| \leq \sum_{(x, y) \in \Ax} |w_x - w_y| + \sum_{(y,z) \in \Az}|w_y-w_z|.
    \end{equation}
      \vskip0.2in
      
    Now we discuss the second term in \eqref{eq:1}. 
    Define 
    $$\Xzero := \{(x, \zero) \in \im(\alpha_2\circ \alpha_1) \mid \exists y \in \y \st (x, y) \in \im(\alpha_1)\setminus \Ax\}.$$
    Observe, if $(x,y) \in \im(\alpha_1) \setminus \Ax$, then for all $z\in \z$, $(y,z) \notin \im(\alpha_2)$. Hence, $(y, \zero) \in \im(\alpha_2)$. This implies the set  
    $$(\x \times \{\zero\}) \cap \im(\alpha_2 \circ \alpha_1) = ((\x \times \{\zero\}) \cap \im(\alpha_1)) \cup \Xzero$$
    and this union is disjoint. Thus,
    
    \begin{equation}\label{eq:4}
    \sum_{(x, \zero) \in \im(\alpha_2 \circ \alpha_1) } w_x = \sum_{(x, \zero) \in \im(\alpha_1)} w_x + \sum_{(x, \zero) \in \Xzero} w_x. 
    \end{equation}
    
    By definition of $\Xzero$, for each $(x, \zero) \in \Xzero$, there exists $y \in \y$ such that $(x, y) \in \im(\alpha_1)\setminus \Ax$ and, at the same time,  $(y, \zero) \in \im(\alpha_2)$. Noting this observation and 
    applying the triangle inequality from the $L_1$-norm to the last term
    of Equation \eqref{eq:4}, gives
    \begin{equation}\label{eq:6}
    \begin{split}
        \sum_{(x, \zero) \in \im(\alpha_2\circ \alpha_1)}w_x \leq \sum_{(x, \zero) \in \im(\alpha_1)} w_x
        &+ \sum_{\im(\alpha_1)\setminus \Ax}|w_x- w_y|\\
         &+ \sum_{(y, \zero) \in \im(\alpha_2) \st (x,y) \in \im(\alpha_1)} w_y.
    \end{split}
    \end{equation}
    
    \vskip0.2in

    Now we discuss the third term in \eqref{eq:1}.
     Define 
     $$\Zzero := \{(\zero, z) \in \im(\alpha_2\circ \alpha_1) \mid \exists y \in \y \st (y, z) \in \im(\alpha_2)\setminus \Az\}.$$
    Similarly, the set 
    $$(\{\zero\} \times \z) \cap \im(\alpha_2 \circ \alpha_1) = ((\{\zero\} \times \z) \cap \im(\alpha_2)) \cup \Zzero.$$ This implies
    
    \begin{equation}\label{eq:5}
        \sum_{(\zero, z) \in \im(\alpha_2 \circ \alpha_1)} w_z
        = \sum_{(\zero, z) \in \im(\alpha_2)} w_z + \sum_{(\zero, z) \in \Zzero}w_z.
    \end{equation}

By definition, for each $(\zero, z) \in \Zzero$, there exists $y\in \y$ such that $(y,z) \in \im(\alpha_2)\setminus \Az$ and 
$(\zero, y) \in \im(\alpha_1)$. Applying the triangle inequality from the $L_1$-norm to the
    last term of Equation \eqref{eq:5}, we get 
    \begin{equation}\label{eq:7}
    \begin{split}
    \sum_{(\zero, z) \in \im(\alpha_2 \circ \alpha_1)} w_z
        \leq \sum_{(\zero, z) \in \im(\alpha_2)} w_z
            &+ \sum_{(y, z) \in \im(\alpha_2)\setminus \Az} |w_z-w_y|\\
            & + \sum_{(\zero, y) \in \im(\alpha_1) \st (y,z) \in \im(\alpha_2) } w_y.
    \end{split}
    \end{equation}

We derive two additional sets of relationships that will be used in the final estimate. We first note that
 since $\im(\alpha_1)\cap (\x \times \y) = \Ax \cup ((\im(\alpha_1)\cap (\x \times \y)) \setminus \Ax)$ and this union is disjoint,
    \begin{equation}\label{eq:9}
    \sum_{(x, y) \in \Ax} |w_x - w_y| +\sum_{(x,y) \in \im(\alpha_1)\setminus \Ax} |w_x - w_y| = \sum_{(x, y) \in \im(\alpha_1) } |w_x - w_y|.
    \end{equation}
    Analogously, since $\im(\alpha_2) \cap (\y \times \z) = \Az \cup ((\im(\alpha_2)\cap (\y\times \z))\setminus \Az)$ and this union is disjoint. Thus,
    \begin{equation}\label{eq:11}
        \sum_{(y,z) \in \Az}|w_y-w_z| +\sum_{(y,z) \in \im(\alpha_2) \setminus \Az}|w_y - w_z| = \sum_{(y,z) \in \im(\alpha_2)}|w_y - w_z|.
    \end{equation}

The second set of relationships are inequalities. First, notice if $(y,z) \in \im(\alpha_2) \setminus \Az$, then $y$ must align with an empty node in the alignment $\alpha_1$. Hence, 
$$\{(\zero, y) \mid (y,z) \in \im(\alpha_2) \setminus \Az\} \subset \{(\zero, y) \in \im(\alpha_1)\}.$$ This implies
    
    \begin{equation}\label{eq:10}
        \sum_{(\zero, y) \st (y, z) \in \im(\alpha_2) \setminus \Az} w_y
        \leq  \sum_{(\zero, y) \in \im(\alpha_1)} w_y.
    \end{equation}

Finally, we note that $\{(y, \zero) \in \im(\alpha_2) \mid (x,y) \in \im(\alpha_1)\} \subset \im(\alpha_2) \cap (\y \times \{\zero\})$. Therefore,
    
    \begin{equation}\label{eq:12}
        \sum_{(y, \zero) \in \im(\alpha_2) \st (x,y) \in \im(\alpha_1)} w_y
        \leq \sum_{(y, \zero) \in \im(\alpha_2)} w_y
    \end{equation}

    Now we can put all these relations together to prove the backbone distance
    satisfies the triangle inequality.

    \begin{equation*}
    \begin{split}
        d_{\mathcal{B}}(\x, \z) &\leq  \sum_{(x,z) \in \im(\alpha_2\circ \alpha_1) } |w_x - w_z|
        + \sum_{(x,\zero) \in \im(\alpha_2\circ \alpha_1)}  w_x
        + \sum_{(\zero, z) \in \im(\alpha_2\circ \alpha_1)} w_z \\
            &\quad\quad \text{ by Equation \eqref{eq:1}}\\
       &\leq \sum_{(x, y) \in \Ax} |w_x - w_y| + \sum_{(y,z) \in \Az}|w_y-w_z| + \sum_{(x,\zero) \in \im(\alpha_2\circ \alpha_1)}  w_x\\
        &\quad\quad + \sum_{(\zero,z) \in \im(\alpha_2\circ \alpha_1)} w_z \text{ by Equation \eqref{eq:2}}\\              
       &\leq \sum_{(x, y) \in \Ax} |w_x - w_y| + \sum_{(x,y) \in \im(\alpha_1) \setminus \Ax}|w_x- w_y|\\ 
       &\quad\quad + \sum_{(x, \zero) \in \im(\alpha_1)} w_x
      + \sum_{(y, \zero) \in \im(\alpha_2) \st (x,y) \in \im(\alpha_1)} w_y \\
       &\quad\quad + \sum_{(y,z) \in \Az}|w_y-w_z| + \sum_{(y, z) \in \im(\alpha_2)\setminus \Az} |w_z-w_y|\\ 
       &\quad\quad+ \sum_{(\zero, y) \in \im(\alpha_1) \st (y,z) \in \im(\alpha_2)} w_y 
        + \sum_{(\zero, z) \in \im(\alpha_2)} w_z\\
        &\quad\quad \text{by Equations \eqref{eq:6} and \eqref{eq:7}}\\        
        & \leq \sum_{(x, y) \in \im(\alpha_1)} |w_x - w_y| + \sum_{(x, \zero) \in \im(\alpha_1)} w_x  +\sum_{(\zero, y) \in \im(\alpha_1)} w_y\\
        & \quad\quad + \sum_{(y,z) \in \im(\alpha_2)}|w_y - w_z| + \sum_{(y, \zero) \in \im(\alpha_2)} w_y + \sum_{(\zero, z) \in \im(\alpha_2)} w_z\\
        &\quad\quad  \text{ by Equations \eqref{eq:9}, \eqref{eq:10}, \eqref{eq:11}, and \eqref{eq:12}}\\
        &= d_{\mathcal{B}}(\x, \y)+d_{\mathcal{B}}(\y, \z).
    \end{split}
    \end{equation*}

    Hence, $d_{\mathcal{B}}(\x, \z) \leq d_{\mathcal{B}}(\x, \y) + d_{\mathcal{B}}(\y, \z).$

\end{proof}

We can now prove the backbone distance is a metric.

\begin{prop}[Backbone Distance is a Metric]
    The backbone distance (\defref{backbone-dist}) is a metric.
\label{prop:backbone-dist-is-dist}
\end{prop}

\begin{proof}
    Let $\x$, $\y$ be backbones.
    Recall that for each $x\in \x$, we can write $x=(s_x, w_x)$; likewise for
    $y\in\y$.
    Let~$\alpha: [k] \rightarrow \tilde{\x} \times
    \tilde{\y}$ be an optimal alignment. We verify all properties of a metric.

    \orgemph{Non-Negativity.} Since all node weights are non-negative,
    $d_{\mathcal{B}}(\x, \y) \geq 0$.

    \orgemph{Symmetry.} By construction, $d_{\mathcal{B}}$ is
    symmetric; see \eqref{backbone-dist}.

    \orgemph{Definiteness.} If $\x = \y$, then the optimal alignment
    aligns each node with itself, and there are no insertions. Hence, all node weights match and
    $d_{\mathcal{B}}(\x, \y) = 0$.

    On the other hand, assume $d_{\mathcal{B}}(\x, \y) = 0$. This implies
    $$
        0=
        \sum_{(x,y) \in \im(\alpha)} |w_x - w_y|
        + \sum_{(x,\zero) \in \im(\alpha)} w_x
        + \sum_{(\zero, y) \in \im(\alpha)} w_y.
    $$
    Since all weights are non-negative, the latter two summands (corresponding
    to nodes aligned with insertions) sum to zero.
    Since there are no null alignments (\defref{alignment} \propertyref{nullalignments}),
    every node in $\x$ is aligned
    with a node in $\y$. We have $$0 =
    \sum_{(x,y) \in \im(\alpha)} |w_x - w_y|$$ and so
    $w_x = w_y$ for all nontrivial pairs $(x,y)$.
    Additionally, if $(x,y) \in \im (\alpha)$, we know that $s_x =
    s_y$.
    Since this is true
    for all nodes, we must have $\x = \y$.

    \orgemph{Triangle Inequality.} By
    \lemref{backbone-triangle}, the triangle inequality holds.

    Therefore, the backbone distance is a metric.

\end{proof}

\subsection{Backbone Infinity Distance is A Metric}\label{app:backboneinfty}

We show the backbone infinity distance is a metric. We start by proving that the triangle inequality holds.

 \begin{lem}[Backbone Infinity Distance Satisfies Triangle Inequality]
     Let $\x, \y, \z$ be backbones. Then,
     \[ d_{\mathcal{B}_{\infty}}(\x, \z) \leq d_{\mathcal{B}_{\infty}}(\x, \y) + d_{\mathcal{B}_{\infty}}(\y, \z) \]
     \label{lem:backbone-infty-triangle}
 \end{lem}

 \begin{proof}
Let $\alpha_1: [k] \rightarrow \tilde{\x}
    \times \tilde{\y}$, $\alpha_2: [m]
    \rightarrow \tilde{\y} \times \tilde{\z}$ be optimal alignments. Consider the composition alignment~\mbox{$\alpha_2 \circ \alpha_1$} from \conref{composition-map}. Define $\A$ as in \conref{composition-map}, 
    $$\Ax := \{(x, y) \in \im(\alpha_1)  \mid \exists z \in \z \st (y, z) \in \im(\alpha_2) \},$$ 
    and 
    $$\Az := \{(y,z) \in \im(\alpha_2)  \mid \exists x \in \x \st (x,y) \in \im(\alpha_1) \}.$$ 

      Let
      $$C_1 := \max_{(x,z) \in \im(\alpha_2\circ \alpha_1)}\{|w_x - w_z|\}, \quad C_2 = \max_{(x, \zero) \in \im(\alpha_2\circ \alpha_1)}\{w_x\}, \quad C_3 = \max_{(\zero, z) \in \im(\alpha_2\circ \alpha_1)}\{w_z\}.$$
     Because~$d_{\mathcal{B}_{\infty}}(\x, \z)$ is computed from an optimal alignment
     between $\x$ and $\z$, then
     \[ d_{\mathcal{B}_\infty}(\x, \z) \leq \max \{ C_1, C_2, C_3 \}. \]

     Suppose $\max\{ C_1, C_2, C_3\} = C_1$. Let $(x, z) \in \im(\alpha_2 \circ \alpha_1)$ such that $| w_x - w_z| = C_1$. Observe since $(x,z) \in \im(\alpha_2 \circ \alpha_1)$ then there exists $y \in \y$ such that $(x, y) \in \im(\alpha_1)$ and $(y,z) \in \im(\alpha_2)$. 
     This observation and the triangle inequality from
     the $L_1$-norm implies
     \begin{equation*}
     \begin{split}
       |w_x - w_z| &\leq |w_x - w_y| +|w_y - w_z|\\
       &\leq d_{\mathcal{B}_{\infty}}(\x, \y) + d_{\mathcal{B}_{\infty}}(\y, \z).
     \end{split}
     \end{equation*}
     Hence, in the case $\max\{C_1, C_2, C_3\} = C_1$, we have
     $d_{\mathcal{B}_{\infty}}(\x, \z) \leq d_{\mathcal{B}_{\infty}}(\x, \y) +
     d_{\mathcal{B}_{\infty}}(\y, \z)$.

     Next, suppose $\max\{C_1, C_2, C_3\} = C_2$. Let $(x, \zero) \in \im(\alpha_2\circ \alpha_1)$ such that
     $w_x = C_2$. Define 
    $$\Xzero := \{(x, \zero) \in \im(\alpha_2\circ \alpha_1) \mid \exists y \in \y \st (x, y) \in \im(\alpha_1)\setminus \Ax\}.$$
    Observe, if $(x,y) \in \im(\alpha_1) \setminus \Ax$, then for all $z\in \z$, $(y,z) \notin \im(\alpha_2)$. Hence, $(y, \zero) \in \im(\alpha_2)$. This implies the set  
    $$(\x \times \{\zero\}) \cap \im(\alpha_2 \circ \alpha_1) = ((\x \times \{\zero\}) \cap \im(\alpha_1)) \cup \Xzero$$
    and this union is disjoint. Either $(x, \zero) \in  (\x \times \{\zero\}) \cap \im(\alpha_1)$ or $(x, \zero) \in \Xzero$. 
    If $(x, \zero) \in (\x \times \{\zero\}) \cap \im(\alpha_1)$, then 
    \[ C_2 = w_x \leq d_{\mathcal{B}_{\infty}}(\x, \y) \leq d_{\mathcal{B}_{\infty}}(\x, \y) +d_{\mathcal{B}_{\infty}}(\y, \z). \]
    Now suppose $(x, \zero) \in \Xzero$.  By definition of $\Xzero$, there exists $y \in \y$ such that $(x, y) \in \im(\alpha_1)\setminus \Ax$ and, at the same time,  $(y, \zero) \in \im(\alpha_2)$. Noting this observation and 
    applying the triangle inequality from the $L_1$-norm gives    
    \begin{equation*}
    \begin{split}
        w_x &\leq |w_x- w_y| + w_y\\
        &\leq d_{\mathcal{B}_{\infty}}(\x, \y) +d_{\mathcal{B}_{\infty}}(\y, \z).
    \end{split}
    \end{equation*}
    We can conclude that in the case $\max\{C_1, C_2, C_3\} = C_2$ that
     $d_{\mathcal{B}_{\infty}}(\x, \z) \leq d_{\mathcal{B}_{\infty}}(\x, \y) +
     d_{\mathcal{B}_{\infty}}(\y, \z)$.

     Lastly, suppose $\max\{C_1, C_2, C_3\} = C_3$. Let $(\zero, z) \in \im(\alpha_2\circ \alpha_1)$ such that \mbox{$w_z = C_3$}. 
      Define 
     $$\Zzero := \{(\zero, z) \in \im(\alpha_2\circ \alpha_1) \mid \exists y \in \y \st (y, z) \in \im(\alpha_2)\setminus \Az\}.$$
    Similarly, the set 
    $$(\{\zero\} \times \z) \cap \im(\alpha_2 \circ \alpha_1) = ((\{\zero\} \times \z) \cap \im(\alpha_2)) \cup \Zzero$$ and this union is disjoint. Either $(\zero, z) \in (\{\zero\} \times \z) \cap \im(\alpha_2)$ or $(\zero , z) \in \Zzero$. 
    If $(\zero, z) \in (\{\zero\} \times \z) \cap \im(\alpha_2)$, then 
    \[ C_3 = w_z \leq d_{\mathcal{B}}(\y, \z) \leq d_{\mathcal{B}}(\x, \y) + d_{\mathcal{B}}(\y, \z) . \]
    Now suppose $(\zero, z) \in \Zzero$. By definition of $\Zzero$, there exists $y\in \y$ such that $(y,z) \in \im(\alpha_2)\setminus \Az$ and $(\zero, y) \in \im(\alpha_1)$. Noting this observation and applying the triangle inequality from the $L_1$-norm, we get 
    \begin{equation*}
    \begin{split}
    w_z &\leq |w_z - w_y| + w_y\\
    &\leq d_{\mathcal{B}_{\infty}}(\y, \z) + d_{\mathcal{B}_{\infty}}(\x, \y).
    \end{split}
    \end{equation*}
     We can conclude in the case $\max\{C_1, C_2, C_3\} = C_3$ that
     $d_{\mathcal{B}_{\infty}}(\x, \z) \leq d_{\mathcal{B}_{\infty}}(\x, \y) +
     d_{\mathcal{B}_{\infty}}(\y, \z)$.

     Therefore, the backbone infinity distance satisfies the triangle inequality.
 \end{proof}

\begin{prop}[Backbone Infinity Distance is a Metric]
    Let $\x, \y, \z$ be backbones. The backbone infinity distance (\defref{backbone-infty-dist}) satisfies all properties of a metric.
\end{prop}

\begin{proof}
Recall that for each $x\in \x$, we can write $x=(s_x, w_x)$; likewise for
    $y\in\y$.
    Let~$\alpha: [k] \rightarrow \tilde{\x} \times
    \tilde{\y}$ be an optimal alignment. We verify all properties of a metric.
    
     \orgemph{Non-Negativity.} Since all node weights are non-negative,
    $d_{\mathcal{B}_{\infty}}(\x, \y) \geq 0$.
    
    \orgemph{Symmetry.} By construction, $d_{\mathcal{B}_{\infty}}$ is
    symmetric; see \defref{backbone-infty-dist}.
    
    \orgemph{Definiteness.} If $\x = \y$, then the optimal alignment
    aligns each node with itself, and there are no insertions. Hence, all node weights match and
    $d_{\mathcal{B}_{\infty}}(\x, \y) = 0$.
    
    On the other hand, assume $d_{\mathcal{B}_{\infty}}(\x, \y) =0$. This implies
    \[ 0 =  \inf_\alpha \max{|\omega_\x(\alpha_\x(i)) - \omega_\y(\alpha_\y(i))|}. \]
            Therefore, $|\omega_\x(\alpha_\x(i)) - \omega_\y(\alpha_\y(i))| \leq 0$
            for all $i \in [n]$. Furthermore, $|\omega_\x(\alpha_\x(i)) -
            \omega_\y(\alpha_\y(i))| \geq 0$ for all $i \in [n]$. Thus,
            $|\omega_\x(\alpha_\x(i)) - \omega_\y(\alpha_\y(i))| = 0$ for all $i \in
            [n]$. Hence, each aligned pair of nodes must have the same node
            weight. By \defref{alignment}, we never align two empty nodes. This
            implies each node in $\x$ is aligned with a node in $\y$.
            Furthermore, each aligned pair must have the same label by
            \defref{alignment}. We can conclude $\x=\y$.
            
       \orgemph{Triangle Inequality.} By \lemref{backbone-infty-triangle}, the triangle inequality holds.
       
    Therefore, the backbone infinity distance is a metric.
\end{proof}

\subsection{Lemmas Used For Proving \lemref{direct-alignment-optimal} that Direct Alignment is Optimal and Unique when Functions are Extremely Close.}\label{app:stability}

\begin{lem}[Extremal Pairs With Node Life Differences Greater Than $\varepsilon$]
    Let $f, f': C\rightarrow \R$ be nicely tame functions such that $f'$ is
    extremely close to $f$. Let $\varepsilon := \norm{f-f'}_{\infty}$. Let $(t, f(t))$ be a local minimum of $f$ and $(t', f'(t'))$ be a local minimum of $f'$.
    Suppose $p:= (f(t), \zeta_f(t)) \in D(f)\setminus \Delta$ and $q := (f'(t'),
    \zeta_{f'}(t')) \in D(f') \cap \square_{\varepsilon}(p)$. Then,
    \begin{enumerate}
        \item $\frac{1}{2}\pers_{f'}(t')>\varepsilon$.
        \item $|\frac{1}{2}\pers_{f'}(t')-\frac{1}{2}\pers_f(s)| >\varepsilon$ for all
            $(f(s), \zeta_f(s)) \in D(f) \setminus \square_{\varepsilon}(p)$.
    \end{enumerate}
    \label{lem:over-epsilon}
\end{lem}

\begin{proof}

    Suppose $p:= (f(t), \zeta_f(t)) \in D(f)\setminus \Delta$ and \mbox{$q := (f'(t'),
    \zeta_{f'}(t'))$ $\in D(f') \cap \square_{\varepsilon}(p)$}.

    \begin{enumerate}

        \item We first show $\frac{1}{2}\pers_{f'}(t') >\varepsilon$. Consider
            $\square_{4\varepsilon}(p)$. Since $f'$ is extremely close to $f$,
            $\varepsilon <\delta_f/2$. Hence, $4\varepsilon <2\delta_f$. Recall
            $2\delta_f$ is at most the smallest distance between any two points in
            $D(f)$, provided that at least one point is not on the diagonal.
            Therefore, $\square_{4\varepsilon}(p)$ does not intersect the
            diagonal. Consider a square of radius $4\varepsilon$ where the right
            bottom corner is the point  $(f(t)+4\varepsilon,f(t)+4\varepsilon)$,
            that is, $\square_{4\varepsilon}((f(t), f(t)+8\varepsilon))$. Since
            $q \in D(f') \cap \square_{\varepsilon}(p)$, the difference of the
            $y$- and $x$- coordinates of $q$ (which is the persistence of $q$)
            is bounded below by the difference of $y$- and $x$- coordinates of
            the point $(f(t)+\varepsilon, f(t)+7\varepsilon)$, see
            Figure~\ref{fig:over-epsilon-diag}. Hence,
            \[
                \frac{1}{2}\pers_{f'}(t')
                > \frac{1}{2}((f(t)+7\varepsilon) - (f(t)+\varepsilon))
                = 3\varepsilon
                > \varepsilon.
            \]

        \item Let $(f(s), \zeta_f(s)) \in D(f) \setminus
            \square_{\varepsilon}(p)$. Since $f'$ is extremely close to $f$,
            $\varepsilon <\delta_f/2$ and  hence $2\varepsilon <\delta_f$. Since
            $\delta_f$ is at most half the smallest distance between two points in
            $D(f)$ where at least one point is not on the diagonal,
            the squares of radius $2\varepsilon$ centered at
            different points in $D(f)\setminus \Delta$ are all disjoint. This
            implies
            \[\norm{ (f(t), \zeta_f(t)) - (f(s), \zeta_f(s)) }_{\infty}
            >2\varepsilon. \]
            Furthermore, by the direct alignment, we know $q \in
            \square_{\varepsilon}(p)$. This implies
            $q\notin\square_{2\varepsilon}((f(s), \zeta_f(s)))$. Consider the
            point $((f(s)-2\varepsilon), \zeta_f(s))$. From planar geometry (See
            \figref{over-epsilon-points}),
            \[|\pers_{f'}(t')-\pers_f(s)| > |(\zeta_f(s)-f(s)) - (\zeta_f(s)-(f(s)
            - 2\varepsilon))| = 2\varepsilon.\]
            Therefore, $\frac{1}{2}|\pers_{f'}(t') - \pers_f(s)| > \varepsilon$.

    \end{enumerate}

\end{proof}

\begin{figure}[htp]
    \centering
    \begin{subfigure}[b]{0.5\textwidth}
        \centering
        \includegraphics[width=\textwidth]{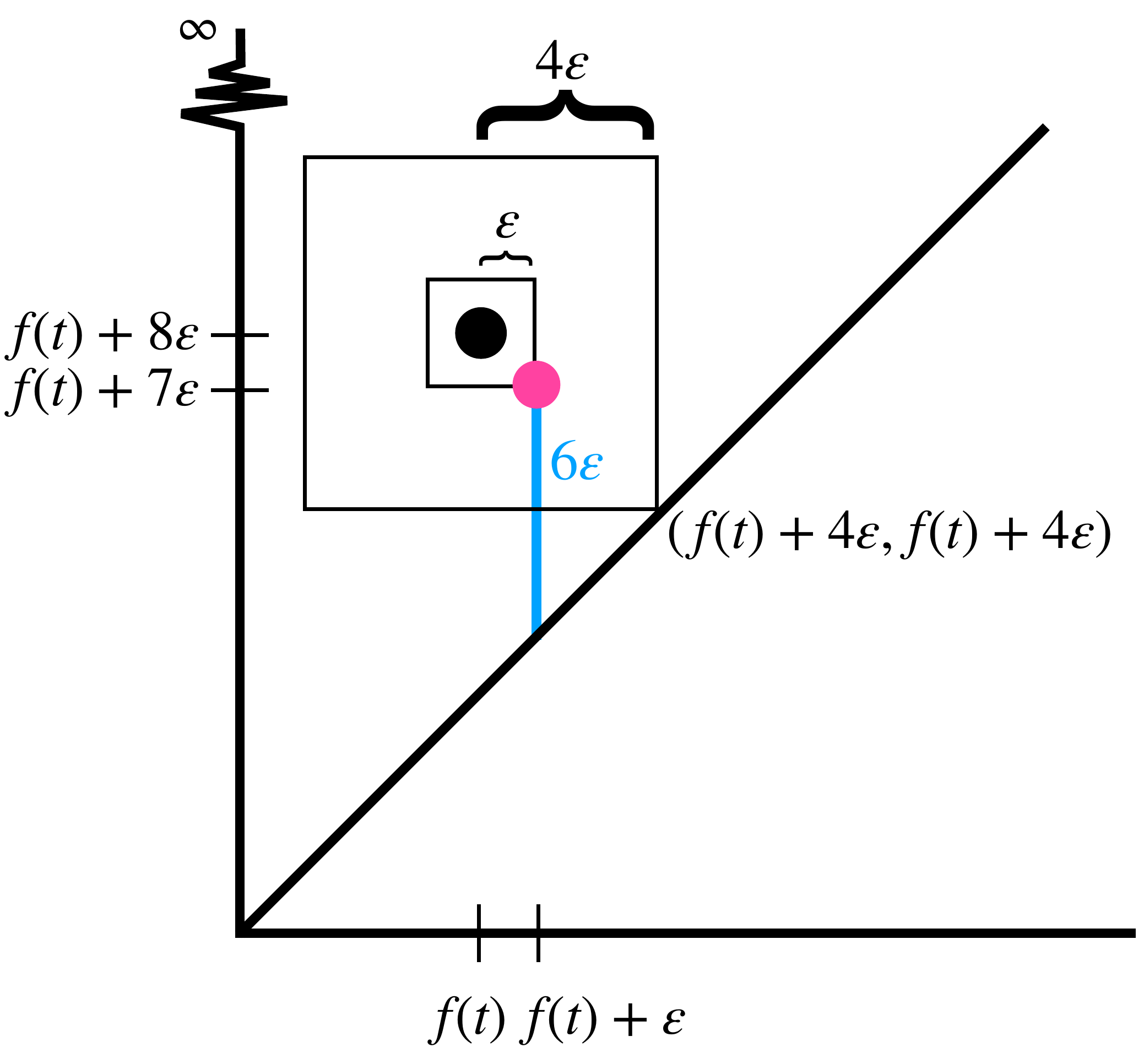}
        \caption{Case 1}
        \label{fig:over-epsilon-diag}
    \end{subfigure}
    \hfil
    \begin{subfigure}[b]{0.48\textwidth}
        \centering
        \includegraphics[width=\textwidth]{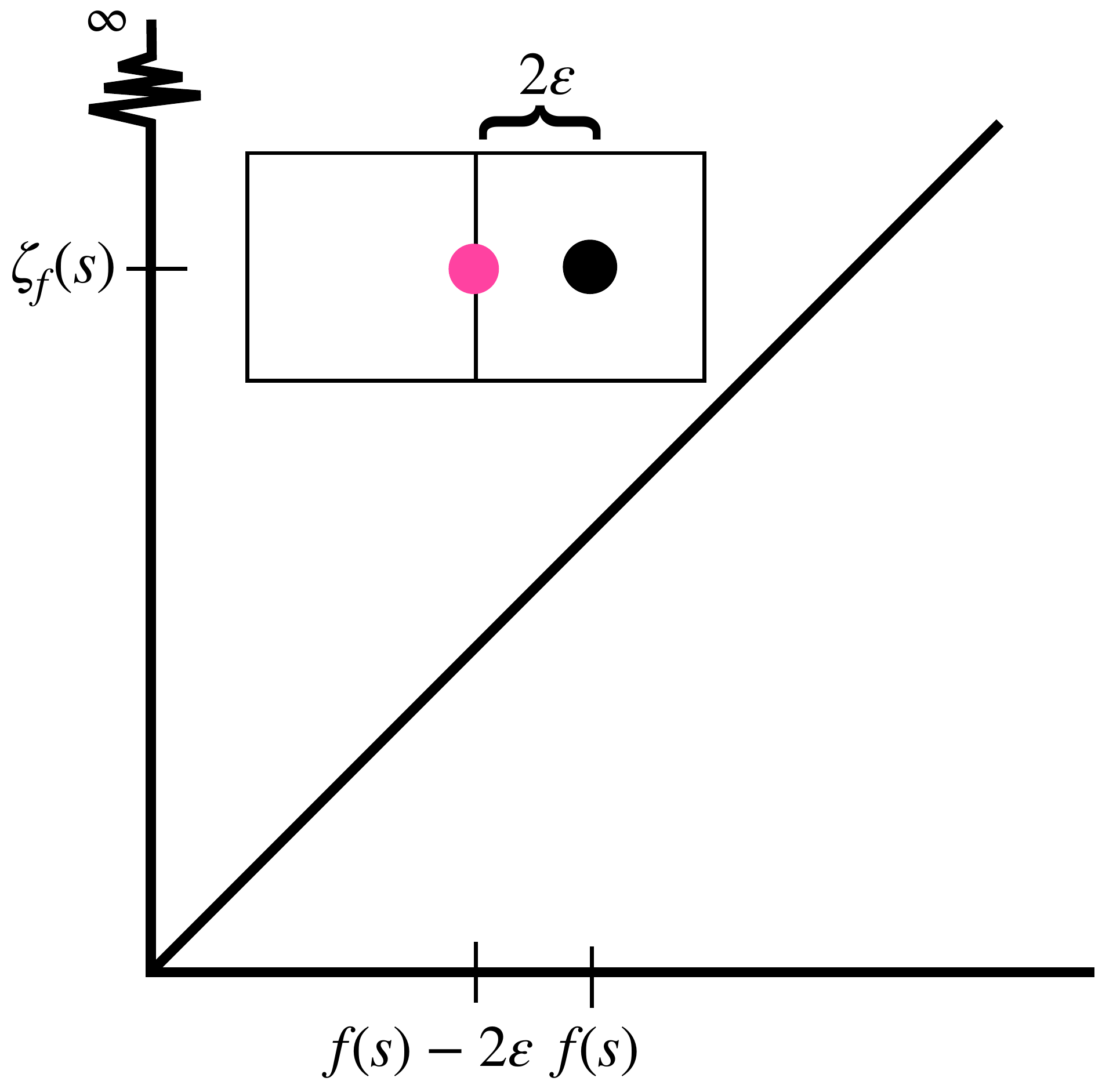}
        \caption{Case 2}
        \label{fig:over-epsilon-points}
    \end{subfigure}
     \hfil
    \caption{Geometric Arguments for \lemref{over-epsilon}.  In
        \figref{over-epsilon-diag},
        $\frac{1}{2}\pers_{f'}(t')>\frac{1}{2}((f(t)+7\varepsilon) -(f(t)+\varepsilon) = 3\varepsilon>\varepsilon.$ In
        \figref{over-epsilon-points},
        $\frac{1}{2}|\pers_{f'}(t')-\pers_{f}(s)|>\frac{1}{2}|(\zeta_f(s)-f(s))-(\zeta_f(s)-(f(s)-2\varepsilon))|=\varepsilon$.}
    \label{fig:over-epsilon}
\end{figure}

Next, we  need the following technical lemma that is used to prove the uniqueness and optimality of the direct alignment for extremely
close and nicely tame functions.

\begin{lem}[Bijections Within Boxes]\label{lem:app-stability}
    Let $f, f': C\rightarrow \R$ be nicely tame functions such that $f'$ is
    extremely close to $f$. Let $\alpha$ be the direct alignment defined in \conref{direct-alignment} between $\x := B(f)$ and $\x' := B(f')$. Suppose $\eta$ is a different alignment between $\x$ and $\x'$ such that $\cost(\eta) \leq \cost(\alpha)$.  For each $x' \in \x'$, let $x_{\alpha} \in \tilde{\x}$
    be the unique element such that~$(x_{\alpha},x') \in
    \im(\alpha)$. For each $x' \in \x'$, define $x_{\eta}$ similarly.
Suppose there exists $x' \in \x'$ such 
that 
\[ |w_{x'}-w_{x_{\eta}}| \leq |w_{x'}-w_{x_{\alpha}}| \]
 and $x_{\eta}\neq x_{\alpha}$. Then 
\begin{enumerate}
\item $|w_{x'}-w_{x_{\eta}}| = |w_{x'}-w_{x_{\alpha}}|.$ \label{stmt:app-stability-equality}
\item There exists $z' \in \x'$ for which $|w_{z'}-w_{x_{\eta}}|>\varepsilon$ where $\varepsilon := \norm{f-f'}_{\infty}$. \label{stmt:app-stability-greater}
\end{enumerate}
\end{lem}

\begin{proof}
\orgemph{We prove \stmtref{app-stability-equality}. }

            Let $(t', f'(t'))$ be the local extemum of $f'$ associated with node
            $x'$, and let $(t,f(t))$ be the local extremum of $f$ associated
            with $x_{\alpha}$. Without loss of generality, we
            asume $(t', f'(t'))$ is a local minimum. (Note that if $(t', f'(t'))$ is
            a local maximum, then we apply the same argument  with
            $-f'$ and~$-f$).
            Since $(x_{\alpha},x')\in \im(\alpha)$, we know that $(t,f(t))$ is
            also a local minimum.
            We first prove, by the way of contradiction,
            that neither $x_{\eta}$ nor
            $x_{\alpha}$ is the empty node.

            If $x_{\eta}$ is the empty node,
            then, by the assumption $x_{\eta} \neq
            x_{\alpha}$, it follows  that
            $x_{\alpha}$ is not the empty node.
            From \conref{direct-alignment}, we know that
            $(f'(t'), \zeta_{f'}(t')) \in D(f') \cap
            \square_{\varepsilon}(f(t), \zeta_f(t))$.

            Applying \lemref{over-epsilon}, we find
            \begin{equation}\label{here:bigger}
                w_{x'}>\varepsilon.
            \end{equation}
            On the other hand, by  assumption,
            $$|w_{x'} - w_{x_{\eta}}|
                \leq |w_{x'} - w_{x_{\alpha}}| \leq
                \varepsilon,
            $$
            where  the last inequality follows from
            \lemref{nodeweights-easyalignment}. Finally, using
            \eqref{here:bigger} and that $w_{x_{\eta}} =0$,
            we have
            $$| w_{x'} - w_{x_{\eta}}|
                = w_{x'}>\varepsilon,
            $$
            giving a contradiction. Therefore, we conclude that
            $x_{\eta}$ is not the empty node.

            If $x_{\alpha}$ is the empty node, then by the same
            argument as above, we arrive at contradiction.
            Therefore,~$x_{\alpha}$ is also not the empty node.

            Therefore, neither $x_{\eta}$ nor
            $x_{\alpha}$ is an empty node.
            Let~$(s_{\eta}, f(s_{\eta}))$ be the local extremum
            of $f$ associated with node $x_{\eta}$.
            By
            \lemref{nodeweights-easyalignment},~$|w_{x'} -w_{x_{\alpha}}| \leq \varepsilon$.
            Additionally, consider the point $p = (f(s_{\eta}),
            \zeta_f(s_{\eta})) \in D(f)$ and $\square_{\varepsilon}(p)$.  In
            \conref{direct-alignment}, we established a bijection between
            the multiplicity of $p$, denoted $\mu(p)$, and the number
            of points contained in $D(f') \cap \square_{\varepsilon}(p)$.
            Additionally, in \lemref{nodeweights-easyalignment}, we showed for all points $(f'(x), \zeta_{f'}(x)) \in D(f') \cap
            \square_{\varepsilon}(p)$, we have
            $ \frac{1}{2}|\pers_{f'}(x) - \pers_f(s_{\eta})| \leq \varepsilon$. Furthermore, by
            \lemref{over-epsilon},~$\frac{1}{2}|\pers_{f'}(y) - \pers_f(s_{\eta})| >
            \varepsilon$ for all~$(f'(y), \zeta_{f'}(y)) \in D(f') \setminus
            \square_{\varepsilon}(p)$. By assumption,
            \begin{equation}\label{here:boxbound}
                |w_{x'} - w_{x_{\eta}}| \leq |w_{x'} - w_{x_{\alpha}}| \leq \varepsilon.
            \end{equation}
            Let~$(s_{\alpha}, f(s_{\alpha}))$ be the local extremum
            of $f$ associated with node $x_{\alpha}$. Since the only
            point of $D(f)$ contained in the square
            $\square_{\varepsilon}(p)$ is the point $p$
            with multiplicity $\mu(p)$,  we must have $(f(s_{\eta}),
            \zeta_f(s_{\eta})) = (f(s_{\alpha}), \zeta_f(s_{\alpha}))$ in order for \eqref{here:boxbound} to hold.  Therefore, the two extrema of $f$ have
            the same node lives and thus
            $w_{x_{\eta}} =
            w_{x_{\alpha}}$. We can
            conclude
            \[|w_{x'} -w_{x_{\eta}}| = |w_{x'} -w_{x_{\alpha}}|, \]
            as was to be shown to prove \stmtref{app-stability-equality}.

    \orgemph{Next we prove \stmtref{app-stability-greater}.}
            We note that the
            $\varepsilon$-extremal intervals discussed in this proof are
            all constructed from $f$.
            We claim there exists a local extremum $(s', f'(s'))$ of $f'$ such
            that $(f'(s'), \zeta_{f'}(s')) \in D(f')\cap \square_{\varepsilon}(p)$
            and is aligned via $\eta$ with an extremum $(s, f(s))$ such that
            $(f(s), \zeta_f(s)) \in D(f) \setminus
            \square_{\varepsilon}(p)$. By way of contradiction, suppose
            that is not the case.
            Hence, all persistence points in $D(f')\cap \square_{\varepsilon}(p)$ are paired with persistence points in $D(f) \cap \square_{\varepsilon}(p)$. Since $\eta\not =\alpha$, $\eta$ restricted to $D(f) \cap
            \square_{\varepsilon}(p)$  is a bijection onto $D(f') \cap
            \square_{\varepsilon}(p)$ that is different from the
            bijection $\alpha$. In other
            words, if we denote $\gamma:= \alpha_{| D(f)\cap
            \square_{\varepsilon}(p)}$ and~$\gamma':= \eta_{|
            D(f)\cap \square_{\varepsilon}(p)}$, we have $\gamma \not =
            \gamma'$. By \conref{direct-alignment} of $\alpha$ and noting $\eta
            \not = \alpha$, there exists an extremum~$(s', f'(s'))$ that aligns via $\eta$ with
            an extremum $(s, f(s))$ such that $s'$ does not belong to the
            interval of size $\varepsilon$ around $s$; that is, $s'
            \notin \varphi_{\varepsilon}(s)$.  Without loss of
            generality, suppose $s < s'$. Since $f'$ is extremely close
            to~$f$, the number of extrema with persistence points
            contained in $(D(f) \cup D(f'))\cap\square_{\varepsilon}(p)$
            where the domain coordinates are greater than $s$ is the
            same for both $f$ and $f'$. Let
            \[  A:= \{ (t,f(t)) \;|\; (f(t),\zeta_f(t)) \in D(f) \cap
            \square_{\varepsilon}(p) \mbox{ and } t>s \} \]
            be the set of extrema of $f$ whose persistence
            points are contained in $D(f) \cap \square_{\varepsilon}(p)$
            such that the domain coordinates of these extrema are
            greater than $s$.  Additionally, let
            \[  B:=  \{ (t,f'(t)) \;|\; (f'(t),\zeta_{f'}(t)) \in D(f') \cap \square_{\varepsilon}(p) \mbox{ and } t>s' \} \]

            To preserve order in $\eta$, elements of $A$ must be aligned with the
            elements of $B$. Since~$|A| > |B|$, by the pigeonhole principle,
            at least two extrema of~$f$ are aligned by $\eta$ with the same extremum of $f'$.
            This contradicts the fact that $\eta$ is a bijection.
            This contradiction shows that there exists a $z \in \x'$
            for which \mbox{$|w_{z'} - w_{z_{\eta}}| >
            \varepsilon.$}
            By construction of the direct alignment, $|w_{z'}
            -w_{z_{\eta}}| \leq \varepsilon.$ Therefore,
            $|w_{z'} - w_{z_{\eta}}| > |w_{z'}
            -w_{z_{\eta}}| $.

\end{proof}

Next, we prove the uniqueness and optimality of the direct alignment for extremely
close and nicely tame functions.

\begin{lem}[Direct Alignment Gives Optimal Backbone Alignment]
    Let $f, f': C\rightarrow \R$ be nicely tame functions such that $f'$ is
    extremely close to $f$. Then, the direct alignment defined in
    \conref{direct-alignment} is the unique optimal alignment that realizes
    $d_{\mathcal{B}}(B(f), B(f'))$.
    \label{lem:direct-alignment-optimal}
\end{lem}

\begin{proof}
    Let $\alpha$ be the direct alignment between $\x:=B(f)$ and $\x':=B(f')$.
    Recall in \defref{backbone} that each $x \in \x$ can be written as a tuple $x=(s_x,w_x)$;
    likewise, we can write $x' \in \x'$ as $x'=(s_{x'},w_{x'})$.
    By way of
    contradiction, suppose~$\eta$ is a different alignment between $\x$ and
    $\x'$ such that $\cost(\eta) \leq \cost(\alpha)$.
    Recall from the construction of the direct alignment (\conref{direct-alignment})
    that the length of the direct alignment is
    the number of nodes in~$\x'$. By \conref{direct-alignment}, 
    for each $x' \in \x'$, there exists a unique $x_{\alpha} \in \tilde{\x}$
    such that~$(x_{\alpha},x') \in
    \im(\alpha)$.
    Hence, we can write
    \[
        \cost(\alpha) = \sum_{(x,x')\in \im(\alpha)} |w_{x}-w_{x'}|
        = \sum_{x' \in \x'} |w_{x'} - w_{x_{\alpha}}|.
    \]
    Since $\eta$ is an alignment, it must
    align all nodes of $\x'$.  Hence, we have \mbox{$\length{\eta} \geq \length{\x'}
    =\length{\alpha}$}.  Let $x_{\eta}$ denote the unique element of $\tilde{\x}$
    such that $(x_{\eta},x')\in \im(\eta)$.
    We now discuss the following  logical dichotomy: either
    \begin{itemize}
        \item[Case 1:] for all $x' \in \x'$,
            $|w_{x'} - w_{x_{\eta}}|
            >  |w_{x'} - w_{x_{\alpha}}|$, or
        \item[Case 2:] there exists at least one $x' \in \x'$ for which
            $|w_{x'} - w_{x_{\eta}}|
            \leq  |w_{x'} - w_{x_{\alpha}}|$.
    \end{itemize}

    \orgemph{We first consider Case 1.}
    Suppose, for all $x' \in \x'$, we have
    \mbox{$|w_{x'} - w_{x_{\eta}}|
            >  |w_{x'} - w_{x_{\alpha}}|$}.
    Since $\eta$ aligns all nodes of $\x'$, we have the first inequality below
    \begin{equation*}
    \begin{split}
        \cost(\eta)
            &\geq \sum_{x' \in \x'} |w_{x'} -
            w_{x_{\eta}}|\\
        &>\sum_{x' \in \x'} |w_{x'} - w_{x_{\alpha}}|\\
        &=\cost(\alpha).
    \end{split}
    \end{equation*}

    This is a contradiction with $\cost(\eta) \leq \cost(\alpha)$.

    \orgemph{Now, we consider Case 2.} Suppose there exists $x' \in \x'$ for which
    \begin{equation}
        |w_{x'} - w_{x_{\eta}}|
            \leq  |w_{x'} - w_{x_{\alpha}}|.
        \label{eq:mismatch-boxes}
    \end{equation}

    If $x_{\eta} =x_{\alpha}$ for all such $x' \in \x'$,
    then we get that either $\eta$ is the same alignment as $\alpha$
    or there must exist~$y' \in \x'$ for which
    (\ref{eq:mismatch-boxes}) is not true and hence
    \mbox{$|w_{y'}-w_{y_{\eta}}| >
    |w_{y'} - w_{y_{\alpha}}|$}. 
In particular, since for all instances where $|w_{x'} - w_{x_{\eta}}| \leq  |w_{x'} - w_{x_{\alpha}}|$, 
$x_{\eta}=x_{\alpha}$, this implies  $|w_{x'} - w_{x_{\eta}}| = |w_{x'} - w_{x_{\alpha}}|$. For all other  
$y' \in \x'$, we must have $|w_{y'} - w_{x_{\eta}}| >  |w_{y'} - w_{x_{\alpha}}|$. Therefore in this case,
    ~$\cost(\eta)>\cost(\alpha)$. In either case, we have a
    contradiction with $\cost(\eta) \leq \cost(\alpha)$.

    Hence, there must exist some $x' \in \x'$ for which
    \eqref{eq:mismatch-boxes} holds and $x_{\eta} \neq x_{\alpha}$.
    
    In \lemref{app-stability}, 
    we show that both of the following statements are true:
    \begin{itemize}
        \item[(a)] $|w_{x'} - w_{x_{\eta}}|
            =  |w_{x'} - w_{x_{\alpha}}|$.
        \item[(b)] There exists a $z' \in \x'$ for which
            $|w_{z'} - w_{x_{\eta}}| >\varepsilon$ where $\varepsilon := \norm{f-f'}_{\infty}$. 
          \end{itemize}
    From (a) and (b), we find that for all $x' \in \x'$ for which
    $|w_{x'} - w_{x_{\eta}}| \leq  |w_{x'}- w_{x_{\alpha}}|$
    and $x_{\eta} \neq x_{\alpha}$, the equality
    $|w_{x'} - w_{x_{\eta}}| =  |w_{x'}- w_{x_{\alpha}}|$
    must hold.
    Furthermore, there exists $z' \in \x'$ for which
    $$|w_{z'} - w_{z_{\eta}}| > |w_{z'}- w_{z_{\alpha}}|.$$
    For all remaining $y' \in \x'$, we have
    $$|w_{y'} - w_{y_{\eta}}| \geq  |w_{y'} - w_{y_{\alpha}}|.$$
    Putting this all together, we obtain
    \begin{equation*}
    \begin{split}
        \cost(\eta) &\geq \sum_{x' \in \x'} |w_{x'} - w_{x_{\eta}}|\\
        &>\sum_{x' \in \x'} |w_{x'} -  w_{x_{\alpha}}|\\
        &=\cost(\alpha),
    \end{split}
    \end{equation*}
    which again contradicts the assumption that $\cost(\eta) \leq
    \cost(\alpha)$.

    We conclude that any alignment that is different from the direct alignment has
    a higher cost. Therefore, the direct alignment  is the unique optimal
    alignment that realizes $d_{\mathcal{B}}(\x, \x')=d_{\mathcal{B}}(B(f), B(f'))$.

\end{proof}

\subsection{Cases for Bound Differences in Edge Weights in \lemref{edge-bound}} \label{app:bound-edges}
    \begin{enumerate}
        \item $E_{\diff}= \frac{1}{2}|\pers(t) - \pers(t')|$. Applying \lemref{node-weights-bound}, we find
            \[ E_{\diff} = \frac{1}{2}|\pers(t) - \pers(t')| \leq \varepsilon_i \leq \varepsilon_{i,j}. \]

            \vspace{1ex}

        \item $E_{\diff} = \frac{1}{2}|\pers(s) - \pers(s')|$. Applying \lemref{node-weights-bound}, we find
            \[ E_{\diff} = \frac{1}{2}|\pers(s) - \pers(s')| \leq  \varepsilon_j \leq \varepsilon_{i,j}. \]

            \vspace{1ex}

        \item $E_{\diff} = |\varepsilon^*(t, s) - \varepsilon^*(t', s')|$. Applying \lemref{intersection-difference-bound}, we find
            \[ E_{\diff} =  |\varepsilon^*(t, s) - \varepsilon^*(t', s')| \leq  \varepsilon_{i,j}. \]

            \vspace{1ex}

        \item $E_{\diff} = \frac{1}{2}(\pers(t) - \pers(s'))$. Then, $\pers(t) \leq \pers(s)$. Applying \lemref{node-weights-bound}, we find
            \[ E_{\diff} = \frac{1}{2}(\pers(t) - \pers(s')) \leq \frac{1}{2}(\pers(s) - \pers(s')) \leq \varepsilon_j \leq \varepsilon_{i,j}. \]

            \vspace{1ex}

        \item $E_{\diff} = \frac{1}{2}(\pers(s') - \pers(t))$. Then, $\pers(s') \leq \pers(t')$. Applying \lemref{node-weights-bound}, we find
            \[ E_{\diff} = \frac{1}{2}(\pers(s') - \pers(t)) \leq \frac{1}{2}(\pers(t') - \pers(t)) \leq \varepsilon_i \leq \varepsilon_{i,j}. \]

            \vspace{1ex}

        \item $E_{\diff} = \frac{1}{2}(\pers(s) - \pers(t'))$. Then, $\pers(s) \leq \pers(t)$. Applying \lemref{node-weights-bound}, we find
            \[ E_{\diff} = \frac{1}{2}(\pers(s) - \pers(t')) \leq \frac{1}{2}(\pers(t) - \pers(t')) \leq \varepsilon_i \leq \varepsilon_{i,j}. \]

            \vspace{1ex}

        \item $E_{\diff} = \frac{1}{2}(\pers(t') - \pers(s))$.Then, $\pers(t') \leq \pers(s')$. Applying \lemref{node-weights-bound}, we find
            \[ E_{\diff} = \frac{1}{2}(\pers(t') - \pers(s)) \leq \frac{1}{2}(\pers(s') - \pers(s)) \leq \varepsilon_j \leq \varepsilon_{i,j}. \]

            \vspace{1ex}

        \item $E_{\diff} = \varepsilon^*(t, s) - \frac{1}{2}\pers(t')$. Then, $ \varepsilon^*(t, s) \leq \frac{1}{2}\pers(t).$ Applying \lemref{node-weights-bound}, we find
            \[ E_{\diff} =  \varepsilon^*(t, s) - \frac{1}{2}\pers(t') \leq \frac{1}{2}(\pers(t) - \pers(t')) \leq \varepsilon_i \leq \varepsilon_{i,j}. \]

            \vspace{1ex}

        \item $E_{\diff} = \frac{1}{2}\pers(t') - \varepsilon^*(t, s)$. Then, $\frac{1}{2}\pers(t') \leq \varepsilon^*(t', s')$. Applying \lemref{intersection-difference-bound}, we find
            \[ E_{\diff} = \frac{1}{2}\pers(t') - \varepsilon^*(t, s) \leq \varepsilon^*(t', s') - \varepsilon^*(t, s) \leq \varepsilon_{i,j}. \]

            \vspace{1ex}

        \item $E_{\diff} = \varepsilon^*(t, s) - \frac{1}{2}\pers(s')$. Then, $\varepsilon^*(t, s) \leq \frac{1}{2}\pers(s).$ Applying \lemref{node-weights-bound}, we find
            \[ E_{\diff} =  \varepsilon^*(t, s) - \frac{1}{2}\pers(s') \leq \frac{1}{2}(\pers(s) - \pers(s')) \leq \varepsilon_j \leq \varepsilon_{i,j}. \]

            \vspace{1ex}

        \item $E_{\diff} = \frac{1}{2}\pers(s') - \varepsilon^*(t, s)$. Then, $\frac{1}{2}\pers(s') \leq \varepsilon^*(t', s')$. Applying \lemref{intersection-difference-bound}, we find
            \[ E_{\diff} = \frac{1}{2}\pers(s') - \varepsilon^*(t, s) \leq \varepsilon^*(t', s') - \varepsilon^*(t, s) \leq \varepsilon_{i,j}. \]

            \vspace{1ex}

        \item $E_{\diff} = \varepsilon^*(t', s') - \frac{1}{2}\pers(t)$. Then, $\varepsilon^*(t', s') \leq \frac{1}{2}\pers(t')$. Applying \lemref{node-weights-bound}, we find
            \[ E_{\diff} =  \varepsilon^*(t', s') - \frac{1}{2}\pers(t) \leq \frac{1}{2}(\pers(t') - \pers(t)) \leq \varepsilon_i \leq \varepsilon_{i,j}. \]

            \vspace{1ex}

        \item $E_{\diff} = \frac{1}{2}\pers(t) - \varepsilon^*(t', s')$. Then $\frac{1}{2}\pers(t) \leq  \varepsilon^*(t, s)$. Applying \lemref{intersection-difference-bound}, we find
            \[ E_{\diff} = \frac{1}{2}\pers(t) - \varepsilon^*(t', s') \leq \varepsilon^*(t, s) - \varepsilon^*(t', s') \leq \varepsilon_{i,j}. \]

    \vspace{1ex}

    \item $E_{\diff} = \varepsilon^*(t', s') - \frac{1}{2}\pers(s)$. Then, $\varepsilon^*(t', s') \leq \frac{1}{2}\pers(s')$. Applying \lemref{node-weights-bound}, we find
        \[ E_{\diff} =  \varepsilon^*(t', s') - \frac{1}{2}\pers(s) \leq \frac{1}{2}(\pers(s') - \pers(s)) \leq \varepsilon_j \leq \varepsilon_{i,j}. \]

    \vspace{1ex}

    \item $E_{\diff} = \frac{1}{2}\pers(s) - \varepsilon^*(t', s')$. Then $\frac{1}{2}\pers(s) \leq  \varepsilon^*(t, s)$. Applying \lemref{intersection-difference-bound}, we find
        \[ E_{\diff} = \frac{1}{2}\pers(s) - \varepsilon^*(t', s') \leq \varepsilon^*(t, s) - \varepsilon^*(t', s') \leq \varepsilon_{i,j}. \]

    \end{enumerate}
    
\subsection{Applications}\label{app:tables}
We provide tables summarizing the results of the computations described in \secref{parasite-data}.

\begin{table}[h!]
\begin{tabular}{|c|c|c|c|}
\hline
Distance & Mean & Median & Standard Deviation \\ \hline
$d_{ED}(\text{DAG}(\hat{\mathcal{D}}_1), \text{DAG}(\hat{\mathcal{D}}_2))$ & 123.13 & 122.68 & 22.12 \\ \hline
$d_{ED}(\text{DAG}(\hat{\mathcal{D}}_1), \text{DAG}(\hat{\mathcal{D}}_2'))$ & 267.97 & 268.18 & 28.68 \\ \hline
$d_{ED}(\text{DAG}(\hat{\mathcal{D}}_1), \text{DAG}(\hat{\mathcal{D}}_3))$ & 143.44 & 142.20 & 30.37 \\ \hline
$d_{ED}(\text{DAG}(\hat{\mathcal{D}}_1), \text{DAG}(\hat{\mathcal{D}}_3'))$ & 301.70 & 302.38 & 30.55 \\ \hline
$d_{ED}(\text{DAG}(\hat{\mathcal{D}}_1), \text{DAG}(\hat{\mathcal{D}}_4))$ & 282.00 & 279.72 & 56.59 \\ \hline
$d_{ED}(\text{DAG}(\hat{\mathcal{D}}_1), \text{DAG}(\hat{\mathcal{D}}_4'))$ & 424.06 & 421.20 & 53.24 \\ \hline
\end{tabular}
\caption{Summary of Results from Parasite Data.}
\end{table}

\begin{table}[H]
\begin{tabular}{|c|c|c|c|}
\hline
Distance & Mean & Median & Standard Deviation \\ \hline
$d_{ED}(\text{DAG}(\hat{\mathcal{D}}_a), \text{DAG}(\hat{\mathcal{D}}_b))$ & 493.40 & 490.88 & 71.38 \\ \hline
$d_{ED}(\text{DAG}(\hat{\mathcal{D}}_a), \text{DAG}(\hat{\mathcal{D}}_b'))$ & 642.51 & 640.78 & 66.70 \\ \hline
$d_{ED}(\text{DAG}(\hat{\mathcal{D}}_a), \text{DAG}(\hat{\mathcal{D}}_c))$ & 436.17 & 432.42 & 59.62 \\ \hline
$d_{ED}(\text{DAG}(\hat{\mathcal{D}}_a), \text{DAG}(\hat{\mathcal{D}}_c'))$ & 607.97 & 607.06 & 60.33 \\ \hline
\end{tabular}
\caption{Summary of Results from Mouse Data.}
\end{table}

\section{Supplementary Materials - Computations of Extremal Event DAGs and Extremal Event Supergraphs}
\label{sec:algorithms}

We describe how to compute the extremal event DAG and extremal event DAG distance. Since
experimental time series data often collects discrete time points as opposed to
continuous functions, we first take a detour to discuss $\varepsilon$-extremal
intervals and their associated properties in the discrete setting.

\subsection{Discrete $\varepsilon$-Extremal Intervals}

\begin{defn}[Collection of Time Series]
    A set $\mathcal{D}=\{D_j\}_{j=1}^K$ is a \emph{dataset} composed of
    \emph{time series} $D_j$ on the closed interval $C := [a,b]$ where $D_j =
    \{(z_i, h_i^j)\}_{j=1}^N$ with
    \[ Z := \{z_1 = a, z_2, \dots, z_{N-1}, z_N= b\}, \]
    is an ordered set with $z_j < z_{j+1}$ and the heights $h_j^i$ are the
    heights of the $j^{th}$ points at $z_j$ of time series $i$.
\label{def:time-series}
\end{defn}

For our purposes, we assume that $Z$ denotes the progression of time. However,
the following results hold even when $Z$ denotes some ordered quantity other
than time, for example, distance. 

\begin{defn}[Discrete $\varepsilon$-Extremal Intervals]
    Let $f_i:[a,b] \rightarrow \R$ be the linear interpolation of the time
    series $D_i$. Let $\varepsilon>0$, and suppose $f_i$ has a local extremum at
    $t$. Define the \emph{discrete $\varepsilon$-extremal interval} to be a
    relatively open interval  \mbox{$\d_{\varepsilon}^{f_i}(t)\subset [a,b]$} with
    endpoints in $Z$ such that
    \begin{enumerate}
        \item $\d_{\varepsilon}^{D_i}(t) \supset \varphi_{\varepsilon}^{f_i}(t)$
        \item $\d_{\varepsilon}^{D_i}(t)$ is the minimal such interval, meaning
            there does not exist an interval $I$ with endpoints in $Z$ such that
            $\d_{\varepsilon}^{D_i}(t) \supsetneq I \supset \varphi_{\varepsilon}^{f_i}(t)$.
    \end{enumerate}
\label{def:discrete-extremal-intervals}
\end{defn}

We note that we omit the superscript $D_i$ from $\d_{\varepsilon}^{D_i}$ if the function is clear.

A few properties of $\varepsilon$-extremal intervals still hold in the discrete
case. Namely, Propositions 1 and 2 of \cite{BerryUsing20}. Proposition 1 implies
that as long as $\varepsilon$ does not exceed the node life of two local minima
$(t_j, f_i(t_j))$, $(t_k, f_i(t_k))$ of $D_i$, then \mbox{$\d_{\varepsilon}(t_j) \cap
\d_{\varepsilon}(t_k) = \emptyset$}. Proposition 2 implies that as long as
$\varepsilon$ does not exceed the node of life of the extremum at $t_j$, then
any $\varepsilon$-perturbation of $f_i$ has a local minimum contained in
$\d_{\varepsilon}(t_j)$.  The $\varepsilon$-extremal interval property that is
lost is minimality, which is Proposition 3 of \cite{BerryUsing20}. This
proposition states that the $\varepsilon$-extremal intervals are the smallest
intervals to guarantee extrema of $\varepsilon$-perturbations of $f_i$, given
the node lives of extrema is less than $\varepsilon$. Since discrete
$\varepsilon$-extremal intervals contain the ones we would get from the linear
interpolations, we cannot guarantee minimality. A more thorough discussion of
these properties in the discrete case can be found in Section 4.2 of
\cite{BerryUsing20}.

In regards to properties mentioned in this paper, \stmtref{monotonicity} of
\lemref{properties} states $\varepsilon$-extremal intervals grow as
$\varepsilon$ increases also holds in the discrete setting. This is because the
computation of the discrete $\varepsilon$-extremal intervals is the same as the
continuous case except that the intervals are widened so that the endpoints are
contained in the domain of the time series. This added computation does not
affect the monotonicity of growth in the $\varepsilon$-extremal intervals.

\begin{lem}[Monotonicity of $\d_{\varepsilon}(t_i)$]
    Let $D_j$ be a time series with $t_1<t_2<...<t_n$
    the domain coordinates of the local extrema.
    Then, for each $i \in [n]$,
    the length $\length(\d_{\varepsilon}(t_i))$ is increasing with
    respect to $\varepsilon$.
    \label{lem:discrete-monotonicity}
\end{lem}

Additionally, \stmtref{extrema-containment} of \lemref{properties} holds in the
discrete case. We prove that here since we apply it in
\secref{extremal-DAG-computation}.

\begin{lem}[Containment Property of Discrete $\varepsilon$-Extremal Intervals]
Let $D_j$ be a time series with domain coordinates $\{t_i\}_{i=1}^n$ for
$i\in [n]$ where $t_1<t_2<...<t_n.$ Suppose $D_j$ has a local minimum at $t_i$ where $i\leq n-1$.
Then, $\varepsilon \leq \frac{1}{2}|f_j(t_i)-f_j(t_{i+1})|$ if and only if
$t_{i+1} \notin \d_{\varepsilon}(t_i)$.
\label{lem:discrete-extremal-containment}
\end{lem}

\begin{proof}
    Let $t_i$ be the domain coordinate of a local minimum of $D_j$. For this
    proof we omit the subscript $j$ from $D_j$ and its linear interpolation
    $f_j$.

    First assume $\varepsilon \leq \frac{1}{2}|f(t_i)-f(t_{i+1})|$.  We show
    $t_{i+1} \notin \d_{\varepsilon}(t_i)$. Since $t_i$ is the domain coordinate of a local minimum, $f(t_i) <f(t_{i+1})$. Thus,
    \[ \varepsilon \leq \frac{1}{2} (f(t_{i+1})-f(t_i)) \]
    \[ f(t_i)+\varepsilon \leq f(t_{i+1})-\varepsilon. \]

    This implies $t_{i+1}\notin (f-\varepsilon)^{-1}(-\infty, f(t_i)+\varepsilon)$. Hence, $t_{i+1}\notin \varphi_{\varepsilon}(t_i)$. 
    By \defref{discrete-extremal-intervals},
    the right endpoint of $\d_{\varepsilon}(t_i)$ is equal to $t_{i+1}$, and the
    right half of this interval is open. Hence, $t_{i+1} \notin
    \d_{\varepsilon}(t_i)$.

    Next we prove that $t_{i+1} \notin \d_{\varepsilon}(t_i)$ implies
    $\varepsilon \leq \frac{1}{2}|f(t_i)-f(t_{i+1})|$. We do this by proving the
    contrapositive. Assume $\varepsilon > \frac{1}{2}|f(t_i)-f(t_{i+1})|$.  Since $t_i$ is the domain coordinate of a local minimum, $f(t_i)<f(t_{i+1})$. Thus,
    \[ \varepsilon >\frac{1}{2}(f({t_{i+1}})-f(t_i)) \]
    \[ f(t_i)+\varepsilon > f(t_{i+1})-\varepsilon. \]
    This implies $t_{i+1} \in (f-\varepsilon)^{-1}(-\infty, f(t_i)+\varepsilon)$. Furthermore, $t_{i+1}$ must be in the same connected component as $t_i$ in  $(f-\varepsilon)^{-1}(-\infty, f(t_i)+\varepsilon)$ since $t_{i+1}$ is adjacent to $t_i$. Therefore, $t_{i+1} \in \varphi_{\varepsilon}(t_i)$. By \defref{discrete-extremal-intervals},
    $\d_{\varepsilon}(t_i) \supset \varphi_{\varepsilon}(t_i)$. Therefore,
    $t_{i+1} \in \d_{\varepsilon}(t_i)$. 
\end{proof}

Applying a symmetric argument as in \lemref{discrete-extremal-containment}, we see
    that \mbox{$\varepsilon \leq \frac{1}{2}|f_j(t_i)-f_j(t_{i-1})|$} if and only if
    $t_{i-1} \notin \d_{\varepsilon}(t_i)$ and $i \geq 2$. 
    
In discrete time series, the idea of incomparability is reduced to intersections
between a finite number of intervals. A linear interpolation of a time series giving rise to $\varepsilon$-extremal
intervals $\varphi_\varepsilon$ results in greater values of $\varepsilon^*$
than the discrete intervals $d_\varepsilon$. Therefore
$\varepsilon^*$ determined by $d_\varepsilon$ is a conservative estimate
of incomparability available from the time series information. 

Lastly, we remark that local stability of the extremal event DAG distance
extends to the discrete case. To see this, note the node lives of extrema in
discrete functions can be computed from the sublevel set persistence diagram
obtained through linearly interpolating the values of the discrete function.
Therefore, \lemref{over-epsilon}, \lemref{direct-alignment-optimal},
\lemref{node-weights-bound}, \corref{local-backbone-stability},
\lemref{aligned-heights-bound} that are all statements about node lives and
differences in aligned node lives all extend to the discrete case. Furthermore,
the proof of \lemref{intersection-difference-bound} (that bounds the difference
between $\varepsilon$-extremal intersection values between aligned edge weights)
relies on the nesting property of $\varepsilon$-extremal intervals. This
property still holds in the discrete case and so
\lemref{intersection-difference-bound} also extends in the discrete case. The
proofs of \lemref{edge-bound} and \thmref{extremal-DAG-stability} use the
aforementioned lemmas. Therefore local stability for extremal event DAGs holds
for discrete time series.

\subsection{Computing the Extremal Event DAG}
\label{sec:extremal-DAG-computation}

In this section, we describe \algref{extremal-DAG},
which computes the extremal event DAG from a collection of time series over a closed
interval. This algorithm is based on two
key insights. First, the edge weight between two local
extrema from the same function is the minimum of the node lives of the two local
extrema (\stmtref{edge-weights-same} of \thmref{edge-weights}). Second, the edge
weight between two local extrema from different
functions is the minimum of the two node lives and the infimum $\varepsilon$ for when the
two~$\varepsilon$-extremal intervals intersect (\stmtref{edge-weights-dif} of
\thmref{edge-weights}). We first describe the
computation of merge trees that are used to compute the node lives of local extrema.

\subsubsection{Merge Trees and Node Lives.}

The merge tree captures the connectivity of sublevel sets of a function. The
information we get from the merge tree of a function is very similar to
information we capture from the zeroth-dimensional persistence diagram from a
sublevel set filtration of $f$.

\begin{defn}[Merge Tree]
    Let $f$ be a real-valued function. Let $\Gamma(f)$ be the graph of $f$. We declare $x\sim y$
    if there exists an $h\in \R$ for which $x,y \in f^{-1}(h)$ and $x,y$ are in the same connected
    component of $f^{-1}(-\infty, h]$. The \emph{merge tree of $f$}, denoted $M_f$, is defined
    to be the quotient space
    $$M_f := \Gamma(f)/\sim.$$
\end{defn}

Given a nicely tame function, $f:C\rightarrow \R$, we construct the structure of the merge
tree, $G=(V, E)$, 
by~\cite{SmirnovTriplet20}. This structure consists of a list of
\emph{merge triplets} where each triplet is a three-tuple of real numbers
$(u,s,v)$ such that $u$ represents the connected component containing itself for $a \in [f(u), f(s))$ and $v$ becomes the representative of the connected component at a height of $f(s)$.

In relation to zeroth-dimensional persistence diagram from a sublevel set filtration,~$(f(u),f(s))$,
is a birth-death pair and the two connected components represented by $u$ and $v$ merge into
the connected component represented by $v$ at a height of $f(s)$.  If the merge triplet has
identical components denoted as $(u,u,u)$, then $u$ is the global minimum of its connected
component in $G$. The time complexity of computing the Merge Tree using Kruskal's algorithm for a function
represented as a graph, $G=(V,E)$ is $O(m\log n)$ where
$n=|V|$ and $m = |E|$ ~\cite{SmirnovTriplet20}. In our setting, the number of edges is bounded by $n$ and so the time complexity is $O(n\log n)$.

We describe how we compute the node lives and node labels of the extremal event DAG,
using the triplets computed from the merge trees.  \algref{minlives} (\textsc{GetMinLives})
takes input a merge tree $M$
for a time series $D=\{(z_i, h_i)\}_{i=1}^N$. \algref{minlives} outputs the
node lives of the local minima of  $D$. For each merge triplet with distinct components, $(z_i,z_j,z_k)$
in $M$, $(h_i, h_j)$ is a point in the zeroth-dimensional persistence diagram from the sublevel
set filtration. By \defref{nodelife}, we get that $\frac{1}{2}\pers_D(z_i)=(h_j-h_i)/2$.
If the merge triplet has all identical components, $(z_j,z_j,z_j)$, then
$\frac{1}{2}\pers_D(z_j)=\frac{1}{2}(\max(\{h_i\}_{i=1}^N)-\min(\{h_i\}_{i=1}^N))$.
Applying one of these two computations to all merge triplets computes node lives
for all local minima in $D$. To compute the node lives of the local maxima, we
apply the same process to $-D$, where we take the negative of all heights $h_i$.

\begin{algorithm}
    \caption{\textsc{GetMinLives}($D$, $M$)}\label{alg:minlives}

        \begin{algorithmic}[1]
            \require Array of merge tree triplets $M$ and time series $D=\{(z_i, h_i)\}_{i=1}^N$.
            \ensure Dictionary of node lives for each point in curve.

        \State minlives $\gets$ Initialize dictionary keyed by locations of extrema

	\For{$(z_i,z_j,z_k) \in M$}
		\If{$z_j = z_k$}
			\State minlives($z_i$) $\gets (\max(\{h_i\}_{i=1}^N)-\min(\{h_i\}_{i=1}^N))/2$
		\Else
			\State minlives($z_i$) $\gets |h_j-h_i|/2$
		\EndIf
	\EndFor\\
         \Return minlives
      \end{algorithmic}
\end{algorithm}

The algorithm \textsc{GetNodeLives}$(D, M_{\min}, M_{\max})$ applies \textsc{GetMinLives}
twice for both $D$ and $-D$ with merge trees $M_{\min}$ and $M_{\max}$ respectively.
Suppose $D$ has $n$ local extrema. Each line in~\algref{minlives} takes constant time. Since
the loop has at most $n$ iterations, then the total time complexity of~\algref{minlives}
is $O(n)$. Hence, the time complexity of \textsc{GetNodeLives} is also $O(n)$.

\subsubsection{Computing Edge Weights.}
We explain how to compute the edge weights of the extremal event DAG. First consider
the edge weight between two nodes in $B(f)$, a single backbone.
By \stmtref{edge-weights-same} of \thmref{edge-weights}, we know the edge weight
is the minimum node life between the two extrema.  Computing the minimum between
two values takes constant time.

Next, we describe how we compute the edge weight between two nodes from
different backbones. In order to do this, we must first compute the infimum
$\varepsilon$ for which two $\varepsilon$-extremal intervals intersect. In the
discrete setting, the growth of the $\varepsilon$-extremal
intervals only change at a finite number of $\varepsilon$. We refer to the
$\varepsilon$ values where discontinuous changes in length occur as
\textit{jumps.} \algref{get-eps-jumps-right} (\textsc{GetEpsJumpsRight})
computes the $\varepsilon$ jumps for the right endpoint of an
$\varepsilon$-extremal interval. We recall in the discrete setting, connected
components are determined by the linear interpolation of points in the image of
$f$. In \lemref{correctness}, we use $\sim$ to denote the equivalence relation
given by connected components i.e., for a time series $D_i=\{(z_j,
h_j^i)\}_{j=1}^N$, and $z_j, z_k \in Z$, we declare $z_j \sim z_k$ at
$\varepsilon>0$ if both $z_j$, $z_k$ are contained in the same connected
component of $(f_i-\varepsilon)^{-1}(-\infty, f_i(x)+\varepsilon)$ where $f_{i}$
is the linear interpolation of $D_i$ and $x \in Z$.

\begin{algorithm}
    \caption{\textsc{GetEpsJumpsRight}($D$, $z_i$)}\label{alg:get-eps-jumps-right}

    \begin{algorithmic}[1]
        \require Time series $D=\{(z_i, h_i)\}_{i=1}^N$ and domain point $z_i \in Z$.
        \ensure Vector of real numbers indicating values for which the right endpoint of $\d_\varepsilon(z_i)$ jumps.

        \State Initialize epsilon array and jump height
        \If{$i \neq N$}
        \State $\text{epsilons} \gets \langle|h_{i+1}-h_i|/2\rangle$ \label{line:initialeps}
        \State $\connectedheight  \gets h_{i+1}$

        \For{$j \gets$ $(i+2) \dots N$} \label{line:jumps-right-for-start}
                    \If{$\left(\mbox{extremum}(z_i) = \min \; \wedge \; h_j\geq \connectedheight \right)$\\
                    \quad \quad \quad $\bigvee \left(\mbox{extremum}(z_i) = \max \; \wedge \; h_j\leq \connectedheight \right)$}
                            \State epsilons \textbf{append} $|h_j-h_i|/2$
                            \State $\connectedheight  \gets h_j$
                    \EndIf
        \EndFor \label{line:jumps-right-for-end}
        \EndIf\\
        \Return epsilons
    \end{algorithmic}
\end{algorithm}

\begin{lem}[Correctness of~\algref{get-eps-jumps-right}]
    Let $D=\{(z_i, h_i)\}_{i=1}^N$ be a time series. Let $z_i$ be the domain
    coordinate of a local extremum of $D$. Let
    $r_i(\varepsilon):\R_{>0}\rightarrow \R$ denote the right endpoint of
    $\d_{\varepsilon}(z_i)$.
    Then, $\textsc{GetEpsJumpsRight}(D, z_i)$ ~\algref{get-eps-jumps-right}
    returns all $\varepsilon$ values for which $r_i$ has a jump discontinuity.
\label{lem:correctness}
\end{lem}

\begin{proof}
    We note that if $i=N$, then $r_{i}(\varepsilon)=z_i$ for all $\varepsilon\geq
    0$. Hence, there are no jumps. \textsc{GetEpsJumpsRight} returns the empty
    array, which is correct. For the rest of this proof, we assume $i\neq N$.
    To prove correctness, we show that we have the following loop invariant. At the start of each
    iteration of the for loop, the array epsilons consists of jumps of $r_i(\varepsilon)$ in sorted
    order. We first remark that jump discontinuities of $r_i(\varepsilon)$ occur at the infimum $\varepsilon$ for
     which a point $z_j \in Z$ is contained in $\d_\varepsilon(z_i)$  with $j > i$. This is
    because in the discrete setting, the $\varepsilon$-extremal intervals only grow when a new point in
    $Z$ is contained in $\d_\varepsilon(z_i)$. In particular, we show the loop
    invariant that $z_{i+j} \in \d_{\varepsilon'}(z_i)$, where $\varepsilon'$ is
    the maximum of the array named epsilons after the $j^{th}$ iteration.

    \begin{description}
        \item[Initialization]  First, consider $\varepsilon_1 = |h_i-h_{i+1}|/2$. By
            \lemref{discrete-extremal-containment}, $z_{i+1} \notin
            \d_{\varepsilon_1}(z_i)$ and for all $\varepsilon > \varepsilon_1$,
            $z_{i+1} \in \d_{\varepsilon}(z_i)$. Hence, $\varepsilon_1$ is the
            infimum $\varepsilon$ for which $z_{i+1}\in \d_{\varepsilon}(z_i)$.
            This implies
            that $\varepsilon_1$ is a jump discontinuity of $r_i$. From \lemref{discrete-monotonicity}, we know that
            $\d_\varepsilon(z_i)$ increases monotonically. Hence no other point in $Z$ is contained in
            $\d_\varepsilon(z_i)$ at a smaller value of $\varepsilon$. Therefore, $\varepsilon_1$ is
            the smallest jump discontinuity of $r_i(\varepsilon)$ and so the loop invariant holds before the
            first iteration.

        \item[Maintenance] Assume the loop invariant holds after the $j^{th}$ iteration. We show it
            also holds after the $j+1^{st}$ iteration. First assume $z_i$ is a local minimum. Then,
            \[ \connectedheight = \max\{h_k \mid k\in [i+1, i+j]\}. \]
            Denote $z_*:=z_{i+(j+1)}$.
            We want to find the infimum $\varepsilon>0$ for which $z_*\in \d_{\varepsilon}(z_i)$.

            Suppose $h_*< \connectedheight$. We claim that $z_*\in \d_{\varepsilon'}(z_i)$ where
            $\varepsilon':=\max\{\text{epsilons}\}$. This means that no new $\varepsilon$ value needs
            to be added to the epsilons vector in \algref{get-eps-jumps}. Let
            \[z':=\text{argmax}\{h_k \mid k\in [i+1,i+j]\}.\]
            Thus $h'=\connectedheight$. Since $h_*< \connectedheight=h'$ and
            $\varepsilon'=(h'-h_i)/2$, then
            \begin{equation*}
            \begin{split}
                h_*-h_i&<h'-h_i=2\varepsilon'\\
                h_*-\varepsilon'&<h_i+\varepsilon'
            \end{split}
            \end{equation*}
            This implies that $z_*\in (f-\varepsilon')^{-1}(-\infty,
            f(z_i)+\varepsilon')$ (recall $f$ is the linear interpolation of
            $D$). Observe, $z_{i+j} \in \d_{\varepsilon'}(z_i)$ by the
            assumption that the loop invariant holds at the $j^{th}$ iteration.
            Since $z_*$ is adjacent to $z_{i+j}$, $z_*$ must be in the same
            connected component of $(f-\varepsilon')^{-1}(-\infty,
            f(z_i)+\varepsilon')$ as $z_i$, i.e., $z_* \sim z_i$. Therefore,
            $z_* \in \d_{\varepsilon'}(z_i)$ and the loop invariant holds.

            Next suppose $h_*\geq\connectedheight$. Let $\varepsilon'$ be as
            before.  Applying a similar
            computation as above we find that $z_*\notin
            (f-\varepsilon')^{-1}(-\infty, f(z_i)+\varepsilon')$.
            Observe that at $\varepsilon^* := (h_*-h_i)/2$ we have
            $$f(z_{*})-\varepsilon^* = f(z_i)+\varepsilon^*.$$

            This leads to the observation that $z_* \in
            (f-\varepsilon)^{-1}(-\infty, h_i+\varepsilon)$ for any $\varepsilon
            > \varepsilon^*$ by \lemref{discrete-monotonicity}. Since
            $\varepsilon^* > \varepsilon'$, $z_{i+j} \in
            (f-\varepsilon)^{-1}(-\infty, f(z_i)+\varepsilon)$ as well. Since
            $z_i \sim z_{i+j}$ by the loop invariant and $z_{i+j} \sim z_*$ by
            adjacency, $z_i \sim z_*$ as desired for any $\varepsilon >
            \varepsilon^*$.
            These two observations tell us that $\varepsilon^*$ is the infimum
            $\varepsilon$ for which $z_{*}\in \d_{\varepsilon}(z_i)$.
            Therefore, $\varepsilon^*$ should indeed be added to the epsilons
            array and is larger than all other values in the array. By
            assumption, the epsilons array is sorted. Hence, the loop invariant
            holds.

            In the case that $(z_i, h_i)$ is a local maximum, we apply a
            symmetric argument by noting that $\d_{\varepsilon}(z_i)$ is the
            (expanded out) connected component of
            $(f+\varepsilon)^{-1}(f(z_i)-\varepsilon, \infty)$ containing $z_i$
            and a new value is added to the $\varepsilon$ vector if $h_j \leq
            \connectedheight$.

        \item[End] Note that the for loop terminates since there are only a finite number of iterations.

    \end{description}

Since the number of jumps of $r_i(\varepsilon)$ is bounded by $N-i$, and epsilons consists of
all infimum $\varepsilon$ for which a point in $N$ and greater than $z_i$ is contained in
$\d_\varepsilon(z_i)$, then~\algref{get-eps-jumps-right} is correct.
\end{proof}

Next, we analyze the time complexity of~\algref{get-eps-jumps-right}. Every line
takes constant time. Since the for loop
(\lineref{jumps-right-for-start}-\lineref{jumps-right-for-end}) has at most
$N-1$ iterations, then the total time complexity of~\algref{get-eps-jumps-right}
is $O(N)$.

Furthermore, we can apply the same algorithm but with going through points on
the left  of $z_i$ to find all the points for which the left endpoint of
$\d_{\varepsilon}(z_i)$ changes. Note that there are $N-1$ points to the left
and right of $z_i$ combined, and so finding all $\varepsilon$-jumps of
$\d_{\varepsilon}(z_i)$ takes $O(N)$. We call this combined function,
\textsc{GetEpsJumps} (\algref{get-eps-jumps}).

\begin{algorithm}
    \caption{\textsc{GetEpsJumps}($D$, $z_i$)}\label{alg:get-eps-jumps}

    \begin{algorithmic}[1]
        \require Time series $D=\{(z_i, h_i)\}_{i=1}^N$ and domain point $z_i \in Z$.
        \ensure Vector of real numbers indicating values for which the left or right endpoint of $\d_\varepsilon(z_i)$ jumps.

        \State Initialize epsilon array and jump height
        \State $\text{epsilons} \gets \{\textsc{GetEpsJumpsLeft}(D, z_i)\}$ \label{line:eps-right}
        \State $\text{epsilons} \textbf{ append } \textsc{GetEpsJumpsRight}(D, z_i)$ \\ \label{line:eps-left}
        \Return epsilons
    \end{algorithmic}
\end{algorithm}

Lastly, to find the infimum $\varepsilon$ for which two $\varepsilon$-extremal
intervals intersect, we apply \algref{get-eps-intersection}.  \algref{get-eps-intersection} takes input of two
time series $D_j, D_k$, two domain coordinates of local extrema of $D_j, D_k$,
and merge trees of  $D_j, -D_j, D_k, -D_k$. \algref{get-eps-jumps}
(\textsc{GetEpsJumps}) is applied to both functions and extrema. Then,
\algref{get-eps-intersection} goes through all jumps in order to find the
smallest one for which the two $\varepsilon$-extremal intervals intersect.

\begin{algorithm}
    \caption{\textsc{GetEpsIntersection}($D_j$, $D_k$, $z_j$, $z_k$, $M_{D_j}$, $M_{-D_{j}}$, $M_{D_k}$, $M_{-D_{k}}$)}\label{alg:get-eps-intersection}

        \begin{algorithmic}[1]
            \require  Time series $D_j=\{(z_i, h_i^j)\}_{i=1}^N$, $D_k=\{(z_i, h_i^k)\}_{i=1}^N$ and domain points of extrema in $D_j$ and $D_k$, denoted $z_j$, $z_k \in Z$.
            		$M_{D_j}$, $M_{-D_{j}}$, $M_{D_k}$, $M_{-D_{k}}$ are merge trees of $D_j, -D_j, D_k, -D_k$ respectively.
            \ensure $\inf_{\varepsilon} \d_{\varepsilon}^{D_j}(z_j) \cap \d_{\varepsilon}^{D_k}(z_k) \neq \emptyset$

	\State Initialize epsilon array
	\State $\text{epsilons} \gets \{\textsc{GetEpsJumps}(D_j, z_j)\}$ \label{line:eps-jumps-f}
	\State $\text{epsilons} \textbf{ append } \textsc{GetEpsJumps}(D_k, z_k)$  \label{line:eps-jumps-g}
	\State $\text{sort epsilons}$ \label{line:eps-sort}
	\For{$\varepsilon \in \text{epsilons}$} \label{line:eps-int-for-start}
		\State $\d_{\varepsilon}^{D_j}(z_j) \gets \textsc{GetExtremalInterval}(D_j, z_j, M_{D_j}, M_{-D_{j}}, \varepsilon)$ \label{line:get-extremal-interval-f}
		\State $\d_{\varepsilon}^{D_k}(z_k) \gets \textsc{GetExtremalInterval}(D_k, z_k, M_{D_k}, M_{-D_{k}}, \varepsilon)$ \label{line:get-extremal-interval-g}
		\If{$\d_{\varepsilon}^{D_j}(z_j) \cap \d_{\varepsilon}^{D_k}(z_k) \neq \emptyset$}\\
			\quad\quad \Return $\varepsilon$
		\EndIf
	\EndFor \label{line:eps-int-for-end}
        \end{algorithmic}
\end{algorithm}

Next, we discuss the time complexity of \algref{get-eps-intersection}. The
number of jumps is bounded by the number of points in the domains of the two
functions. Additionally, at each jump, we compute \textsc{GetExtremalInterval}
for the two extrema, where the computation of the discrete extremal intervals
are discussed in \cite{BerryUsing20} and implemented in
\cite{BeltonComputing21}.
In particular, the time complexity for computing
\textsc{GetExtremalInterval}$(D, z_i, M_D, M_{-D})$ is $O(N)$. This is because
computing the $\varepsilon$-extremal interval requires evaluating $f$ at points
in $Z$ near $z_i$. In summary

\begin{itemize}
    \item \lineref{eps-jumps-f} and \lineref{eps-jumps-g} each take $O(N)$.
    \item \lineref{eps-sort} takes $O(N\log N)$.
    \item \lineref{get-extremal-interval-f} and \lineref{get-extremal-interval-g} take $O(N)$.
    \item  All other lines take constant time.
    \item The number of iterations of the for loop in
        \lineref{eps-int-for-start}-\lineref{eps-int-for-end} is bounded by
        $2N$.
\end{itemize}

All together we compute the time complexity of \textsc{GetEpsIntersection}$(D_j,
D_k,$ $z_j, z_k, M_{D_j}, M_{-D_j}, M_{D_k}, M_{-D_k})$ as
\[ O(2N + N\log N + 2N^2) = O(N \log N + N^2). \]

\subsubsection{Computing Extremal Event DAG}

\algref{extremal-DAG} computes the extremal event DAG from a
collection of time series $\mathcal{D}=\{D_j\}_{j=1}^K$. \algref{extremal-DAG}
uses previously defined algorithms and functions from this section along with
\textsc{InitializeGraph}.  This algorithm is designed and implemented in
\cite{BeltonComputing21}.
\textsc{InitializeGraph} takes a collection of time series $\mathcal{D}$ as
input and outputs $(T, H, V, E)$ where $V, E$ are vertices and directed edges of
the extremal event DAG of $\mathcal{D}$, $T$ is the domain coordinates of the local extrema, and $H$
is the heights of local extrema. This function checks through all points in $Z$
for extrema to record as vertices and then goes through all vertex pairs to
check for edges. Let $N = |Z|$. The number of vertices is bounded by $NK$ and
the number of edges is bounded by ${NK}\choose{2}$. Hence, the time complexity
of \textsc{InitializeGraph} is
\[ O(NK+\frac{(NK)(NK-1)}{2}) = O((NK)^2). \]

\begin{algorithm}
    \caption{\textsc{GetExtremalEventDAG}($\mathcal{D}$)}\label{alg:extremal-DAG}

        \begin{algorithmic}[1]
            \require A collection of time series $\mathcal{D} =\{D_j\}_{j=1}^K$.
            \ensure The extremal event DAG of $\mathcal{D}$.

	    \State $M_{D_j} \gets \text{merge tree for }D_j$ \label{line:mergetreemin}
	    \State $M_{-D_j} \gets \text{merge tree for }-D_j$ \label{line:mergetreemax}
            \State $\text{NodeLives}_j \gets \textsc{GetNodeLives}(D_j, M_{D_j}, M_{-D_j})$ \label{line:nodelivesi}

            \State $(T, V, E) \gets \textsc{InitializeGraph}(\mathcal{D})$
                \quad Initialize unweighted extremal event DAG. \label{line:initializegraph}

            \State Initialize function $\omega_{V} \colon V \to \R$ with all
            values set to zero. \label{line:initializeverts}

            \State Initialize function $\omega_{E} \colon E \to \R$ with all
            values set to zero.

            \For{$v \in V$} \label{line:forloop-nodes}
            	\If{$(T(v), H(v)) \in D_j$}
            	\State $\omega_{V}(v) \gets \text{NodeLives}_j(v)$
		\EndIf
	\EndFor\\ \label{line:endforloop-nodes}

           \For{$e=(u,v) \in E$} \label{line:forloop-edges}
                    \If{$(T(u), H(u)), (T(v), H(v))\in D_j$}
			\State $\omega_{E}(e) \gets \min(\text{NodeLives}_j(u), \text{NodeLives}_j(v))$
                    \ElsIf{$(T(u), H(u))\in D_j, (T(v), H(v))\in D_k$}

			\State $ \varepsilon \gets \textsc{GetEpsIntersection}(D_j, D_k, T(u), T(v), M_{D_j}, M_{-D_j}, M_{D_k}, M_{-D_k})$
				\State $\omega_{E}(e) \gets \min(\text{NodeLives}_j(u), \text{NodeLives}_k(v), \varepsilon)$

                    \EndIf
            \EndFor\\ \label{line:endforloop-edges}

            \Return $(V, E, \omega_{V}, \omega_{E})$

        \end{algorithmic}
\end{algorithm}

The correctness of~\algref{extremal-DAG} follows from the correctness of all our
other previously defined
algorithms. Next we analyze the time complexity.
\begin{itemize}
    \item \lineref{mergetreemin} and \lineref{mergetreemax} each take $KO(N\log N)$.
    \item  \lineref{nodelivesi} takes $KO(N)$.
    \item  Initializing the extremal event DAG in~\lineref{initializegraph} takes $O((NK)^2)$.
    \item Computing all the node weights in
        \lineref{forloop-nodes}-\lineref{endforloop-nodes} has a time complexity
        of $O(NK)$.
    \item Each iteration in the for loop between \lineref{forloop-edges} -
        \lineref{endforloop-edges} takes at most $O(N\log N+N^2)$.  Since the
        number of vertices is bounded above by $NK$, then the number of edges is
        bounded above by \mbox{${NK}\choose{2}$ $= \frac{(NK)(NK-1)}{2}$}. Thus,
        the number of iterations of this for loop is bounded by
        $\frac{(NK)(NK-1)}{2}$.
\end{itemize}

 In total, we get the time complexity to be
 \begin{align*}
&O\left(K(N\log N)+KN+(NK)^2+NK+\frac{(NK)(NK-1)}{2}\left((N\log N)+N^2\right)\right)\\
&= O(N^4K^2).
\end{align*}

\subsection{Computing Optimal Backbone Alignments}
To compute a distance between extremal event DAGs, we align backbones in the extremal event DAGs in
an optimal manner. Here, we describe how the alignment is computed and prove
that the alignment is optimal. Recall, that $\alpha$ denotes an alignment
between two backbones (\defref{alignment}).

\begin{defn}[Alignment Matrix]
    Let $\x =(x_1, x_2, \dots, x_m)$, $\y =(y_1,y_2, \dots, y_n)$ be backbones.
    Note that each $x_i$ can be written as the pair $x_i=(s_{\x,i},w_{\x,i})$;
    likewise each $y_i$ can be written as the pair~$y_i=(s_{\y,i},w_{\y,i})$.
    The \emph{alignment matrix}, denoted $\mat$, is an $(m+1) \times (n+1)$
    matrix recursively defined as follows:

    $$
        \mat[i,j]=
        \begin{cases}
            0 & i=j=1 \\
            \sum_{k=1}^{i-1} w_{\x,k}, & i>1,j=1 \\
            \sum_{k=1}^{j-1} w_{\y,k}, & i=1,j>1 \\
            \min \left\{
                \begin{array}{lr}
                    \mat[i-1, j]+w_{\x,i-1} \\
                    \mat[i, j-1]+w_{\y,j-1} \\
                    \mat[i-1,j-1]+\diff(x_{i-1}, y_{j-1})
                \end{array}
            \right\}                             & \text{otherwise},
        \end{cases}
    $$
    where $\diff: \x \times \y \rightarrow \R_{\geq 0} \cup \{\infty\}$
    is defined by
    $$\diff\left((s_x,w_x), (s_y,w_y)\right) =
        \begin{cases}
            | w_x - w_y |, & \text{if } s_x = s_y \\
            \infty, & \text{otherwise}\\
        \end{cases}
        .
    $$

\label{def:alignment-mat}
\end{defn}

Next we note that the bottom right entry in the alignment matrix is the minimum
cost of aligning two backbones $\x$ and $\y$. This follows from
\cite{WagnerThe74}. Recall the definition of the cost function in \defref{cost}. We also prove it here.

\begin{prop}[Alignment Matrix Finds Minimum Cost]\label{propref:alignment-matrix}
    Let $\x = (x_1, x_2, \dots, x_m)$ and $\y = (y_1, y_2, \dots, y_n)$ be
    backbones. Let $\mat$ be the $(m+1) \times (n+1)$ alignment matrix. Then, $\mat[m+1, n+1] = c_{\x, \y}(m, n)$.
\end{prop}

\begin{proof}
    For $i \in [n]$, let $x_i=(s_{\x, i},w_{\x, i})$ and $y_i=(s_{\y, i},w_{\y, i})$.

    We proceed by induction. For the base case, first observe that
    $c_{\x,\y}(1,0) = w_{\x,1}$ and $c_{\x,\y}(0,1) = w_{\y, 1}$. By
    construction, $\mat[2,1] = w_{\x, 1}=c_{\x,\y}(1,0)$ and \mbox{$\mat[1,2] =
    w_{\y, 1}=c_{\x,\y}(0,1)$}. Next consider $c_{\x,\y}(1,1)$. The possible
    alignments (see \defref{alignment}) of $\x_1:=\x[1:1]$ and $\y_1:=\y[1:1]$ are
    \begin{enumerate}
        \item $\alpha_1: \{1, 2\} \rightarrow \tilde{\x}_1 \times \tilde{\y}_1$,
            where  $\alpha_1(1) = (x_1, \zero)$ and $\alpha_1(2) = (\zero, y_1)$.
        \item $\alpha_2: \{1, 2\} \rightarrow \tilde{\x}_1 \times \tilde{\y}_1$,
            where  $\alpha_2(1) = (\zero, y_1)$ and $\alpha_2(2) = (x_1, \zero)$.
        \item $\alpha_3: \{1\} \rightarrow \tilde{\x}_1 \times \tilde{\y}_1$,
            where $\alpha_3(1) = (x_1, y_1)$ is a possible alignment.
    \end{enumerate}
    Observe
    $\cost(\alpha_1) = \cost(\alpha_2) = w_{\x, 1}+w_{\y, 1}$ and
    $\cost(\alpha_3) =\diff(x_1,y_1)= |w_{\x, 1}-w_{\y, 1}|$.
    Therefore,
    $$
       c_{\x,\y}(1,1) = \min \left\{
         \begin{array}{lr}
             w_{\x, 1}+w_{\y, 1} \\
           \diff( x_1, y_1)
         \end{array}
       \right.
    $$

    \noindent By construction,
    $$
       \mat[2,2] = \min \left\{
        \begin{array}{lr}
            \mat[2,1]+w_{\y, 1} \\
            \mat[1,2]+w_{\x, 1} \\
          \mat[1,1]+\diff(x_1, y_1)
         \end{array}
        \right.
    $$

    \noindent Substituting $w_{\x, 1}$ for $\mat[2,1]$ , $w_{\y, 1}$
    for $\mat[1,2]$ and zero for $\mat[1,1]$, we find $c_{\x,\y}(1,1) =
    \mat[2,2]$. This shows the base case holds.

    For the induction hypothesis we assume that $\mat[h, k] =
    c_{\x,\y}(h-1,k-1)$ for all $h \leq i$ and $k \leq j$ where $i \leq m$ and
    $j\leq n$.

    In the induction step, we show $\mat[i+1,j] = c_{\x,\y}(i, j-1)$, $\mat[i,
    j+1]=c_{\x, \y}(i-1, j)$, and \mbox{$\mat[i+1, j+1] = c_{\x, \y}(i,j)$}. First
    consider $c_{\x,\y}(i, j-1)$. Let $\alpha: [k]\rightarrow \tilde{\x}[1:i] \times
    \tilde{\y}[1:j-1]$ be an alignment
    of the first $i$ nodes of $\x$ with the first $j-1$ nodes of $\y$ with cost
    $c_{\x, \y}(i,j-1)$. Consider the last pair of nodes aligned via
    $\alpha(k)$. The cost of these two nodes is either (a) the cost of $x_i$
    aligned with an insertion, (b) the cost of $y_{j-1}$ aligned  with an
    insertion, or (c) the cost of $x_i$ aligned with $y_{j-1}$. Note, by
    \defref{alignment}, we never have an insertion aligned with an insertion.
    Since the cost is the minimum across these three possibilities, the cost is
    $$
       c_{\x, \y}(i, j-1)  = \min \left\{
        \begin{array}{lr}
            c_{\x, \y}(i-1, j-1)+ w_{\x, i} \\
            c_{\x, \y}(i, j-2)+w_{\y, j-1} \\
          c_{\x, \y}(i-1, j-2) + \diff(x_i, y_{j-1})\\
         \end{array}
        \right.
    $$

    \noindent Applying the induction hypothesis, we find

    $$
       c_{\x, \y}(i, j-1)  = \min \left\{
        \begin{array}{lr}
            \mat[i, j]+ w_{\x, i} \\
            \mat[i+1, j-1] +w_{\y, j-1} \\
          \mat[i, j-1]+ \diff(x_i, y_{j-1})\\
         \end{array}
        \right.
    $$

    \noindent By construction of $\mat$ (\defref{alignment-mat}), we see that
    $c_{\x, \y}(i, j-1) = \mat[i+1, j]$. Using a similar approach, we find
    $\mat[i, j+1]=c_{\x, \y}(i-1, j)$, and $\mat[i+1, j+1] = c_{\x, \y}(i,j)$.

    This concludes the induction argument. Thus, $\mat[i, j] = c_{\x, \y}(i-1,
    j-1)$ for all $i\leq m+1$ and $j\leq n+1$. In particular, $\mat[m+1, n+1] =
    c_{\x, \y}(m, n)$.
\end{proof}

\subsubsection{Finding Optimal Alignment from Alignment Matrix}

\begin{defn}[Path]
    Let $M$ be an $m \times n$ matrix with real valued entries. A \emph{path in
    $M$} is an injective function $p: [k] \rightarrow M$ such that $p(i)$ and
    $p(i+1)$ are adjacent values in a row, column, or diagonal for all $i \in
    \{1, 2, \dots, k-1\}$.
\end{defn}

\noindent To find an optimal alignment from the alignment matrix we construct a path via \textit{backtracking}.

\vspace{1ex}

\noindent \textbf{Path via Backtracking}.
Let $\x =(x_1, x_2, \dots, x_m)$ and $\y = (y_1, y_2, \dots, y_n)$ be backbones.
Let $\mat$ be the corresponding alignment matrix. We construct a path $p$ in
$\mat$ recursively as follows:
\begin{itemize}
    \item $p(1) = \mat[m+1, n+1]$
    \item If $p(h) = \mat[i, j]$ for  $h \geq 1$ and $i, j >1$, then,
        $$
           p(h+1) = \left\{
            \begin{array}{lr}
                \mat[i-1, j] \text{ if } \mat[i,j] = \mat[i-1, j] + w_{\x, i-1} \\
                \mat[i, j-1] \text{ if } \mat[i,j] = \mat[i, j-1]+ w_{\y, j-1} \\
              \mat[i-1, j-1] \text{ if } \mat[i,j] = \mat[i-1,j-1]+\diff(x_{i-1}, y_{j-1})
             \end{array}
            \right.
        $$
        If multiple of the conditions hold, then define $p(h+1)$ to be any
        \textit{one} of them.  We call $p$ a \emph{backtracking} path.
\end{itemize}

In summary, we are undoing the matrix construction to figure out which matrix
entries lead to the cost in $\mat[m+1, n+1]$. Once we apply backtracking, we
have at least one path from $\mat[m+1, n+1]$ to $\mat[1,1]$. We remark that
backtracking is well-defined. For any entry $p(h) = \mat[i,j]$, one of the three
upper left entries ($\mat[i-1, j]$, $\mat[i, j-1]$, $\mat[i-1, j-1]$) equals
$p(h+1)$ by construction of the alignment matrix (\defref{alignment-mat}). Since
we have a finite matrix, we eventually end at $\mat[1,1]$. We note that a path
constructed from backtracking is not necessarily unique. For describing the
alignment from a backtracking path $p: [k] \rightarrow \mat$, we consider the
\emph{reverse path} $p': [k] \rightarrow \mat$ where $p'(i) = p(k-(i-1))$.

\begin{figure}[htp]
    \includegraphics[width=0.6\textwidth]{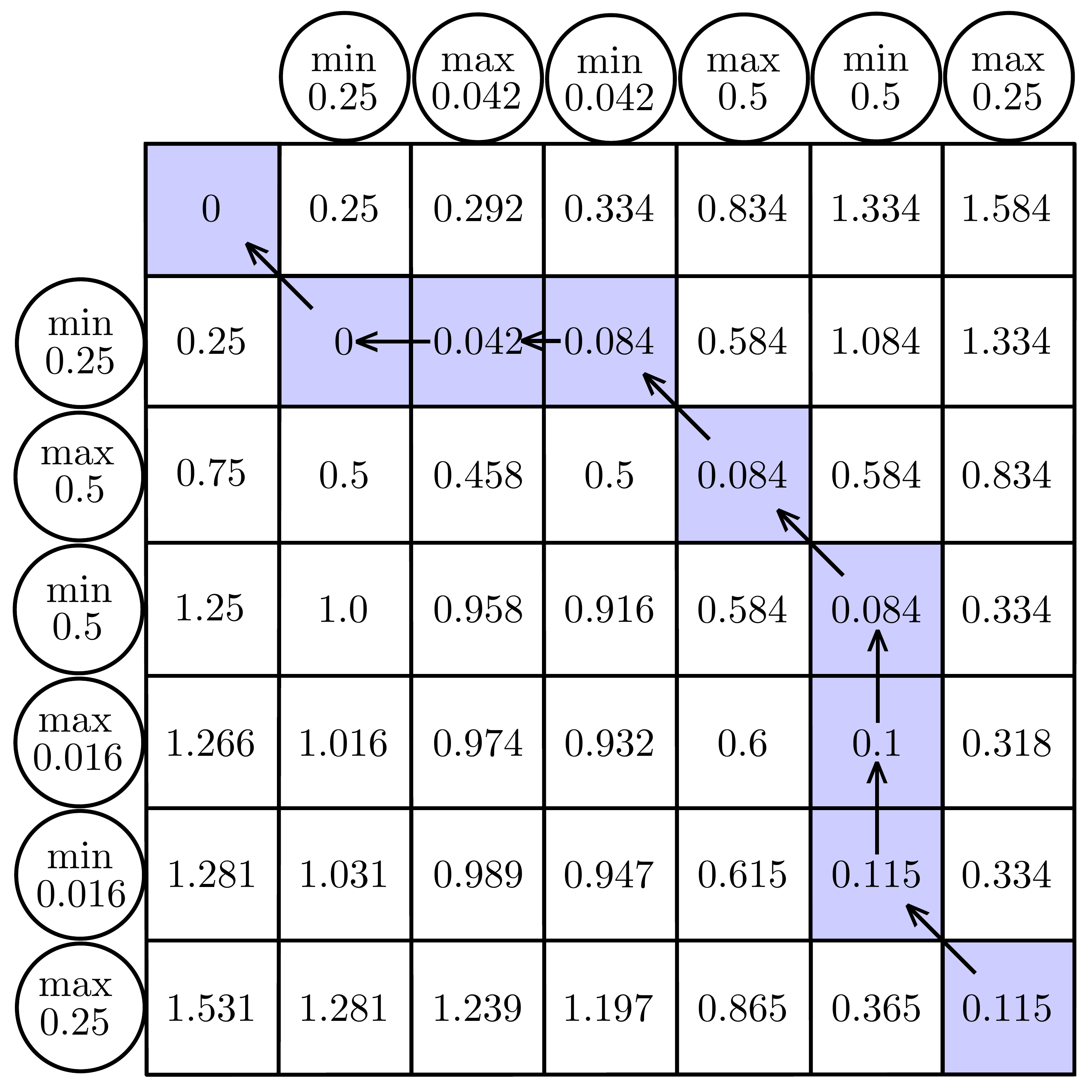}
    \caption{ Backtracking in alignment matrix of sine backbones. Consider the
    backbones~$\x$ and~$\y$ from \figref{backbones}. The path $p:\{1, 2,
    \ldots, 7\} \rightarrow \tilde{\x}\times\tilde{\y}$, constructed from backtracking is illustrated
    through the black arrows and purple highlighted entries. In particular,
    $p(1) = \mat[7,7] = 0.115$ and $p(9) = \mat[1,1] = 0$. }
    \label{fig:alignment-path}
\end{figure}

\vspace{1ex}

\noindent \textbf{Alignment from Backtracking}.  Let $\x =(x_1, x_2, \dots,
x_m)$ and $\y = (y_1, y_2, \dots, y_n)$ be backbones. Let $\mat$ be the
corresponding alignment matrix. Let $p: [k] \rightarrow \mat$ be a path computed
from backtracking and $p': [k] \rightarrow \mat$ be the reverse path. We
construct an alignment $\alpha: [k-1] \rightarrow
\tilde{\x}\times\tilde{\y}$  such that
$$
   \alpha(h) = \left\{
    \begin{array}{lr}
      (x_i, \zero) \quad \text{ if } p'(h) = \mat[i, j] \text{ and } p'(h+1) = \mat[i+1, j]\\
      (\zero, y_j) \quad \text{ if } p'(h) = \mat[i,j] \text{ and } p'(h+1) = \mat[i, j+1]\\
      (x_i, y_j) \quad \text{ if } p'(h) = \mat[i, j] \text{ and }p'(h+1) = \mat[i+1, j+1]
     \end{array}.
    \right.
$$

In other words, the following moves of $p'$ through the matrix $\mat$ mean:
\begin{itemize}
\item Vertical move from $\mat[i,j]$ to $\mat[i+1, j]$ indicates an alignment of $x_i$ with an insertion.
\item Horizontal move from $\mat[i,j]$ to $\mat[i, j+1]$ indicates an alignment of $y_j$ with an insertion.
\item Diagonal move from $\mat[i,j]$ to $\mat[i+1,j+1]$ indicates an alignment of $x_i$ with $y_j$.
\end{itemize}

Next we verify that an alignment found from backtracking is indeed an alignment.

\begin{prop}[Backtracking Finds an Alignment]
    Let $\x = (x_1, x_2, \dots, x_m)$ and $\y = (y_1, y_2, \dots, y_n)$ be
    backbones. Let $\mat$ be the $(m+1) \times (n+1)$ alignment matrix. Let
    $\alpha: [k]\rightarrow \tilde{\x}\times\tilde{\y}$ be an
    alignment found from backtracking. Then,
    $\alpha$ is an alignment.
\end{prop}

\begin{proof}

    We verify that $\alpha$ is well-defined,  and $\alpha$ satisfies all four properties of being an alignment (\defref{alignment}).
    Let $p': [k+1]\rightarrow \mat$ be the path used to construct $\alpha$.

    \orgemph{First we show $\alpha$ is well-defined.} Let $h \in [k-1]$ and consider
    $\alpha(h)$. Observe exactly one of the three conditions (vertical move, horizontal move, diagonal move) that define
    $\alpha(h)$ holds by the construction of $p'$, and $\alpha(h) \in \tilde{\x}\times\tilde{\y}$. Hence,
    $\alpha$ is well-defined.
    
    Observe by construction, $\alpha$ has no null alignments. Hence, \propertyref{nullalignments} of \defref{alignment} holds. Next we show the remaining properties.

    \orgemph{\propertyref{preservesbackbones} (Preserves Order of Backbones).}
     We use the index function $\iota_\x$ and the coordinate function $\alpha_\x$
    from \defref{alignment}. Suppose $\alpha_\x(i)$, $\alpha_\x(j) \in \x$ where $i <
    j$. Suppose further that $p'(i) = \mat[q, r]$ and $p'(j) = \mat[s, t]$,
    indicating that $\alpha_\x(i) = x_q$, $\alpha_\x(j) = x_s$. By construction of
    $p'$ and the fact that $i<j$ either (1) $q = s$ and $r < t$ or (2) $q < s$.
    We claim that only (2) holds in our setting.

    Assume for a contradiction that $q = s$ and $r < t$. By construction of
    $p'$, this means that there are only horizontal moves between $p'(i)$ and
    $p'(j)$. Hence, $\alpha_\x(i)$ or $\alpha_\x(j)$ or both must be the empty node.
    This contradicts the assumption that both $\alpha_\x(i)$ and $\alpha_\x(j)$ are
    in $\x$. Therefore $q< s$, and in particular, $q < s$ if and only
    if $i < j$.

    Since $\alpha_\x(i) = x_q$ and $\alpha_\x(j) = x_s$, we have the index function
    $\iota_\x(\alpha_\x(i)) = q$ and $\iota_\x(\alpha_\x(j)) = s$ if and only if
    $\iota_\x(\alpha_\x(i))  < \iota_\x(\alpha_\x(j))$. Therefore
    $\iota_\x(\alpha_\x(i))  < \iota_\x(\alpha_\x(j))$ if and only if $i < j$, so
    that $\alpha$ preserves the order of nodes in the backbone $\x$. The same
    argument substituting $\y$ for $\x$ also shows that $\alpha$ preserves the
    order of nodes in the backbone $\y$.
    
    \orgemph{\propertyref{nomisalignments} (No Misalignments).} By design of
    $\mat$, a misalignment has an infinite cost. Since each entry in the
    alignment matrix is a minimum of three values where at least two values are
    finite, then $\mat$ does not contain any infinite entries. This implies that
    when applying backtracking, we never have a diagonal move corresponding to a
    misalignment.
    
    \orgemph{\propertyref{restrictiontomatching} (Restriction to Matching).} 
     Let $x_i \in \x$. 
    Recall in the definition of $p'$, we construct a path using adjacent entries in the alignment matrix 
    such that $p'(0) = \mat[1, 1]$ and $p'(k)=\mat[1,1]$. By definition of $p'$, there exists a $j \in [k+1]$ such that 
    $p'(j) = \mat[i+1, h]$ where $h \in [n+1]$. This implies that $x_i$ appears in 
    $\im(\alpha_\x)$ at least once. Similarly, for $y_i \in \y$, there exists a
    $j \in [k+1]$ such that $p'(j) = \mat[h, i+1]$ where $h \in [m+1]$. Hence $y_i$ 
    appears in $\im(\alpha_\y)$ at least once. Furthermore, in backtracking, no same move (vertical, horizontal, or diagonal) between two matrix entries is repeated twice. 
    Hence, each $x_i \in \x$ and each $y_i \in \y$ appears in $\im(\alpha_\x)$ and $\im(\alpha_\y)$ exactly once. Therefore we have a restriction to matching.    

\end{proof}

We now prove that an alignment found using this backtracking has a cost equal to $\mat[m+1, n+1]$. 

\begin{prop}[Backtracking Finds Alignment with Cost Computed from Alignment Matrix]
    Let \mbox{$\x = (x_1, x_2, \dots, x_m)$} and $\y = (y_1, y_2, \dots, y_n)$ be
    backbones. Let $\mat$ be the $(m+1) \times (n+1)$ alignment matrix. Let
    $\alpha: [k]\rightarrow \tilde{\x}\times\tilde{\y}$ be an alignment found
    from backtracking. Then,
    $\cost(\alpha) = \mat[m+1, n+1]$.
\label{propref:backtracking-optimal}
\end{prop}

\begin{proof}
    Let $p':[k+1] \rightarrow \mat$ be the path used to construct $\alpha$. 
    We show $\cost(\alpha) = \mat[m+1, n+1]$. To do this, we prove
    $\cost(\alpha[1:h]) = p'(h+1)$ for all $h\leq k$ by induction. 
    For the base case, consider $\alpha[1:1] =\alpha(1)$. There are three
    possibilities for $\alpha(1)$. Either:
    \begin{enumerate}
        \item $\alpha(1) = (x_1, \zero)$
        \item $\alpha(1) = (\zero, y_1)$
        \item $\alpha(1) = (x_1, y_1)$.
    \end{enumerate}
    Recall that $x_i \in \x$ can be expanded as $x_i=(s_{\x, i},w_{\x,i})$;
    likewise, for $y_i \in \y$, we can write $y_i=(s_{\y, i},w_{\y, i})$.
    If~(1), then $\cost(\alpha(1)) = w_{\x,1} = \mat[2,1] = p'(2)$. If~(2),
    then $\cost(\alpha(1)) = w_{\y,1} = \mat[1, 2] = p'(2)$. If (3),
    then $\cost(\alpha(1)) = \diff(x_1, y_1) = \mat[2,2] = p'(2)$. In all three
    cases we find $\cost(\alpha(1)) = p'(2)$.

    Next, we assume the induction hypothesis that $\cost(\alpha[1:h]) = p'(h+1)$ for
    some $h < k$.

    Suppose $p'(h+1) = \mat[i,j]$.  There are three possibilities for
    $\alpha(h+1)$. Either
    \begin{enumerate}
        \item $\alpha(h+1) = (x_i, \zero)$
        \item $\alpha(h+1) = (\zero, y_j)$
        \item $\alpha(h+1) = (x_i, y_j)$.
    \end{enumerate}
    If (1), then
    $$
        \cost(\alpha[1:h+1])
        = \cost(\alpha[1:h]) + w_{\x,i} = p'(h+1) + w_{\x,i} = \mat[i, j] + w_{\x,i} = p'(h+2).
    $$
    If (2), then
    $$
        \cost(\alpha[1:h+1])
        = \cost(\alpha[1:h]) + w_{\y, j} = p'(h+1) + w_{\y, j} = \mat[i, j] + w_{\y, j} = p'(h+2).
    $$
    If (3), then
    \begin{equation*}
    \begin{split}
        \cost(\alpha[1:h+1])
        &= \cost(\alpha[1:h]) + \diff(x_i, y_j) = p'(h+1) + \diff(x_i, y_j)\\
         &= \mat[i, j] + \diff(x_i, y_j) = p'(h+2).
    \end{split}
    \end{equation*}
    All equalities follow from either the induction hypothesis or construction
    of~$p'$

    By induction, $\cost(\alpha[1:h]) = p'(h+1)$ for all $h\leq k$. In particular,
    \[ \cost(\alpha[1:k]) = \cost(\alpha) = p'(k+1) = \mat[m+1, n+1]. \]
\end{proof}

Observe that Proposition~\ref{propref:alignment-matrix} and
Proposition~\ref{propref:backtracking-optimal} give the following corollary.

\begin{cor}[Backtracking Finds Optimal Alignment]
    Let $\x = (x_1, x_2, \dots, x_m)$ and $\y = (y_1, y_2, \dots, y_n)$ be
    backbones. Let $\mat$ be the $(m+1) \times (n+1)$ alignment matrix. Let
    $\alpha: [k]\rightarrow \tilde{\x} \times \tilde{\y}$ be an alignment found from backtracking. Then
    $\cost(\alpha) = c_{\x, \y}(m, n)$.
\end{cor}

\subsubsection{Time Complexity of Computing Backbone and Extremal Event DAG Distance}
Using the dynamic program above, we can compute the backbone distance in $O(mn)$
time where $m$ and $n$ are the lengths of the two backbones. However, since
backbone alignments are not always unique, then computing all optimal backbone
alignments can become costly.

When computing the extremal event DAG distance, we must compute the optimal backbone
alignments that minimize the difference in weights over all possible aligned
edges. Since we could have multiple optimal backbone alignments, then computing
the extremal event DAG distance in the worst case is expensive. However, we have found
empirically for the applications below that almost always there is a unique optimal
alignment. This then results in a polynomial time complexity for computing the
extremal event DAG distance.

\end{document}